\documentclass[prx,twocolumn,floatfix,amsmath,amssymb,aps,superscriptaddress,longbibliography,notitlepage]{revtex4-1}

\usepackage[markup=default]{changes}
\usepackage{graphicx} % Example of a commonly used package
\usepackage{amsmath}
\usepackage{amsthm}
\usepackage{amssymb}
\usepackage{bm}
\usepackage{appendix}
\usepackage{float}
\usepackage{soul}

\let\oldmathbf\mathbf
\renewcommand{\mathbf}[1]{\ifcat\noexpand#1\relax\bm{#1}\else\oldmathbf{#1}\fi}

\usepackage{hyperref} % For clickable links in the table of contents
\usepackage{array}
\usepackage{booktabs}
\usepackage{dcolumn}% Align table columns on decimal point

\newcommand{\bra}[1]{\langle#1\rvert} % Bra
\newcommand{\ket}[1]{\lvert#1\rangle} % Ket
\usepackage{xcolor}

\newtheorem{theorem}{Theorem}% meant for continuous numbers
\newtheorem{corollary}{Corollary}[theorem]
\newtheorem{definition}{Definition}%

\begin{document}

\title{Non-Unitary Quantum Machine Learning}% Force line breaks with \\

\author{Jamie Heredge}
\email{jamie.heredge@student.unimelb.edu.au}
 \affiliation{School of Physics, University of Melbourne, Parkville, VIC School of Physics, Australia}%Lines break automatically or can be forced with \\

\author{Maxwell West}%
 \affiliation{School of Physics, University of Melbourne, Parkville, VIC School of Physics, Australia}
\author{Lloyd Hollenberg}
 \affiliation{School of Physics, University of Melbourne, Parkville, VIC School of Physics, Australia}
\author{Martin Sevior}
 \affiliation{School of Physics, University of Melbourne, Parkville, VIC School of Physics, Australia}
 
%\date{\today}% It is always \today, today,
             %  but any date may be explicitly specified

\begin{abstract}
We introduce several probabilistic quantum algorithms that overcome the normal unitary restrictions in quantum machine learning by leveraging the Linear Combination of Unitaries (LCU) method. Among our investigations are quantum native implementations of Residual Networks (ResNet), where we show that residual connections between layers of a variational ansatz can prevent barren plateaus in models which would otherwise contain them. Secondly, we implement a quantum analogue of average pooling layers from convolutional networks using single qubit controlled basic arithmetic operators and show that the LCU success probability remains stable for the MNIST database. This method can be further generalised to convolutional filters, while using exponentially fewer controlled unitaries than previous approaches. Finally, we propose a general framework for applying a linear combination of irreducible subspace projections on quantum encoded data. This enables a quantum state to remain within an exponentially large space, while selectively amplifying specific subspaces relative to others, alleviating simulability concerns that arise when fully projecting to a polynomially sized subspace. We demonstrate improved classification performance for partially amplified permutation invariant encoded point cloud data when compared to non-invariant or fully permutation invariant encodings. We also demonstrate a novel rotationally invariant encoding for point cloud data via Schur-Weyl duality. These quantum computing frameworks are all constructed using the LCU method, suggesting that further novel quantum machine learning algorithms could be created by utilising the LCU technique. 
 
\end{abstract}
\maketitle

\section{Introduction}

Quantum computing is an emerging technology with the potential to solve certain problems far more efficiently than classical techniques \cite{Shor_1997, Cao_2019, herman2022survey}. The search for useful applications of quantum computers within the domain of machine learning is of particular importance given the significant impact classical machine learning has had on many fields and industries in recent years \cite{touvron2023llama, chang2023muse, frieder2023mathematical}. Various quantum machine learning algorithms have been proposed, with some of the most common examples attempting to recreate a quantum analogue of the classical neural network \cite{Abbas_2021, Beer_2020}. However, the inherently unitary operations of quantum algorithms impose significant limitations, as they restrict the implementation of non-unitary operations that are essential in many classical machine learning models. To address this issue, we explore the Linear Combination of Unitaries (LCU) method \cite{Childs2012LCU, Berry_2015}, a technique that employs ancilla qubits to probabilistically implement a linear combination of unitary operations, which therefore permits the implementation of non-unitary operations in quantum circuits. This unlocks the possibility of implementing non-unitary operations within quantum machine learning (QML) applications. We utilise this to present several novel applications of the LCU method with various advantages in the field of quantum machine learning.

In Section~\ref{sec:resnet} we adapt the classical residual learning framework to quantum variational circuits, facilitating a quantum native implementation of Residual Networks (ResNet) \cite{he2015deep}. By partially skipping layers of the network, classical residual networks are able to avoid the vanishing gradient problem by allowing the gradients to flow through shallower sections of the network \cite{he2015deep, marion2022scaling}. We show that a quantum ResNet may similarly provide a method of avoiding barren plateaus in Variational Quantum Circuit (VQC) models, by maintaining shallow depth contributions in the final loss function. Furthermore, we show that by including terms that parameterise the strength of the residual connections, it is possible to increase the lower bound of the probability of success of the LCU procedure helping to alleviate one of the LCU method's principal issues.

In Section~\ref{sec:cnn} we implement a quantum analogue of average pooling operations from convolutional neural networks (CNN) \cite{RawatCNN, SHARMA2018377, SaffarCNN} by using the LCU method. This provides an efficient implementation of average pooling for amplitude encoded image data which we demonstrate for any size pooling window. This may further be generalised to convolutional filters leading to an exponential improvement in the number of controlled unitaries required compared to previous techniques \cite{wei2021quantumconvolutionalneuralnetwork, chen2022novelarchitectureparameterizedquantum}. We demonstrate that the LCU probability of success will equal $1$ in the case where all pixels are of the same colour and will decrease when pixels in the local pooling window of dimension $D \times D$ become more diverse. For real-world images, we provide the intuition that the probability remains relatively stable since most pixels are similar locally, except at the edges of subjects in the image. This intuition is supported by empirical evidence on $N \times N$ pixel images in the MNIST \cite{deng2012mnist} database, which shows that probability decreases but levels off to a finite value as $D$ increases, and shows no discernible trend when increasing image size $N$. 

In Section~\ref{sec:irrep_projections} we present a method for projecting quantum encoded data to any combination of irreducible representation subspaces of a finite group, presenting a general framework for implementing full or partial symmetries in the encoding step of quantum machine learning models. These techniques are intended to reduce the effective dimension of quantum encoded data in an effort to improve generalisation performance, which has been reported to decline as the number of qubits and hence the dimension of the encoded quantum states increases \cite{huang_power_2021}. We show that this technique can recreate previous work on permutation invariant encodings for point cloud data \cite{heredge2023permutation} as a special case. Furthermore, we implement a novel rotationally invariant encoding for point cloud data using the new technique by leveraging Schur-Weyl duality. This results in an encoded quantum state for point cloud data that is invariant if classical input point cloud data is rotated in 3-dimensional space, hence strongly enforcing rotation invariance on any model that subsequently uses this rotationally invariant input state. We further show that any combination of projections can be implemented at once, allowing certain symmetric subspaces of the data to be amplified or contracted to give increased flexibility over the amount of symmetry in the encoding. We demonstrate that intermediate levels of permutation symmetry for point cloud encoded data leads to an improved classification performance when compared to a non-symmetric or a fully permutation symmetric encoding.

These implementations illustrate the ability of the LCU method to benefit the field of QML and a summary of these contributions is provided in Table~\ref{tab:intro-table}. This work covers three different algorithmic frameworks that all utilise the same LCU method in their construction which are located in self-contained sections which may be read in any order:
\begin{itemize}
  \item \textbf{Quantum Native ResNet} \newline is detailed in Section~\ref{sec:resnet}
  \item \textbf{Quantum Native Average Pooling} \newline is detailed in Section~\ref{sec:cnn}
  \item \textbf{Irreducible Representation Projections} \newline are detailed in Section~\ref{sec:irrep_projections}
\end{itemize}
The remainder of this introductory section will introduce these three frameworks, followed by a summary of the LCU method used in their constructions.

\begin{table*}[t!]
 \caption{Summary of the algorithms introduced in this work using the LCU framework. $G$ is a finite group. $\ket{\psi_r}$ is the portion of a state $\ket{\psi}$ that occupies the subspace of the irreducible representation $r$, with $\ket{\psi_{\text{rot}}}$ meaning the portion of $\ket{\psi}$ that is in the rotationally invariant subspace. Note that $\langle \psi_r \rvert \psi_r \rangle$ is not necessarily equal to 1, as the condition is $\langle \psi \rvert \psi \rangle = \sum_r \langle \psi_r \rvert \psi_r \rangle$ = 1. The parameters $a_r$ can be chosen by the user subject to normalisation constraints. In average pooling we use a pooling window of dimension $D \times D$ pixels. $L$ is the number of layers in the quantum ResNet. $\beta_l \in [0,1], \beta_l\in \mathbb{R}$ parametrise the strength of the skipped connections in the quantum ResNet. Throughout this work $\log \equiv \log_2$ is the logarithmic function for base 2. 
 \vspace{0.2cm} % This creates a horizontal line before the table
}
  \label{tab:intro-table}
  \centering
\begin{tabular}{|>{\raggedright\arraybackslash}p{2.9cm}||>{\raggedright\arraybackslash}p{4.8cm}|>{\raggedright\arraybackslash}p{1.7cm}|>
{\raggedright\arraybackslash}p{3.4cm}|>{\raggedright\arraybackslash}p{1.5cm}|}
\hline
 \textbf{Algorithm} & \textbf{Result} & \textbf{Ancilla Qubits} & \textbf{Success Probability} & \textbf{Section} \\
\hline
Quantum ResNet Residual Layers & Residual connections can avoid barren plateaus  & $\mathcal{O}(L)$ &$[\prod_{l} \big(1 - 4\beta_l(1-\beta_l)\big), 1]$ & Sec~\ref{sec:resnet} \\
 \hline
Average Pooling & Exponentially fewer controlled unitaries in averaging / convolutional step than previous techniques \cite{wei2021quantumconvolutionalneuralnetwork, chen2022novelarchitectureparameterizedquantum} & $\mathcal{O}(\log(D))$ & Shown for MNIST in Fig~\ref{fig:success_prob_L} and Fig~\ref{fig:success_prob_N}. & Sec~\ref{sec:cnn} \\
\hline
 Irreducible Subspace Projections & Can freely enforce and parameterise symmetry of a quantum encoded state & $\mathcal{O}(\log(\rvert G \rvert ))$ & $\sum_{r} \rvert a_r \rvert^2 \langle \psi_r \rvert \psi_r \rangle$ & Sec~\ref{sec:irrep_projections} \\
\hline
Rotationally Invariant Encoding & Encoding respects rotational invariance of point cloud data & $\mathcal{O}(\log(\rvert G \rvert ))$ & $\langle \psi_{\text{rot}} \rvert \psi_{\text{rot}} \rangle $ & Sec~\ref{sec:rotat-invariant-encoding} \\
\hline
\end{tabular}
\end{table*}

\subsection{Quantum Native ResNet}

A common issue in the training of variational quantum circuit (VQC) models is the issue of barren plateaus leading to vanishing gradients. Classical Residual Networks (ResNet) have profoundly impacted deep learning by enabling the training of extremely deep neural networks through the introduction of skip connections that mitigate the vanishing gradient problem \cite{he2015deep, marion2022scaling}. This suggests that the implementation of Residual Networks within a quantum variational model could provide a significant advancement in quantum machine learning if they could similarly be utilised to avoid the problem of barren plateaus inherent in many VQC models. This line of reasoning was suggested in a recent review of barren plateaus \cite{larocca2024review} with the caveat that the no-cloning theorem may make residual and skipped connections difficult to achieve in quantum circuits. In this work, we shall show a probabilistic implementation of a quantum native ResNet that does not require cloning states and could show promise against barren plateaus in VQC models.

There has been active interest in building quantum ResNet inspired algorithms in the literature, with a primary focus being on quantum-classical hybrid models that use the powerful modern architecture of classical ResNet models, while including quantum machine learning subroutines to produce a hybrid algorithm \cite{HASSAN2024105560, Sagingalieva_2023,zaman2024comparative}. There have also been proposals to implement quantum residual connections by utilising several quantum neural network models in series with residual connections between them \cite{Kashif_2024}. A native VQC implementation of a quantum ResNet algorithm has previously been introduced in \cite{crognaletti2024estimateslossfunctionconcentration, CrognalettiQResNet}, in which skipped connections are possible in a VQC model if the implemented variational layers are restricted to be of the form of a unitary circuit followed by a product Pauli encoding, followed by the unitary circuit conjugate. This implements the quantum ResNet natively on a quantum device; however, as the authors note, this does not cover any general unitary variational layer $W_l(\mathbf{\theta}_l)$, since implementing residual connections for a general $W_l(\mathbf{\theta}_l)$ would not be a unitary process overall and was therefore not considered in this approach. A common theme explored in these previous works was utilising the effectiveness of classical ResNet models in tackling the vanishing gradients problem of deep neural networks and applying this reasoning in a quantum setting to the problem of barren plateaus. 

Applying residual connections in a quantum setting corresponds to operating on a state $\ket{\psi_{l-1}}$ to produce
\begin{equation}
 \ket{\psi_l} = W_l(\mathbf{\theta}_l)\ket{\psi_{l-1}} + \ket{\psi_{l-1}},
\end{equation}
whereby a portion of the previous state $\ket{\psi_{l-1}}$ is able to skip the variational operator $W_l(\mathbf{\theta}_l)$ in that layer. This allows the overall model cost function to retain terms that have only passed through one variational layer, corresponding to very shallow circuits, while still providing terms that have been passed through all layers and hence potentially very deep circuits. The operator which would result in this would be of the form
\begin{equation}
 A_l = W_l(\mathbf{\theta}_l) + I ,
\end{equation}
which in general is not a unitary operator. The key challenge therefore in implementing a quantum native ResNet analogue is that it would require implementation of non-unitary operators, something that becomes possible with the LCU framework. In this work, we translate the ResNet architecture to a quantum setting by applying residual connection with the LCU framework to facilitate the flow of quantum information across deeper or more complex quantum circuits, while allowing the strength of the residual connection to be chosen freely. This introduces a potential new class of VQC ResNet models which could provide possible protection against barren plateaus in complex VQC models.

The use of ancilla qubits to implement ResNet-like architectures in QML models has previously been explored \cite{wen2024enhancingexpressivityquantumneural}, revealing that residual connections within the data encoding segment of a circuit can expand the frequency spectrum of the resulting model, leading to more expressive encodings. In contrast, our work introduces a framework for quantum ResNets exclusively within the variational portion of a VQC model. We provide a detailed proof of the probability of success of the LCU procedure, demonstrating that the lower bound of this probability can be adjusted by varying the residual connection strength. Additionally, we illustrate that applying quantum residual layers to a model can mitigate the occurrence of barren plateaus in circuits which would otherwise contain them. We also demonstrate that quantum ResNets can be viewed as equivalent to ensembles of unitary VQC models with additional non-unitary terms. While we show that quantum ResNets can avoid barren plateaus, we also discuss how they may likely be classically simulatable in many cases, at least in the case that the connection between absence of barren plateaus and classical simulatability \cite{cerezo2023does} is valid for the constituent components. We propose the solution to this quantum ResNet simulatability issue may lie in the non-unitary terms and suggest a characterisation of these terms as a topic for further research.

\subsection{Average Pooling Layers}

Classical Convolution Neural Networks (CNN) are of significant importance to machine learning, mainly due to their structured manner of handling image and video data \cite{RawatCNN, SHARMA2018377, SaffarCNN, vu2024qcnn}. Inspiration from these models has led to the development of quantum analogues. Quantum convolution neural networks have been previously proposed \cite{Cong_2019, umeano2023learn} to classify quantum states with certain symmetry-protected topological phases and to classify image datasets \cite{Hur_2022, LiYao_2020}. In these models variational circuits are used during convolutional layers followed by further variational circuits and measurements in the pooling layer to perform dimensionality reduction such that the operations are performed in a manner that respects certain symmetries of the data. Significant benefits of these models, such as avoiding barren plateaus \cite{Pesah_2021}, have been identified. Quantum-classical hybrid techniques have also been developed that utilise classical convolution neural network architectures alongside quantum models \cite{Liu_2021}.

The implementation we present here differs in that we focus entirely on implementing a subroutine of classical CNN models, the average pooling layer, for amplitude encoded image data. We consider the subroutine where a pooling window of size $D \times D$ passes over the image and outputs the average of all pixels found within the pooling window. We show that this can be implemented natively on a quantum circuit by utilising the LCU method which could lead to an improvement over performing the subroutine classically, as quantum parallelism allows the averaging operation to apply to all pixel simultaneously. This demonstrates the possible utility of the LCU method, while providing a potentially advantageous subroutine for future quantum convolutional neural network models.

Previous work has investigated the LCU technique for creating convolutional layer filters \cite{wei2021quantumconvolutionalneuralnetwork} based on spatial filtering \cite{Yao_2017}, which can recover average pooling as a special case. This work did not generalise to any size $D$, restricting instead to $D=3$, but stated the method would be require $D^2$ multi-controlled operators in general. In contrast, we show a valid proof for any $D$ and an efficient circuit implementation which requires only $\mathcal{O}(\log(D))$ single-qubit controlled unitaries, leading to an exponential improvement in $D$ over previous work \cite{wei2021quantumconvolutionalneuralnetwork}. We also show that our construction can be generalised by adjusting the ancilla qubit state initialisation in order to implement a general convolutional layer filter, recovering the main result of \cite{wei2021quantumconvolutionalneuralnetwork} while maintaining an exponential improvement in $D$.

\subsection{Irreducible Representation Subspace Symmetry Projections}

\begin{figure*}[htb]%
\centering
\includegraphics[width=0.95\linewidth]{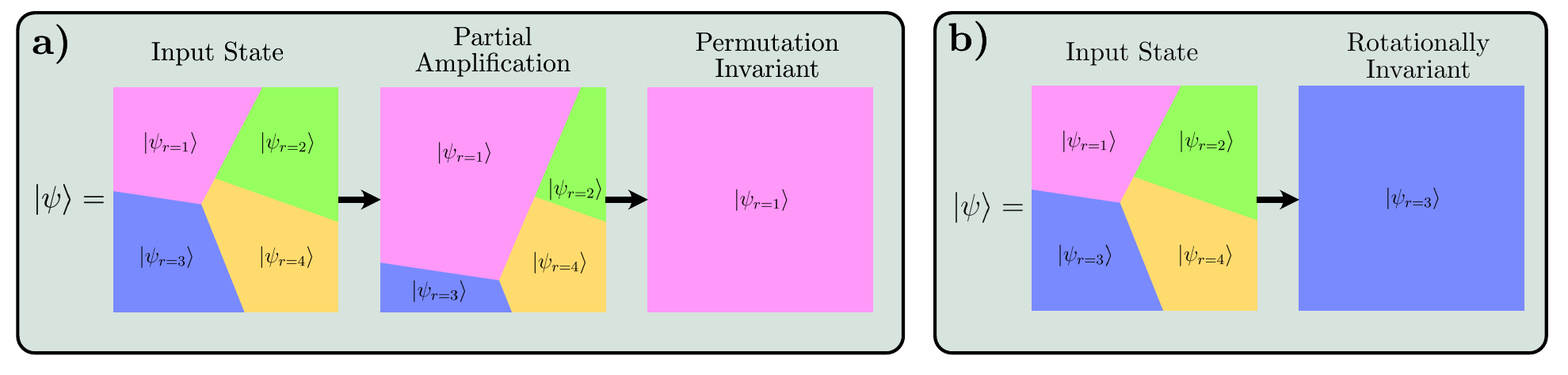}
\caption{Visual representation of the irreducible subspace projection of a quantum state. The input state $\ket{\psi}$ has components in each irreducible representation subspace denoted by $\ket{\psi_r}$. a) The component of $\ket{\psi}$ in the subspaces can be partially amplified to any desired ratio between the subspace components. Controlling the amplification of the permutation symmetric subspace $\ket{\psi_{r=1}}$ for point cloud data is the focus of Section~\ref{sec:subspace-amplification}. The state can also be fully projected to $\ket{\psi_{r=1}}$ which for the permutation group $S_n$ would correspond to a fully permutation invariant state. Permutation invariant encodings have been the subject of previous work \cite{heredge2023permutation} and are reviewed within our new framework in Section~\ref{sec:perm-invariance}. b) The input state $\ket{\psi}$ is fully projected to $\ket{\psi_{r=3}}$ within the $r=3$ irreducible subspace. In Section~\ref{sec:rotat-invariant-encoding} we show that this corresponds to creating rotationally invariant encodings of point cloud data.}\label{fig:projection-graphical}
\end{figure*}

A known issue in Quantum Machine Learning (QML) is that as the number of qubits increases there is a decrease in the generalisation performance of algorithms \cite{huang_power_2021}. A common empirical explanation of this is that the exponentially large Hilbert space leads to an overly expressive feature encoding where overtraining on the data becomes commonplace. Without an accompanied exponential increase in the training data, this leads to an overall reduction in the performance in the validation data set. A possible solution to the issue of overtraining is to reduce the expressibility of the encoding. Examples of techniques that have attempted this range from projecting kernels to a lower-dimensional space \cite{inductive_bias} and approaches that are capable of encoding inductive biases directly into quantum states \cite{bowles2023contextuality}. Furthermore, geometric QML techniques have studied methods for creating variational circuits that are equivariant with respect to data symmetries \cite{meyer22, Nguyen22, west2024provably, Schatzki_2024, tüysüz2024symmetry} in similar attempts to reduce the expressibility of QML models. We instead focus on implementing symmetries directly into the quantum encoded data, which is a stricter implementation of symmetry, meaning our procedure is agnostic to the trainable classification procedure.

In previous work a quantum encoding was proposed using permutation symmetry, which led to a reduction in the dimensionality of the encoding and improved classification performance \cite{heredge2023permutation}. This was discussed in the context of point cloud data (unordered collection of points in 3 dimensions that collectively represent an image), where each point cloud $X$ consists of $n$ points and each point $\textbf{p}_i$ is a 3-dimensional vector. As points do not have intrinsic ordering, the ordering of the points as they are input into a classification algorithm should ideally have no effect on the outcome. Therefore, the point cloud data naturally has point ordering permutation invariance. In general, a machine learning classifier $f$, could return a different result depending on the order of the points, that is, $f(\textbf{p}_1, \textbf{p}_2) \neq f(\textbf{p}_2, \textbf{p}_1)$, unless it has been specifically constructed to respect the permutation symmetry. The permutation invariant encoding \cite{heredge2023permutation} acts by creating an equal quantum superposition of all permutations of the data. In the two-state case, given two points $\textbf{p}_1, \textbf{p}_2$ encoded into the quantum state $\ket{\textbf{p}_1}\otimes\ket{\textbf{p}_2}$ this would correspond to preparing the permutation invariant state
\begin{equation} 
\ket{X_s}=\frac{1}{\sqrt{2}}(\ket{\textbf{p}_1}\ket{\textbf{p}_2}+\ket{\textbf{p}_2}\ket{\textbf{p}_1}),
\end{equation}
However, it has been shown that the permutation invariant state preparation procedure
\begin{equation}
 \ket{\textbf{p}_1}\ket{\textbf{p}_2} \Rightarrow \frac{1}{\sqrt{2}}(\ket{\textbf{p}_1}\ket{\textbf{p}_2} + \ket{\textbf{p}_2}\ket{\textbf{p}_1}),
\end{equation}
cannot be implemented via a unitary operation \cite{buzek_optimal_2000}. However, this process can be implemented in a probabilistic manner using ancilla qubits \cite{barenco1996}. This was shown to lead to improved classification performance for point cloud image classification \cite{heredge2023permutation}. 

In this work, we demonstrate a generalisation of this technique that allows projections to any irreducible representation subspace of a finite group. Unlike previous works \cite{heredge2023permutation}, this means that we are no longer restricted to the symmetric subspace or the permutation group $S_n$. Furthermore, our technique allows for linear combinations of projections to any irreducible representation subspaces.

Utilising this framework, we demonstrate a rotationally invariant encoding for point cloud data as an example use case. This encoding produces the same quantum encoded state each time, even if the data input point cloud is rotated by any amount in 3-dimensions. We show that this is achieved both theoretically and numerically. This is a highly desirable property of the model as point cloud data naturally has rotational symmetry. Especially in applications such as computer vision for autonomous vehicles it is of upmost importance that subjects in the image, such as pedestrians, are correctly identified regardless of the angle from which they are being viewed. Previous work has focused on implementing rotational and permutation symmetry in models for point cloud data \cite{li2024enforcing}, which did so by implementing an equivariant variational model. We highlight that our work does not use equivariant variational models but projects quantum input states to a rotationally invariant subspace, hence strictly enforcing rotational invariance into any model for which this input encoded state is passed, as we effectively delete all information of the state which is not rotationally invariant. 

The permutation and rotationally invariant encodings mentioned previously succeed in reducing the dimension of the encoding, but we note that this reduction may indeed be too drastic, which could lead to classically simulatable approaches or simply delete too much information about the input state, hindering the model performance. We therefore show how our new framework allows for linear combinations of projections that can be utilised to introduce parameterised symmetry subspace amplification. In this setting the dimension of the quantum state can remain exponentially large, while subspaces associated with certain symmetries can have their relative weightings adjusted. We focus on the weighting of the permutation symmetric subspace relative to all other subspaces, which can be continuously adjusted using a hyperparameter $\alpha$. We show that by implementing an intermediate amount of permutation symmetry for point cloud data classification, it is possible to gain higher accuracy scores than using either non-invariant encodings or the fully permutation invariant encodings suggested in previous work \cite{heredge2023permutation}. A visualisation of these applications is shown in Figure~\ref{fig:projection-graphical}.

\subsection{Linear Combinations of Unitaries Method}\label{sec:lcu-framework}

All results in this paper are specific cases of the LCU method described in this section. Let us define the general framework of how the LCU method works in a quantum circuit as detailed in \cite{Childs2012LCU, Berry_2015}. An insightful tutorial on the LCU method can also be found at \cite{pennylane-tutorial}. The LCU procedure allows the implementation of any operator $A$ that is itself a linear combination of $N$ unitary operators $U_j \in SU(2^n)$. The operator $A$ acts on a target state that is contained in an $n$ qubit target register. The target state is denoted $\ket{\psi} \in (\mathbb{C}^2)^{\otimes n}$. We can define the operator $A$ acting on $\ket{\psi}$ as
\begin{equation}
 A \ket{\psi} = \frac{1}{\Omega'} \sum_{j=1}^N \alpha_j U_j \ket{\psi},
\end{equation}
where for simplicity $\alpha_j \geq 0 ,\alpha_j \in \mathbb{R}$ and any negative sign or complex phase can be absorbed into the unitary $U_j \in SU(2^n)$ and $\Omega'$ is a normalisation constant for the final target state. 

\begin{figure}[h]%
\centering
\includegraphics[width=1\linewidth]{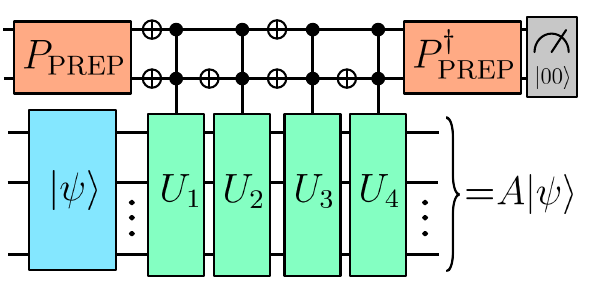}
\caption{Visual representation of the LCU procedure to implement $A\ket{\psi} = \frac{1}{\Omega'}(\alpha_1 U_1 + \alpha_1 U_2 +\alpha_1 U_3 + \alpha_1 U_4)\ket{\psi}$. The upper two qubits are the ancilla qubits initially prepared by the $P_{\text{PREP}}$ operator to be in the state $P_\text{PREP} \ket{b_1}= \frac{1}{\sqrt{\Omega}}\sum_{j=1}^{4} \sqrt{ \alpha_j }\ket{b_j}$. Note that in this case $\ket{b_1} \equiv \ket{00}$ and will implement $U_1$. $\ket{b_2} \equiv \ket{01}$ and will implement $U_2$. $\ket{b_3} \equiv \ket{10}$ and will implement $U_3$. $\ket{b_4} \equiv \ket{11}$ and will implement $U_4$. The inverse preparation $P_{\text{PREP}}^\dagger$ is then applied and the ancillas are measured. If the ancillas are measured to be in the state $\ket{b_1} \equiv \ket{00}$ then $A$ is successfully implemented on the target state $\ket{\psi}$. This example demonstrates how multi-control unitary gates can be implemented in practice which are controlled by the ancilla qubits and applied to the target register containing $\ket{\psi}$. }\label{fig:lcu-graphical}
\end{figure}

In order to implement this, we need to define an ancilla preparation operator $P_\text{PREP}$ which prepares the $k$ ancilla qubits, that are initially in a basis state $\ket{b_1}$, into the following state
\begin{equation}
 P_\text{PREP} \ket{b_1}= \frac{1}{\sqrt{\Omega}}\sum_{j=1}^{N} \sqrt{ \alpha_j }\ket{b_j},
\end{equation}
where $\Omega = \sum_{j=1}^N \alpha_j $ is a normalisation constant for the ancilla state and $\{ \ket{b_j} \} \in (\mathbb{C}^2)^{\otimes k}$ are basis states of the $2^k$ dimensional Hilbert space for the $k$ qubit ancilla register denoted by $\mathcal{H}= (\mathbb{C}^2)^{\otimes k}$, which can be taken to be the computational basis states. The operator itself can be explicitly written as
\begin{align}
  P_\text{PREP} = & \frac{1}{\sqrt{\Omega}}\sum_{j=1}^{N} \sqrt{ \alpha_j }\ket{b_j}\bra{b_1} + \sum_{i = 2}^{2^k}\sum_{j=1}^{2^k} u_{i,j} \ket{b_j}\bra{b_i},
\end{align}
where any terms $(...)\bra{b_i}$ for $i \geq 2$ will not be used and hence can be ignored.

 After the ancilla qubits are prepared, we then apply a selection operator $S_\text{SELECT}$. The selection operator applies the unitary operation $U_j \in SU(2^n)$ to the target register state $\ket{\psi}$ on the condition that the ancilla qubit is in the state $\ket{b_j}$, which can be defined as
\begin{equation}
 S_\text{SELECT}\ket{b_j}\ket{\psi} = \ket{b_j}U_j\ket{\psi}.
\end{equation}
If we now combine the preparation and selection operators, we have
\begin{equation}
 S_\text{SELECT} P_\text{PREP} \ket{b_1} \ket{\psi}= \frac{1}{\sqrt{\Omega}} \sum_{j=1}^{N} \sqrt{\alpha_j} \ket{b_j} U_j 
 \ket{\psi}.
\end{equation}
The final step consists of applying $P_\text{PREP}^\dagger$ which is defined as
\begin{equation}
 P_\text{PREP}^\dagger = \frac{1}{\sqrt{\Omega}}\sum_{j=1}^{N} \sqrt{ \alpha_j }\ket{b_1}\bra{b_j} + \sum_{i = 2}^{2^k}\sum_{j=1}^{2^k} u_{i,j}^* \ket{b_i}\bra{b_j}.
\end{equation}
Applying this to the circuit results in
\begin{align}
&P_\text{PREP}^\dagger S_\text{SELECT} P_\text{PREP} \ket{0} \ket{\psi} \nonumber \\
& = \frac{1}{\Omega} \sum_{j=1}^N \alpha_j \ket{b_1} U_j \ket{\psi} + \sum_{i=2}^{2^k}(...)\ket{b_{i}}\ket{\psi},
\end{align}
where $\sum_{i=2}^{2^k}(...)\ket{b_{i}}\ket{\psi}$ collects terms that will be discarded and can therefore be ignored. We now need to measure the ancilla register and discard any results when the ancilla is not measured in the $\ket{b_1}$ state. The probability of measuring the $\ket{b_1}$ state will equal the probability of success of the LCU method $\pi_S$ which can be written as
\begin{equation}
 \pi_S = \left\lvert \frac{1}{\Omega} \sum_{j=1}^N \alpha_j U_j \ket{\psi} \right\rvert^2.
\end{equation}
Discarding any results in which the ancilla is not measured in the $\ket{b_1}$ state we see that the remaining state will be projected to
\begin{equation}
 \bra{0}P_\text{PREP}^\dagger \rvert S_\text{SELECT} \rvert P_\text{PREP} \ket{0} \ket{\psi} \negthinspace =\negthinspace \frac{1}{\Omega'} \sum_{j=1}^N \alpha_j U_j \ket{\psi} ,
\end{equation}
where $\Omega' = \sqrt{\pi_S} \Omega$ is the normalisation constant for the final state, which can then be written as
\begin{equation}
 \frac{1}{\Omega'} \sum_{j=1}^N \alpha_j U_j \ket{\psi} = A \ket{\psi}.
\end{equation}
Hence, the operator $A$, which is a linear combination of unitaries and hence may itself be non-unitary, has been applied to the state $\ket{\psi}$ \cite{camps2023explicit}. 

The goal of this work is to demonstrate how the LCU method described above can be used in QML tasks to achieve desirable traits in the model architecture that are not possible in a strict unitary setting. The main results of this work rely on specifying preparation and selection operators, showing that they can be implemented on a quantum device, and then repeating the LCU framework detailed here to prove that they result in the desired non-unitary operation. An example LCU circuit for the implementation of $A\ket{\psi} = \frac{1}{\Omega'}(\alpha_1 U_1 + \alpha_1 U_2 +\alpha_1 U_3 + \alpha_1 U_4)\ket{\psi}$ is shown in Figure~\ref{fig:lcu-graphical}.

\section{Quantum Native ResNet}\label{sec:resnet}

\subsection{Variational Quantum Circuit Model Preliminaries}

For a vector of classical input data $\mathbf{x} \in \mathbb{R}^d$ a standard quantum variational circuit model consists of an $n$ qubit encoding circuit $V(\mathbf{x}) \in SU(2^n)$ that encodes the classical data into a quantum state $\ket{\psi_0} \equiv \ket{\phi(\mathbf{x})} = V(\mathbf{x}) \ket{0}^{\otimes n}$. This encoded state can also be represented as a density matrix defined by
\begin{equation}
 \rho(\mathbf{x}) = V(\mathbf{x}) \ket{0}^{\otimes n}\bra{0}^{\otimes n} V(\mathbf{x})^\dagger.
\end{equation}
Alternatively $\ket{\psi_0}$ and $\rho = \ket{\psi_0}\bra{\psi_0}$ could be quantum data in which no encoding process needs to be considered.

The input state $\ket{\psi_0}$ is then passed through $L$ layers of variational quantum circuits $W_l(\mathbf{\theta}_l) \in SU(2^n)$ where $\mathbf{\theta}_l = (\theta_{l1},\theta_{l2},... \theta_{lD_l} )$ is a vector of variational parameters that can be adjusted as the model is trained. We can therefore represent the overall variational circuit as
\begin{equation}
 W(\mathbf{\theta}) = \prod_{l=1}^L W_l(\mathbf{\theta}_l),
\end{equation}
where $\mathbf{\theta}$ is a vector containing all variational parameters for all layers such that $\mathbf{\theta} = \{ \mathbf{\theta}_l \}, l\in [1, L]$. Finally, measurement is made of some Hermitian observable $O$. For simplicity, we will consider a loss function for the model defined as
\begin{equation}
 \mathcal{L}_{\mathbf{\theta}}(\rho, O) \equiv \langle O \rangle_{\rho,\mathbf{\theta}} = \text{Tr}( W(\mathbf{\theta}) \rho W(\mathbf{\theta})^\dagger O),
\end{equation}
where any insights from this loss function with regards to barren plateaus can often be extended to other, more general loss functions. Training a model consists of variationally adjusting the parameters $\mathbf{\theta}$ until the cost function matches some known data labels $y$ to within some acceptable error.

\subsection{Quantum ResNet Implementation}

The key concept of ResNet \cite{he2015deep} is the introduction of residual blocks. Instead of learning a mapping $H(a)$ on some data, the model instead learns some residual function $W(a)$ such that
\begin{equation}
 H(a) = a + W(a) ,
\end{equation}
where $a$ is the input to the block, $W(a)$ is the learned residual mapping and $H(a)$ is the output of the block. This means the output $a_l$ of each layer $l$ can be defined by
\begin{equation}
 a_{l} = a_{l-1} + W_{l}(a_{l-1}, \mathbf{\theta}_{l}) ,
\end{equation}
where $a_0$ is the initial data input. In order to create a quantum version of this framework, we consider the quantum ResNet \cite{crognaletti2024estimateslossfunctionconcentration, CrognalettiQResNet} that can be defined as
\begin{equation}
 \ket{\psi_{l}} = \ket{\psi_{l-1}} + W_{l}(\mathbf{\theta}_{l}) \ket{\psi_{l-1}},
\end{equation}
where $W_{l}(\mathbf{\theta}_{l})$ is the variational unitary for layer $l$, $\ket{\psi_{l}}$ is the output of the layer, $\ket{\psi_{l-1}}$ is the input of the layer, and $\ket{\psi_{0}}$ corresponds to the initial input state. This initial data input state could be quantum data $\ket{\psi_0}$ or it could be a quantum state that encodes classical data $\ket{\psi_0} = \ket{\phi(\mathbf{x})}$, where $\mathbf{x} \in \mathbb{R}^d$ is a $d$-dimensional classical input data vector encoded into a quantum state through a data encoding circuit $V(\mathbf{x})$ such that $\ket{\phi(\mathbf{x})} = V(\mathbf{x})\ket{0}^{\otimes n}$. Implementation of quantum ResNet of this form involves applying an operator 
\begin{equation}
 A_{\text{RES}} = (I + W_l(\mathbf{\theta}_l)),
\end{equation}
to the state $\ket{\psi_{l-1}}$. This presents a problem, as the operator is not necessarily unitary in general for all $W_l(\mathbf{\theta}_l) \in SU(2^n)$. Therefore, we proceed with implementing it via the LCU method.

To maintain the most general case, we shall consider adding some additional control into the magnitude of the skipped connections parameterised for each layer by some amount $\beta_l \in \mathbb{R}$. Therefore the definition of a quantum residual connection we shall use is
\begin{equation}\label{eqn-residual-connection-quantum-def}
 \ket{\psi_{l}} = \frac{1}{\Omega'}\big((1 -\beta_l )\ket{\psi_{l-1}} + \beta_l W_{l}(\mathbf{\theta}_{l}) \ket{\psi_{l-1}}\big),
\end{equation}
where $\Omega'$ is a normalisation constant. 

This can be encoded with the LCU framework in a quantum circuit by applying the $P_{\text{PREP}}^{(l)}$ single qubit operator to the $l$-th ancilla qubit denoted by $\ket{0}^{(l)} $ defined by 
\begin{align}
 P_{\text{PREP}}^{(l)} = & ( \sqrt{1- \beta_l} \ket{0}^{(l)} + \sqrt{\beta_l} \ket{1}^{(l)}) \bra{0}^{(l)} \nonumber \\
 & + (...)\bra{1}^{(l)},
\end{align}
where terms involving $(...)\bra{1}^{(l)}$ do not have an effect on the circuit and can be ignored. Hence we can write the action of preparing the $l$-th ancilla qubit as
\begin{align}
 P_{\text{PREP}}^{(l)}\ket{0}^{(l)} = ( \sqrt{1- \beta_l} \ket{0}^{(l)} + \sqrt{\beta_l} \ket{1}^{(l)}) ,
\end{align}
such that $P_{\text{PREP}}\ket{0}^{\otimes L} = \bigotimes_{l=1}^L P_{\text{PREP}}^{(l)}\ket{0}^{(l)} $ and hence
\begin{equation}\label{eqn:lcu-resnet-prep}
 P_{\text{PREP}}\ket{0}^{\otimes L} = \bigotimes_{l=1}^L ( \sqrt{1- \beta_l} \ket{0}^{(l)} + \sqrt{\beta_l} \ket{1}^{(l)}) ,
\end{equation}
where for simplicity $\beta_l \in \mathbb{R}$. The selection operator is defined to be controlled by the $l$-th ancilla qubit as
\begin{equation}
 S_{\text{SELECT}}^{(l)} \ket{0}^{(l)} \ket{\psi}= \ket{0}^{(l)} \ket{\psi},
\end{equation}
\begin{equation}
 S_{\text{SELECT}}^{(l)} \ket{1}^{(l)} \ket{\psi}= \ket{1}^{(l)} W_{l}(\mathbf{\theta}_{l}) \ket{\psi}.
\end{equation}
We can then define the overall selection operator by specifying the ordering
\begin{equation}\label{eqn:lcu-resnet-selection}
 S_{\text{SELECT}}= \prod_{l=1}^L S_{\text{SELECT}}^{(L - l + 1)} = S_{\text{SELECT}}^{(L)} \cdots S_{\text{SELECT}}^{(1)},
\end{equation}
such that $S_{\text{SELECT}}^{(1)}$ is applied first.

 \begin{figure}[h]%
\centering
\includegraphics[width=1\linewidth]{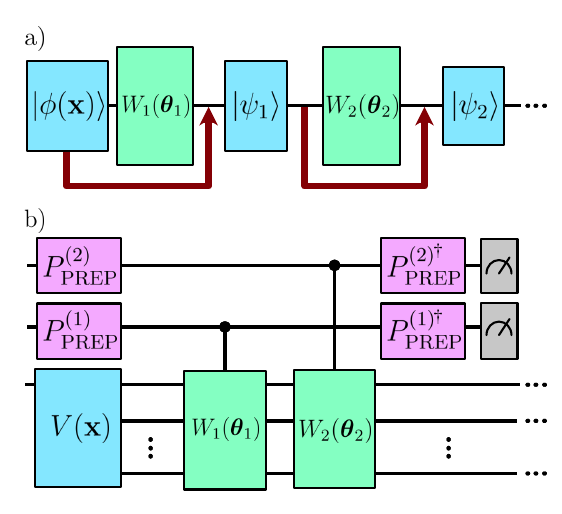}
\caption{a) Classical ResNet conceptual illustration in which residual information can be skipped over each layer. b) Quantum ResNet implementation using LCU methods, showing the first two layers only. In this case a quantum analogue of ResNet can be probabilistically implemented using ancilla qubits. The $P_{\text{PREP}}^{(1)}$ prepares the first ancilla qubit into the state $( \sqrt{1- \beta_{1}} \ket{0}^{(1)} + \sqrt{\beta_{1}} \ket{1}^{(1)})$; when this ancilla is measured at the end in the $\ket{0}^{(1)}$ state it corresponds to applying the operator $(1-\beta_1)\mathbb{I} + \beta_1 W_1(\mathbf{\theta}_1)$ to the initial input state $\ket{\phi(\mathbf{x})} = V(\mathbf{x})\ket{0}^{\otimes n}$. The process is repeated for future layers allowing a portion of the quantum state to skip each layer and be added to the output of the layer.}\label{fig:quantum_resenet_real}
\end{figure}

\begin{theorem}[Quantum ResNet]\label{thm:resnet-main}
 It is possible to probabilistically implement a quantum ResNet architecture on a quantum circuit where each layer of the variational circuit permits a residual connection defined as \[ \ket{\psi_{l}} = \frac{1}{\Omega'}\big( (1 -\beta_l )\ket{\psi_{l-1}} + \beta_l W_{l}(\mathbf{\theta}_{l}) \ket{\psi_{l-1}} \big) ,\]
 where $W_{l}(\mathbf{\theta}_{l})$ is the $l$-th layer of the variational circuit, the constants $\beta_l \in \mathbb{R}, \beta_l \in [0, 1]$ can be freely chosen, the input to the layer is the state $\ket{\psi_{l-1}}$ and the output of the layer is the state $\ket{\psi_{l}}$. The term $\Omega'$ is a normalisation constant. The algorithm can be applied for all layers $l\in[1,L]$ simultaneously with a fixed probability of success
 \[\pi_{S} \negthinspace = \negthinspace \prod_{l=1}^L \negthinspace \Big( \negthinspace 1 \negthinspace - \negthinspace 2\beta_l(1 \negthinspace - \negthinspace \beta_l)\big( 1 \negthinspace - \negthinspace \textup{Re}(\bra{\psi_{l-1}}W_l(\mathbf{\theta}_l)\ket{\psi_{l-1}})\big) \negthinspace \Big),\]
 where $\ket{\psi_0}$ is an initial input state given to the algorithm which can be quantum data, or a quantum state encoded by classical data via an encoding circuit $V(\mathbf{x})$ such that $\ket{\psi_0} = \ket{\phi(\mathbf{x})} = V(\mathbf{x})\ket{0}^{\otimes n}.$
\end{theorem}

\begin{proof}

We use the preparation and selection operators defined in Equation~\ref{eqn:lcu-resnet-prep} and Equation~\ref{eqn:lcu-resnet-selection}, respectively, with the LCU framework as outlined in Section~\ref{sec:lcu-framework}.

The $ P_{\text{PREP}}$ operator corresponds to performing single qubit initialisation on $L$ qubits and is clearly implementable on a quantum device. Likewise $S_{\text{SELECT}} $ corresponds to a controlled $W_l(\mathbf{\theta}_l)$ gate, where $W_l(\mathbf{\theta}_l)$ is defined to be unitary, which is controlled by the $l$-th qubit. As $W_l(\mathbf{\theta}_l)$ is unitary, the $S_{\text{SELECT}}$ operation can also be implemented on a quantum device.

We consider applying the LCU procedure to the first ancilla qubit only. Starting with the first qubit preparation operator
\begin{align}
 &P_{\text{PREP}}^{(1)}\ket{0}^{\otimes L} \ket{\psi_0} \nonumber \\
 & = ( \sqrt{1- \beta_1} \ket{0}^{(1)} + \sqrt{\beta_1} \ket{1}^{(1)}) \bigotimes_{l=2}^L \ket{0}^{(l)} \ket{\psi_0}.
\end{align}
Subsequently, the selection operator which is controlled by the first qubit is applied
\begin{align}
  S&_{\text{SELECT}}^{(1)} P_{\text{PREP}}^{(1)}\ket{0}^{\otimes L} \ket{\psi_0} \nonumber \\
  =& S_{\text{SELECT}}^{(1)}( \sqrt{1- \beta_1} \ket{0}^{(1)} + \sqrt{\beta_1} \ket{1}^{(1)}) \bigotimes_{l=2}^L \ket{0}^{(l)} \ket{\psi_0} \nonumber \\
  =& \sqrt{1- \beta_1} \ket{0}^{(1)} \bigotimes_{l=2}^L \ket{0}^{(l)} \ket{\psi_0} \nonumber \\
  & + \sqrt{\beta_1} \ket{1}^{(1)} \bigotimes_{l=2}^L \ket{0}^{(l)} W_1(\mathbf{\theta}_1) \ket{\psi_0}.
\end{align}
Consider the conjugate preparation operator, 
\begin{align}
 (P_{\text{PREP}}^{(1)})^\dagger = & \ket{0}^{(1)} ( \sqrt{1- \beta_1} \bra{0}^{(1)} + \sqrt{\beta_1} \bra{1}^{(1)}) \nonumber \\
 & + \ket{1}^{(1)}(...),
\end{align}
where $\ket{1}^{(1)}(...)$ collects terms that will not be used, as we will require the ancilla to be in the $\ket{0}^{(1)}$ state. Applying this conjugate operator we can write the state as
\begin{align}
  (P&_{\text{PREP}}^{(1)})^\dagger 
 S_{\text{SELECT}}^{(1)} P_{\text{PREP}}^{(1)}\ket{0}^{\otimes L} \ket{\psi_0} \nonumber \\
  = & \ket{0}^{(1)}\Big( (1 - \beta_1) \bigotimes_{l=2}^L \ket{0}^{(l)} \ket{\psi_0} \nonumber \\
 & + \beta_1 \bigotimes_{l=2}^L \ket{0}^{(l)} W_1(\mathbf{\theta}_1) \ket{\psi_0}\Big) + \ket{1}^{(1)}(...),
\end{align}
where $\ket{1}^{(1)}(...)$ collects terms that will be discarded and can therefore be ignored. The first ancilla qubit is now measured, and the result is discarded unless it is found to be in the $\ket{0}^{(1)}$ state. The probability of measuring the ancilla in the $\ket{0}^{(1)}$ state will be equal to
\begin{align}
 \pi_1 & = \left\lvert (1 - \beta_1) \ket{\psi_0} + \beta_1 W_1(\mathbf{\theta}_1) \ket{\psi_0}) \right\rvert^2 \nonumber \\
  & = 1 - 2\beta_1(1 - \beta_1)\big( 1 - \text{Re}(\bra{\psi_0}W_1(\mathbf{\theta}_1)\ket{\psi_0})\big),
\end{align}
which will depend on the strength of the residual connection parameterised by $\beta_1$ as well as the real component of $\bra{\psi_0}W_1(\mathbf{\theta}_1)\ket{\psi_0}$, which will depend on the quantum encoded state $\ket{\psi_0}$ and the variational operator used $W_1(\mathbf{\theta}_1)$. 

If we succeed in finding the first ancilla in the $\ket{0}^{(1)}$ state, then the target state will now be
\begin{align}
  & \bra{0}^{(1)} (P_{\text{PREP}}^{(1)})^\dagger S_{\text{SELECT}}^{(1)} P_{\text{PREP}}^{(1)}\ket{0}^{\otimes L} \ket{\psi_0} = \nonumber \\ 
  & \frac{1}{\sqrt{\pi_1}}\big( (1 \negthinspace - \negthinspace \beta_1) \bigotimes_{l=2}^L \ket{0}^{(l)} \ket{\psi_0} \negthinspace+\negthinspace \beta_1 \bigotimes_{l=2}^L \ket{0}^{(l)} W_1(\mathbf{\theta}_1) \ket{\psi_0} \big) ,
  \end{align}
and hence ignoring the extra ancilla qubits we have
\begin{equation}
 \ket{\psi_1} = \frac{1}{\Omega_1'}\big((1 - \beta_1) \ket{\psi_0} + \beta_1 W_1(\mathbf{\theta}_1) \ket{\psi_0}\big),
\end{equation}
where $\Omega'_1 = \sqrt{\pi_1}$ is the required normalisation constant. Hence, for layer $l=1$ we have shown that the quantum residual connection defined in Equation~\ref{eqn-residual-connection-quantum-def} is implemented for the original input state $\ket{\psi_0}$. 

Assume now that $\ket{\psi_{p-1}}$ is correctly implemented by the procedure and consider the subsequent implementation of $\ket{\psi_p}$. We prepare the $p$-th ancilla using

\begin{align}
 & P_{\text{PREP}}^{(p)}\ket{0}^{\otimes L-p} \ket{\psi_{p-1}} \nonumber \\
 & = ( \sqrt{1- \beta_p} \ket{0}^{(p)} + \sqrt{\beta_p} \ket{1}^{(p)}) \bigotimes_{l=p + 1}^L \ket{0}^{(l)} \ket{\psi_{p-1}}.
\end{align}
Through the same procedure as previously we see
\begin{align}
  & S_{\text{SELECT}}^{(p)} P_{\text{PREP}}^{(p)} \ket{0}^{\otimes {L- p}} \ket{\psi_{p-1}} \nonumber \\
  & = \sqrt{1- \beta_p} \ket{0}^{(p)} \bigotimes_{l=p + 1}^L \ket{0}^{(l)} \ket{\psi_{p-1}} \nonumber \\
  & + \sqrt{\beta_p} \ket{1}^{(p)} \bigotimes_{l=p + 1}^L \ket{0}^{(l)} W_p(\mathbf{\theta}_{p}) \ket{\psi_{p-1}}.
\end{align}
Applying the conjugate preparation operator for the $p$-th qubit which may be written as
\begin{align}
 (P_{\text{PREP}}^{(p)})^\dagger = & \ket{0}^{(p)} ( \sqrt{1- \beta_p} \bra{0}^{(p)} + \sqrt{\beta_p} \bra{1}^{(p)}) \nonumber \\
 & + \ket{1}^{(p)}(...),
\end{align}
to the ancillas we see that
\begin{align}
  & (P_{\text{PREP}}^{(p)})^\dagger S_{\text{SELECT}}^{(p)} P_{\text{PREP}}^{(p)}\ket{0}^{\otimes L - p} \ket{\psi_0} \nonumber \\
  & = (1 - \beta_p) \ket{0}^{(p)} \bigotimes_{l=p + 1}^L \ket{0}^{(l)} \ket{\psi_{p-1}} \nonumber \\
  & + \beta_p \ket{0}^{(p)} \bigotimes_{l=p + 1}^L \ket{0}^{(l)} W_p(\mathbf{\theta}_p) \ket{\psi_{p-1}} + \ket{1}^{(p)}(...),
  \end{align}
where $\ket{1}^{(p)}(...)$ collects terms that will end up being discarded and hence can be ignored. Measuring the ancilla qubit we can see that the probability of measuring the $p$-th qubit in the $\ket{0}^{(p)}$ state is given by
\begin{align}
 \pi_p & = \left\lvert ( (1 - \beta_p) \ket{\psi_{p-1}} + \beta_p W_p(\mathbf{\theta}_p) \ket{\psi_{p-1}} ) \right\rvert^2 \nonumber \\
   & = 1 - 2\beta_p(1 \negthinspace - \negthinspace \beta_p)\big( 1 \negthinspace - \negthinspace \text{Re}(\bra{\psi_{p-1}}W_p(\mathbf{\theta}_p)\ket{\psi_{p-1}})\big) ,\label{eq:betaprob}
\end{align}
Ensuring the process is only continued if the ancilla is in the $\ket{0}^{(p)}$ state we find
\begin{align}
   \bra{0}^{(p)} & (P_{\text{PREP}}^{(p)})^\dagger S_{\text{SELECT}}^{(p)} P_{\text{PREP}}^{(p)}\ket{0}^{\otimes L - p} \ket{\psi_0} \nonumber \\
   = & \frac{1}{\sqrt{\pi_p}}\big( (1 - \beta_p) \bigotimes_{l=p + 1}^L \ket{0}^{(l)} \ket{\psi_{p-1}} \nonumber \\
   & + \beta_p \bigotimes_{l=p + 1}^L \ket{0}^{(l)} W_p(\mathbf{\theta}_p) \ket{\psi_{p-1}} \big).
   \end{align}
 Ignoring the unused ancillas we can therefore write the output of layer $l=p$ as 
 \begin{equation}
  \ket{\psi_{p}} = \frac{1}{\Omega'_p}\big( (1 -\beta_p )\ket{\psi_{p-1}} + \beta_p W_{p}(\mathbf{\theta}_{p})\ket{\psi_{p-1}} \big),
\end{equation} 
where $\Omega'_p = \sqrt{\pi_p}$ is a normalisation constant. Hence, assuming that the state $\ket{\psi_{p-1}}$ is correctly implemented by the procedure, we have shown that the output state $\ket{\psi_{p}}$ will be correctly implemented according to the quantum residual connection defined in Equation~\ref{eqn-residual-connection-quantum-def}. As we have shown that $\ket{\psi_1}$ can be implemented correctly from the data input state $\ket{\psi_0}$, then by inductive reasoning $\ket{\psi_l}$ is correctly implemented for all $l \in [L]$.

The overall probability of success $\pi_S$ will rely on all $L$ qubits being successfully measured in the $\ket{0}^{(l)}$ state such that $\pi_{S} = \prod_{l=1}^L \pi_l $ and hence
\begin{align}
 \pi_{S} \negthinspace = \negthinspace \prod_{l=1}^L \negthinspace \Big( \negthinspace 1 \negthinspace - \negthinspace 2\beta_l(1 \negthinspace - \negthinspace \beta_l)\big( 1 \negthinspace - \negthinspace \text{Re}(\bra{\psi_{l-1}}W_l(\mathbf{\theta}_l)\ket{\psi_{l-1}})\big) \negthinspace \Big).
\end{align}
as required.
\end{proof}

We therefore show that the most general form of ResNet where connections can be skipped at every layer in a variational quantum circuit is probabilistically implementable. An example circuit architecture that demonstrates the process in the first two layers is shown in Figure~\ref{fig:quantum_resenet_real}. This implementation requires the number of ancilla qubits to equal $L$. For a more qubit efficient but less general version, see Appendix~\ref{apn:input-focussed-quantum-resnet} in which the number of ancilla used scales $\mathcal{O}(\log(L))$.

\subsection{Probabilistic Scaling}

During the proof in the previous section we found in Equation~\ref{eq:betaprob} that the probability $\pi_l$ of measuring the $l$-the qubit in the $\ket{0}^{(l)}$ state, which is required for implementing the LCU method, to be
\begin{equation}
 \pi_l = 1 - 2\beta_l(1 - \beta_l)\big( 1 - \text{Re}(\bra{\psi_{l-1}}W_l(\mathbf{\theta}_l)\ket{\psi_{l-1}})\big).
\end{equation}
This depends on the strength of the residual connection $\beta_l$, the quantum state of the previous step $\ket{\psi_{l-1}}$ and the variational circuit $W_l(\mathbf{\theta}_l)$.

Note that in general $\text{Re}(\bra{\psi_{l-1}}W_l(\mathbf{\theta}_l)\ket{\psi_{l-1}})$ lies in the range $[-1, 1]$. Therefore, the probability of success lies in the range $[1 - 4\beta_l(1-\beta_l), 1]$. The worst-case scenario therefore occurs when $\beta_l = \frac{1}{2}$ as in this case the probability of success is within the range $[0, 1]$ and therefore there is a chance that the algorithm cannot be implemented. However, as $\beta_l$ is varied to be closer to $0$ or $1$ this gives an increasing lower bound of $1 - 4\beta_l(1-\beta_l)$ for the probability of success (if $\beta_l = 1$ the probability of success is equal to $1$, but would trivially mean that the skip connection is simply not performed. Likewise, if $\beta_l = 0$ then the success probability is equal to $1$ but corresponds to not implementing the variational circuit $W_l(\mathbf{\theta}_l)$ at all). This means that the lower bound of the probability of success varies with the strength of the residual connection for the layer, as shown in Figure~\ref{fig:beta-scaling}.

 \begin{figure}[h]%
\centering
\includegraphics[width=1\linewidth]{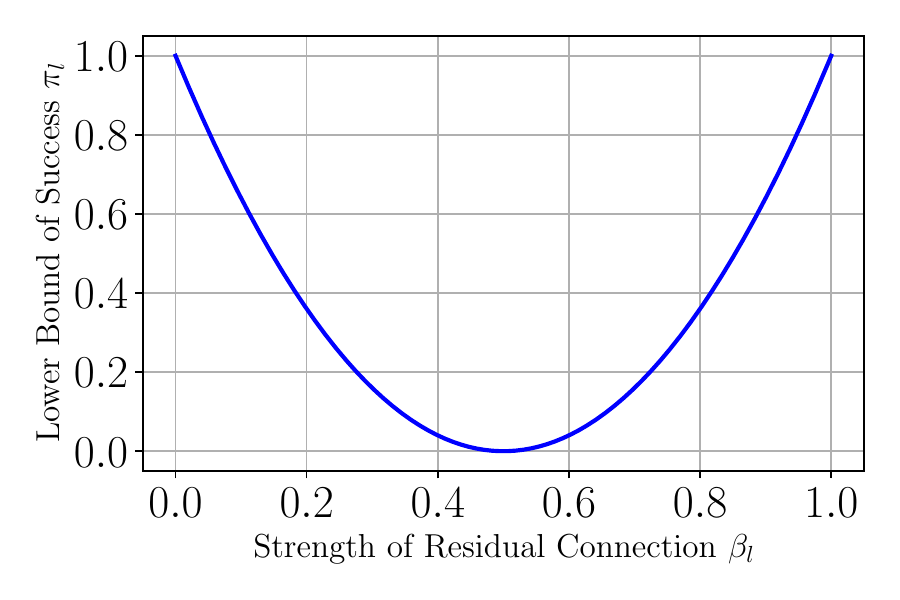}
\caption{Lower bound of success probability $\pi_l$ for layer $l$ in the LCU quantum native ResNet, as the strength of residual connection parameter $\beta_l$ varies. The lower bound is equal to $1-4 \beta_l(1-\beta_l)$. }\label{fig:beta-scaling}
\end{figure}

The term $\text{Re}(\bra{\psi_{l-1}}W_l(\mathbf{\theta}_l)\ket{\psi_{l-1}})$ determines whether we will be close to the lower bound probability $1 - 4\beta_l(1-\beta_l)$ or the upper bound of $1$. We see that if $W_l(\mathbf{\theta}_l) = \mathbb{I}$ then we reach the upper bound and if $W_l(\mathbf{\theta}_l) = - \mathbb{I}$ we reach the lower bound. If restrictions were placed on the form $W_l(\mathbf{\theta}_l)$ can take, there is potential to increase the lower bound probability even further.

As shown previously overall probability of success $\pi_S$ will rely on all $L$ qubits being successfully measured in the $\ket{0}$ state such that $\pi_{S} = \prod_{l=1}^L \pi_l $ and hence
\begin{align}
 \pi_{S} \negthinspace = \negthinspace \prod_{l=1}^L \negthinspace \Big( \negthinspace 1 \negthinspace - \negthinspace 2\beta_l(1 \negthinspace - \negthinspace \beta_l)\big( 1 - \text{Re}(\bra{\psi_{l-1}}W_l(\mathbf{\theta}_l)\ket{\psi_{l-1}})\big) \negthinspace \Big).
\end{align}
Therefore, the lower bound of the overall probability of success is given by $\prod_l^L \big(1 - 4\beta_l(1-\beta_l)\big)$ which can be adjusted to be between $0$ and $1$ by varying the strength of the residual layers through the variables $\beta_l$. Although we highlight that this cannot be arbitrarily adjusted, as if $\beta_l$ reaches $0$ or $1$ this would trivially correspond to no residual connection at all, and more generally the value of $\beta_l$ may be decided by alternative factors for a particular model or dataset.

If for a given architecture the probability of success for any layer is bounded between $[\pi_{\text{low}}, \pi_{\text{high}}]$. Then the probability of overall success must be bounded between $[\pi_{\text{low}}^L, \pi_{\text{high}}^L]$. In the case of $\pi_{\text{low}}, \pi_{\text{high}} < 1$, this means that the probability will decay exponentially with $L$, which is a key drawback of this method. However, if the number of layers used is chosen to scale logarithmically with the number of qubits $L = \mathcal{O}(\log(n))$ then the algorithm can still run in time $\mathcal{O}(\text{poly}(n))$ in general, where $n$ is the number of qubits for the target state register that initially contains $\ket{\psi_0}$. Note that when $L=1$ we have a purely unitary model. Therefore, for $L>1$ we will have access to models that are at least as expressive as a standard unitary VQC model, with the potential to be even more expressive. Hence, even setting $L=\mathcal{O}(1), L>1$ would still allow a greater variation in expressivity of the models compared to the standard unitary case. Furthermore, as discussed previously, the lower bounds can be adjusted by varying the strength of the residual connections for a given architecture.

\subsection{Avoiding Barren Plateaus with Residual Connections}\label{sec:avoiding-bp-resnet}

Several works have been undertaken to characterise the conditions under which VQC models exhibit barren plateaus \cite{larocca2024review}. A model is said to contain barren plateaus if $\forall \theta_\mu \in \mathbf{\theta}$ the following condition holds \cite{Larocca_2022}
\begin{equation}
  \text{Var}_{\mathbf{\theta}} \left( \frac{\partial \langle O \rangle_{\rho, \mathbf{\theta}}}{\partial \theta_\mu} \right) \leq F(n) \text{  ,  } F(n) = \mathcal{O}\Big(\frac{1}{b^n}\Big)\text{  ,  } b>1.
\end{equation}
where $\text{Var}_{\mathbf{\theta}}$ is the variance calculated over a uniform distribution of parameters $\mathbf{\theta}$. If the model exhibits exponential gradient concentration then this makes training models difficult. The structure of the variational ansatz plays an important role in determining this value, as for a wide range of cases the loss function gradient variance will decay exponentially if the dynamical Lie algebra (DLA) of the variational circuit generators is exponential in dimension \cite{fontana2023adjoint, ragone2023unified}. Note that an important assumption in these characterisations is usually that the circuits are sufficiently deep to form approximate designs, meaning that the depth of the circuit is also an important consideration regarding the presence of barren plateaus \cite{ragone2023unified}. Previous work has shown that using an ensemble of several shallow depth models can avoid loss function concentration \cite{friedrich2024quantumneuralnetworkensemble}. We show that the quantum ResNet provides a method of creating exponentially large ensembles of unitary VQC models, in which shallower depth components can help avoid barren plateaus, while also providing additional non-unitary contributions to the variance. However, if these shallower layers can be efficiently simulated classically and the non-unitary terms exponentially vanish, then the entire model may be at risk of being classically simulated. This finding suggests the importance of further exploring the non-unitary components of quantum ResNet models.

\subsubsection{Avoiding Barren Plateaus}

By introducing residual connections between layers of the ansatz it is possible to effectively mitigate barren plateaus in the model. This is similar to the classical role of ResNet in which vanishing gradients are mitigated in deep networks by allowing the gradients to skip over portions of the overall network \cite{marion2022scaling, he2015deep}. For example, take the case where a variational ansatz can be separated into two separate sections $W_1(\mathbf{\theta}_1)$ followed by $W_2(\mathbf{\theta}_2)$. Furthermore, assume that when combined in series the overall operator $W_2(\mathbf{\theta}_2)W_1(\mathbf{\theta}_1)$ exhibits barren plateaus. Consider a residual connection introduced between the input and the output of the first layer $W_1(\mathbf{\theta}_1)$ with residual connection strength $\beta$, and no residual connection over the second layer $W_2(\mathbf{\theta}_2)$. Starting with an initial state $\ket{\psi_0}$ we see the states produced in this circuit are given by
\begin{equation}
 \ket{\psi_1} = \frac{1}{\Omega'}( (1-\beta)\ket{\psi_0} + \beta W_1(\mathbf{\theta}_1) \ket{\psi_0} ),
\end{equation}
\begin{equation}
 \ket{\psi_2} = W_2(\mathbf{\theta}_2) \ket{\psi_1}.
\end{equation}
where $\Omega'$ is a normalisation constant. It follows that the final state can be written as
\begin{align}
 \ket{\psi_2} = \frac{1}{\Omega'}( (1-\beta) W_2(\mathbf{\theta}_2)\ket{\psi_0} + \beta W_2(\mathbf{\theta}_2) W_1(\mathbf{\theta}_1) \ket{\psi_0} ).
\end{align}
This corresponds to applying the operator 
\begin{align}
 A(\mathbf{\theta})_{\text{RES}} = ((1-\beta) W_2(\mathbf{\theta}_2) + \beta W_2(\mathbf{\theta}_2) W_1(\mathbf{\theta}_1) ),
\end{align}
and normalising the resultant state by a factor $\frac{1}{\Omega'}$. We now consider some quantum encoded density matrix state $\rho$, where $\rho \equiv \ket{\psi_0}\bra{\psi_0}$, which is evolved by this variational operator $A(\mathbf{\theta})_{\text{RES}}$. The resulting loss function for the model with respect to a measurement operator $O$ is defined by
\begin{equation}
 \langle O \rangle_{\rho, \mathbf{\theta}} = \text{Tr}(A(\mathbf{\theta})_{\text{RES}} \rho A(\mathbf{\theta})_{\text{RES}}^\dagger O ),
\end{equation}

While $A(\mathbf{\theta})_{\text{RES}}$ is not in general unitary we can decompose its corresponding loss function into unitary evolution type terms and non-unitary evolution terms by considering the following expansion
\begin{align}
 \langle O \rangle_{\rho, \mathbf{\theta}} & = \frac{1}{\Omega'^2} \Big( (1-\beta)^2 \text{Tr}(  W_2(\mathbf{\theta}_2) \rho W_2(\mathbf{\theta}_2)^\dagger O ) \nonumber \\ 
  & + \beta^2 \text{Tr}( W_2(\mathbf{\theta}_2) W_1(\mathbf{\theta}_1) \rho W_1(\mathbf{\theta}_1)^\dagger W_2(\mathbf{\theta}_2)^\dagger O ) \nonumber \\
  & + \beta (1-\beta) \text{Tr}( W_2(\mathbf{\theta}_2) \rho W_1(\mathbf{\theta}_1)^\dagger W_2(\mathbf{\theta}_2)^\dagger O ) \nonumber \\ 
  & + \beta (1-\beta) \text{Tr}( W_2(\mathbf{\theta}_2) W_1(\mathbf{\theta}_1) \rho W_2(\mathbf{\theta}_2)^\dagger O ) \Big).
\end{align}
where we can collect the non-unitary terms together by representing the expression as 
\begin{align}
 \langle O \rangle_{\rho, \mathbf{\theta}} & = \frac{1}{\Omega'^2} \Big( (1-\beta)^2 \text{Tr}( W_2(\mathbf{\theta}_2) \rho W_2(\mathbf{\theta}_2)^\dagger O ) \nonumber \\ 
  & + \beta^2 \text{Tr}( W_2(\mathbf{\theta}_2) W_1(\mathbf{\theta}_1) \rho W_1(\mathbf{\theta}_1)^\dagger W_2(\mathbf{\theta}_2)^\dagger O ) \nonumber \\
  & + \beta(1-\beta) \text{Tr}( W_2(\mathbf{\theta}_2) \tilde{\rho}_{\mathbf{\theta}_1} W_2(\mathbf{\theta}_2)^\dagger O ),
\end{align}
where $\tilde{\rho}_{\mathbf{\theta}_1} = \rho W_1(\mathbf{\theta}_1)^\dagger + W_1(\mathbf{\theta}_1) \rho$. It is now possible to identify three distinct terms in the loss function. 

\begin{description}
    \item[Deep - $L_{\text{BP}} = \text{Tr}( W_2(\mathbf{\theta}_2) W_1(\mathbf{\theta}_1) \rho W_1(\mathbf{\theta}_1)^\dagger W_2(\mathbf{\theta}_2)^\dagger O )$]
    This component is identical to the unitary VQC loss function if no residual component had been implemented. It corresponds to a deep circuit consisting of $ W_1(\mathbf{\theta}_1)$ followed by  $W_2(\mathbf{\theta}_2)$ in series, which we will assume exhibits barren plateaus.

    \item[Shallow - $L_{\text{No-BP}} = \text{Tr}( W_2(\mathbf{\theta}_2) \rho W_2(\mathbf{\theta}_2)^\dagger O )$]
    This component is identical to the unitary VQC loss function if the variational circuit consisted of $ W_2(\mathbf{\theta}_2)$ only. We will assume that this model does not exhibit barren plateaus, as has been shown to be the case for many variational circuits \cite{Schatzki_2024, wiersema2023classification, west2024provably, diaz2023showcasing}, but may be vulnerable to being classically simulated \cite{cerezo2023does}.

    \item[Non-unitary - $L_{\tilde{\rho}} = \text{Tr}( W_2(\mathbf{\theta}_2) \tilde{\rho}_{\mathbf{\theta}_1} W_2(\mathbf{\theta}_2)^\dagger O )$]
    This component contains the non-unitary element of the loss function. It appears similar to the $L_{\text{No-BP}}$ term with the key difference that instead of $\rho$ being unitary evolved we have $\tilde{\rho}_{\mathbf{\theta}_1} = \rho W_1(\mathbf{\theta}_1)^\dagger + W_1(\mathbf{\theta}_1) \rho$. While $\tilde{\rho}_{\mathbf{\theta}_1}$ is Hermitian, it does not necessarily have a trace equal to 1 nor is it positive semidefinite in general.
\end{description}

\begin{figure*}[htb]%
\centering
\includegraphics[width=0.7\linewidth]{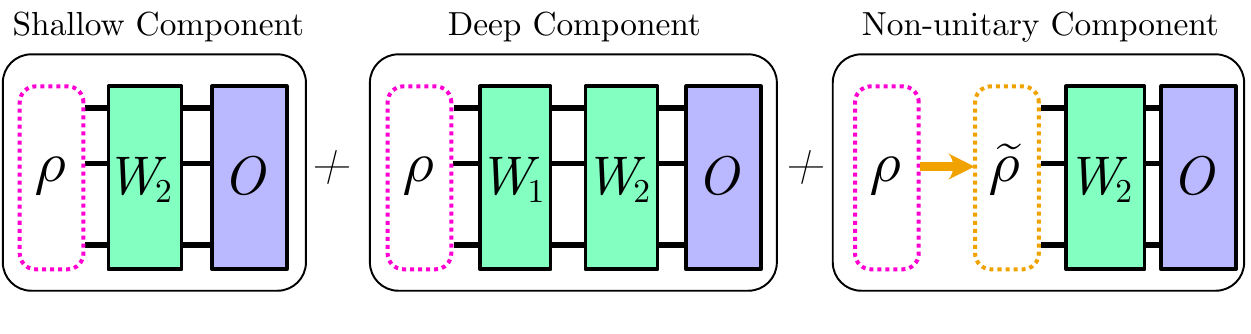}
\caption{Visualisation of the single residual connection $L=1$ quantum ResNet loss function $\langle O \rangle_{\rho, \mathbf{\theta}}$ when decomposed into the individual terms corresponding to the shallow component $L_{\text{No-BP}}$, deep component $L_{\text{BP}}$ and non-unitary component $L_{\tilde{\rho}}$. In the non-unitary component the arrow corresponds to the transformation $\rho \rightarrow \tilde{\rho} = \rho W_1^\dagger + W_1 \rho$.}\label{fig:loss_term_visualisation}
\end{figure*}

Setting $\beta=\frac{1}{2}$ and absorbing this factor along with the normalisation constant into the loss definition we can write
\begin{equation}
    \langle O \rangle_{\rho, \mathbf{\theta}} =  L_{\text{No-BP}} +  L_{\text{BP}} +  L_{\tilde{\rho}}.
\end{equation}
A visual interpretation of this resulting quantum ResNet model is shown in Figure~\ref{fig:loss_term_visualisation}. We can therefore find the gradient of this loss function with respect to parameters $\theta_\mu \in \mathbf{\theta}$ as
\begin{equation}
  \frac{\partial\langle O \rangle_{\rho, \mathbf{\theta}}}{\partial \theta_\mu} =  \frac{ \partial L_{\text{No-BP}}}{\partial \theta_\mu} +  \frac{\partial L_{\text{BP}} }{\partial \theta_\mu} +  \frac{\partial L_{\tilde{\rho}}}{\partial \theta_\mu}.
\end{equation}
The resulting gradient variance term can then be written as
\begin{align}
  & \text{Var}_{\mathbf{\theta}} \left( \frac{\partial \langle O \rangle_{\rho, {\mathbf{\theta}}}}{\partial \theta_\mu} \right) = \nonumber \\ 
  &  \text{Var}_{\mathbf{\theta}}(\frac{ \partial L_{\text{No-BP}}}{\partial \theta_\mu}) + \text{Var}_{\mathbf{\theta}}(\frac{ \partial L_{\text{BP}}}{\partial \theta_\mu}) + \text{Var}_{\mathbf{\theta}}(\frac{ \partial L_{\tilde{\rho}}}{\partial \theta_\mu}) \nonumber \\
  &\negthinspace \negthinspace + \negthinspace\text{Cov}_{\mathbf{\theta}}(\frac{ \partial L_{\text{No-BP}}}{\partial \theta_\mu}, \frac{ \partial L_{\text{BP}}}{\partial \theta_\mu}) +\text{Cov}_{\mathbf{\theta}}(\frac{ \partial L_{\text{No-BP}}}{\partial \theta_\mu}, \frac{ \partial L_{\tilde{\rho}}}{\partial \theta_\mu}) \nonumber \\
  &\negthinspace \negthinspace + \negthinspace\text{Cov}_{\mathbf{\theta}}(\frac{ \partial L_{\text{BP}}}{\partial \theta_\mu}, \frac{ \partial L_{\tilde{\rho}}}{\partial \theta_\mu}).
\end{align}

By construction we know that $L_{\text{No-BP}}$ does not lead to barren plateaus and we can assume that $\text{Var}_{\mathbf{\theta}}(\frac{ \partial L_{\text{No-BP}}}{\partial \theta_\mu}) \sim \frac{1}{\text{poly}(n)}$ does not decay exponentially. Hence the overall quantum ResNet model has gradient variance $\text{Var}_{\mathbf{\theta}} \left( \frac{\partial \langle O \rangle_{\rho, {\mathbf{\theta}}}}{\partial \theta_\mu} \right)  \sim \frac{1}{\text{poly}(n)}$ and therefore does not exhibit barren plateaus. The shallow depth component has mitigated barren plateaus in the overall model. This is demonstrated experimentally in Figure~\ref{fig:bp_evidence} which shows that the gradient variance $\text{Var}_{\mathbf{\theta}} \left( \frac{\partial \langle O \rangle_{\rho, {\mathbf{\theta}}}}{\partial \theta_\mu} \right)$ for the overall quantum ResNet model decays sub-exponentially.

 \begin{figure}[h]%
\centering
\includegraphics[width=0.8\linewidth]{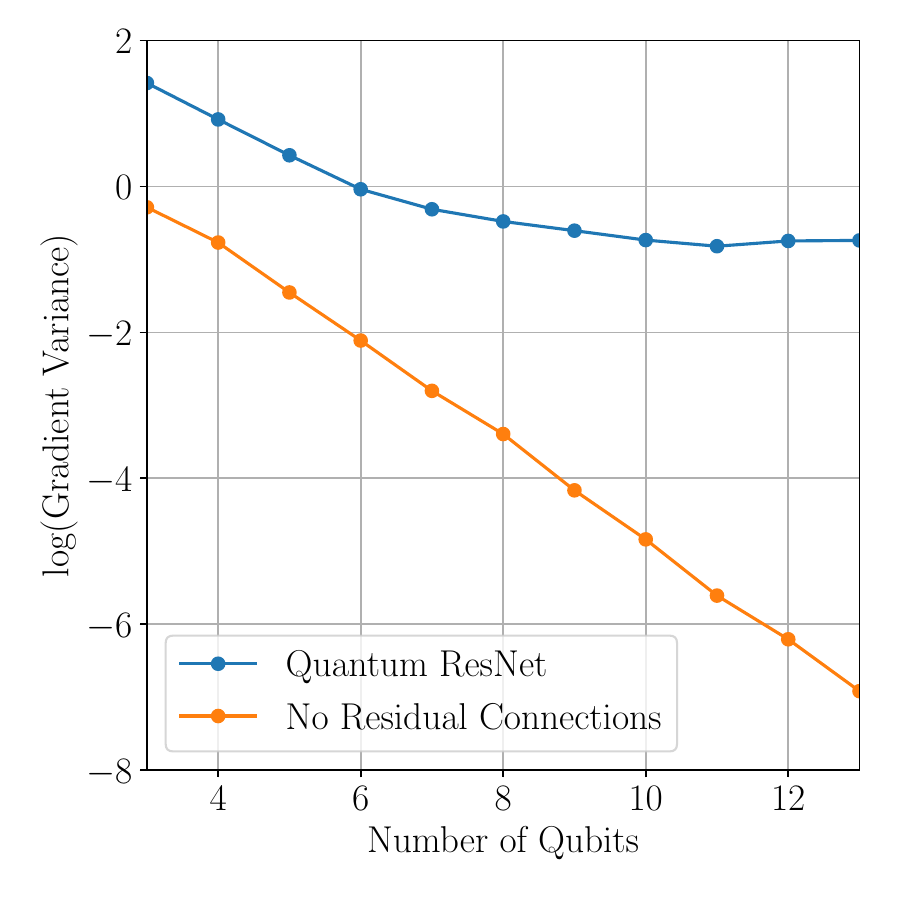}
\caption{The logarithm of the gradient variance as the number of qubits is increased. (Quantum Residual Connections) The model in which one residual layer is implemented between the input and output of the $W_1$ portion of the BP ansatz which consists of $W_2 W_1$. (No Residual Connections) The model with no residual connections is an ansatz which is known to exhibit barren plateaus and corresponds to the operator $W_2 W_1$. In this example $W_1$ was generated by $[XY, YX, YZ]$ leading to an exponential DLA dimension \cite{wiersema2023classification} while $W_2$ was generated by $[XY]$ leading to a polynomial DLA dimension \cite{wiersema2023classification}. Both $W_1$ and $W_2$ consisted of 25 repeated sublayers of parameterised gates. Each sublayer consisted of parameterised two-qubit gates that acted on all adjacent qubits of the form $\prod_{k=1} \big( (\bigotimes_{j = \text{odd}} e^{i \theta_j H^{(k)}_{j,j+1}})( \bigotimes_{j' = \text{even}} e^{i \theta_{j'} H^{(k)}_{j',j'+1}} ) \big)$ where $H^{(k)}_{j,j+1}$ is the $k$-th element in the specified set of generators which acts on the qubits indexed by $j$ and $j+1$. Variational parameters $\theta \in \mathbf{\theta}$ were sampled over the range $\theta \in [0, 2\pi]$ and 1000 samples were taken in total. The initial state $\rho_0 = \ket{0}^{\otimes n}\bra{0}^{\otimes n}$ and the final observable was $O = Y \otimes Y \otimes \mathbb{I}^{n-2}$. }\label{fig:bp_evidence}
\end{figure}

\subsubsection{Classical Simulation}

It has recently been reported that a broad class of VQC models that avoid barren plateaus are also classically simulatable  \cite{cerezo2023does}, limiting their potential for quantum advantage. While quantum residual connections offer a means to circumvent barren plateaus, they may still lead to scenarios where the model is well-approximated by the shallow component of the loss function $L_{\text{No-BP}}$, which may be classical simulatable \cite{cerezo2023does}. By construction we expect the $L_{\text{BP}}$ to exponentially vanish, and if the same happens to the $L_{\tilde{\rho}}$ term then this implies that the entire quantum ResNet model could be approximated efficiently on a classical computer within a given error. This raises the question: Are quantum ResNet models that avoid barren plateaus classically simulatable in general? If we take the assumption that $L_{\text{BP}}$ exponentially vanishes and $L_{\text{No-BP}}$ is classically simulatable, this means the answer depends on the characteristics of the non-unitary term $L_{\tilde{\rho}}$. If the $L_{\tilde{\rho}}$ term does not decay exponentially and cannot be classically simulated, then this condition could lead to a situation in which a model with at least one residual connection could be difficult to simulate, and simultaneously avoid barren plateaus. This non-unitary term, has not yet been characterised in general and therefore highlights a possible search space for future research, however, we note that this may not necessarily be a common situation. We show in Figure~\ref{fig:avg_abs_loss} that for the particular quantum ResNet architecture we consider in our numerics, that $L_{\tilde{\rho}}$ exponentially decays with an increasing number of qubits, which suggests that this specific architecture may be vulnerable to classical simulation. However, we only consider a limited example and note that there is a large space which remains unknown and is open to further work. Generalised quantum operations of the form $\sum_i \alpha_i U_i$, which are implemented by the LCU procedure, have been shown to form a convex set and the extreme points of this set are the unitary operations \cite{gudder_duality}. This means that the space of operators that can be considered is much larger than that which has previously been studied in the literature, and a full characterisation of this general space is beyond the scope of this current work.

 \begin{figure}[h]%
\centering
\includegraphics[width=0.8\linewidth]{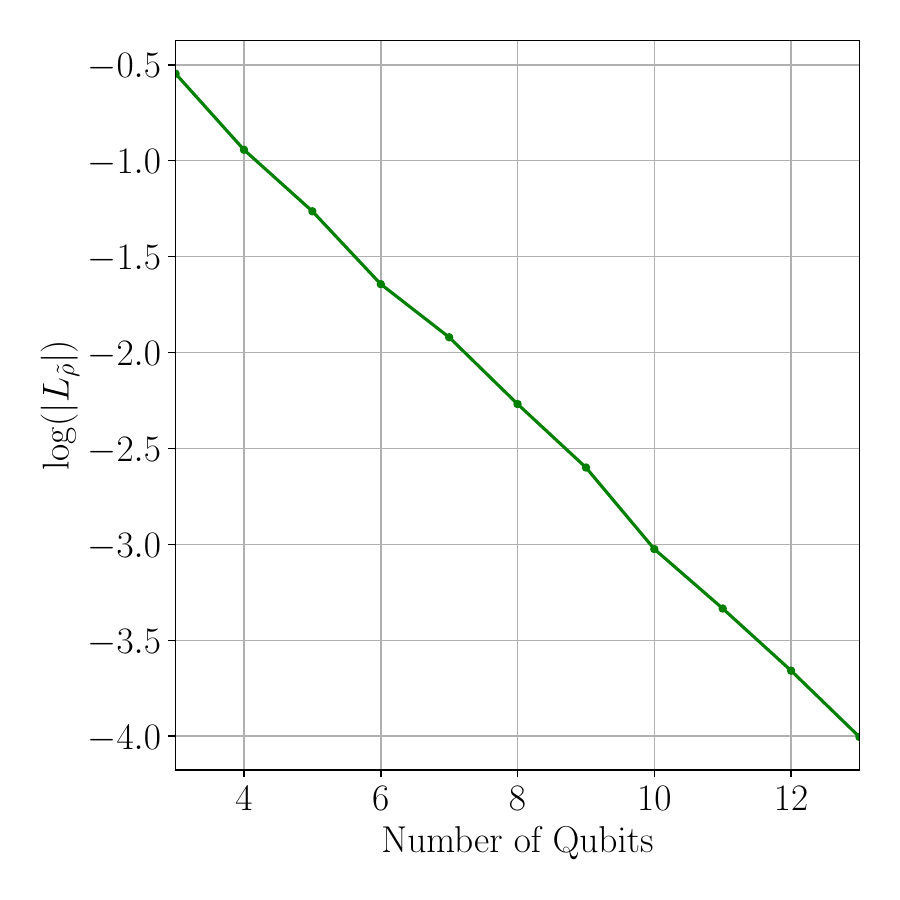}
\caption{The logarithm of the average magnitude of the non-unitary term $L_{\tilde{\rho}}$. The average was performed over 1000 initialisation. Data is from the same experiment specified in Figure~\ref{fig:bp_evidence}. In this particular setting the non-unitary contribution vanishes exponentially, suggesting it may be possible to approximate the model by $L_{\text{BP}} + L_{\text{No-BP}}$. If the $L_{\text{No-BP}}$ loss vanishes exponentially too, then the entire quantum ResNet could be approximated by $L_{\text{BP}}$. }\label{fig:avg_abs_loss}
\end{figure}

\subsubsection{Quantum ResNet as an Ensemble}

In the following portion we use the term layer to refer to quantum ResNet layers, such that $L$ corresponds to the number residual connections and hence the number of controlled unitaries in the model. Each individual controlled unitary $W_l$ is a variational ansatz, which in many cases is composed of repeated layers of a certain gate architecture; to avoid confusion, we will refer to the repeated layers within these unitary $W_l$ terms as sublayers. We will now denote the unitary in the $l$-th residual network layer as $W_l^{(\lambda)}$, where $\lambda$ specifies the number of repeated sublayers in the ansatz corresponding to that particular unitary.

%It has been reported that a wide class of VQCs that do not exhibit barren plateaus are classically simulatable \cite{cerezo2023does}, limiting the potential for quantum advantage. For example, it has been shown that for certain classes of VQC models, a likely requirement for trainability (the absence of barren plateaus) is that the dynamical Lie Algebra (DLA) dimension of the VQC is polynomial in size \cite{fontana2023adjoint, ragone2023unified}. Simultaneously, many VQC models with a polynomially sized DLA have also been shown to be classically simulatable \cite{goh2023liealgebraic}, bringing into question their potential for quantum advantage. Beyond seeking quantum advantage for its own sake, it has also recently been demonstrated that polynomial sized DLA VQC models of this type are vulnerable to ``weak privacy breaches" in which expectation value snapshots of quantum-encoded confidential data can be retrieved from the training gradients shared between nodes in a distributed learning setting; allowing potential attack agents to learn information about confidential input data \cite{heredge2024prospects}. The quantum ResNet framework allows for the integration of non-simulatable and trainable ans\"{a}tze within a single quantum device operation. This approach may be particularly advantageous for privacy-related challenges, where the inclusion of non-simulatable ans\"{a}tze may help obscure data without impacting model performance as can be the case when applying random noise.

Previous work on the characterisation of barren plateaus makes assumptions that the circuits involved are sufficiently deep \cite{ragone2023unified, fontana2023adjoint}. As the quantum ResNet allows layer skipping, this means that there may always be shallow depth circuit components in the final loss function, which will not satisfy this assumption. We now consider a deep circuit composed of many repeated sublayers of the same structure of variational gates. We show that it is possible with $L$ layers of quantum ResNet to effectively implement an ensemble of $2^L$ different ans\"{a}tze where the number of sublayers of the individual ansatz terms in the ensemble ranges from $1$ to $2^L$. 

The uniform ensemble quantum ResNet architecture is created using $L$ quantum residual layers by ensuring the number of repeated sublayers $\lambda$ in the $l$-th unitary $W_l^{(\lambda)}$ is given by $\lambda = 2^{l-1}$ for $1 \leq l \leq L$. To prevent a constant term in the loss function arising from the identity, we also initially apply a unitary operator consisting of a single sublayer $\lambda = 1$, without using a residual connection, which we denote as $W_0^{(1)}$. This corresponds to applying the operator

\begin{equation}
    (\mathbb{I} + W_L^{(2^{L-1})})...(\mathbb{I} + W_3^{(4)})(\mathbb{I} + W_2^{(2)})(\mathbb{I} + W_1^{(1)})W_0^{(1)},
\end{equation}
and normalising the resultant state by a factor $\frac{1}{\Omega'}$. This term expands out to give an operation consisting of a linear combination of $2^L$ terms. When unitaries multiply to create these terms, the number of sublayers add together. It is therefore clear that this procedure results in a operator 
\begin{equation}
    \sum_{\lambda' = 1}^{2^L} V^{(\lambda')},
\end{equation}
where $V^{(\lambda')}$ is a unitary corresponding to 
\begin{equation}
    V^{(\lambda')} = \Big( \prod_{k \in S_{\lambda'}} W_k^{(2^{k-1})} \Big) W_0^{(1)}.
\end{equation}
Where $S_{\lambda'} \subseteq [L,1]$ is the set that indexes the $W_k^{(2^{k-1})}$ terms that multiply together to give a total number of sublayers $\lambda' = 1 + \sum_{k \in S_{\lambda'}}2^{k-1}$. The state is also normalised by a factor $\frac{1}{\Omega''}$ after the operation is applied. In total, we have $2^L$ terms where the number of sublayers in each terms is in the range $\lambda' \in [1, 2^L]$. This will lead to a loss function with $2^L$ unitary VQC model loss functions, covering sublayer depths from $1$ to $2^L$, in addition to non-unitary cross terms.

Many investigations into barren plateaus specifically require circuits to be of a certain depth, for example, requiring sufficient depth to form an approximate design over a Lie group \cite{ragone2023unified}. While bounds have been found for this depth for certain circuits, the quantum ResNet would allow an agnostic approach where many different layers can be trialled in the same model run and where shallow circuit components could guarantee the overall model does not exhibit barren plateaus. We see that for $L$ quantum ResNet layers one can implement $2^L$ different unitary terms in an ensemble, in addition to further non-unitary cross terms. 

Loss function concentration implies loss function gradient concentration, and hence this metric can be similarly used to discuss the presence of barren plateaus \cite{Arrasmith_2022}. An ensemble sum of all terms with differing numbers of sublayers may be able to avoid barren plateaus due to the presence of shallower depth components within the average ensemble, as shown in Figure~\ref{fig:4-layer-ensemble}. This reaffirms previous results showing quantum ensemble methods can avoid barren plateaus \cite{friedrich2024quantumneuralnetworkensemble}. However, the quantum ResNet will also provide additional sources of variance from the non-unitary terms in the loss function, although a detailed characterisation of these terms is beyond the scope of this work.

 Note that as the number of quantum ResNet layers $L$ increases, the number of terms grows exponentially, meaning that any individual term will have a weighting that decays exponentially due to the normalisation of the state after the operation is applied. Therefore, if $L$ is increased, one would need to ensure a sufficient amount of non-exponentially concentrating terms survive to still avoid barren plateaus. By initialising the ancilla qubits, one can vary the weightings of particular terms in the ensemble. This flexibility could allow the creation of models that can be pushed close to the boundary of barren plateaus but still remain in the trainable region. Similarly to the arguments made for the previous model, if the non-unitary terms vanish, then we would expect the loss function to be on average well approximated by the shallow non-vanishing terms which do not exhibit barren plateaus individually (excluding rare potential configurations in which deep terms still provide non-zero contributions). If these shallow terms turn out to be classically simulatable \cite{cerezo2023does} then this means that the overall quantum ResNet will be classically simulatable to within some given error. Efforts to avoid this situation should focus on finding settings in which the non-unitary terms do not exponentially vanish; it remains an open question whether this is possible.

 \begin{figure}[h]%
\centering
\includegraphics[width=0.95\linewidth]{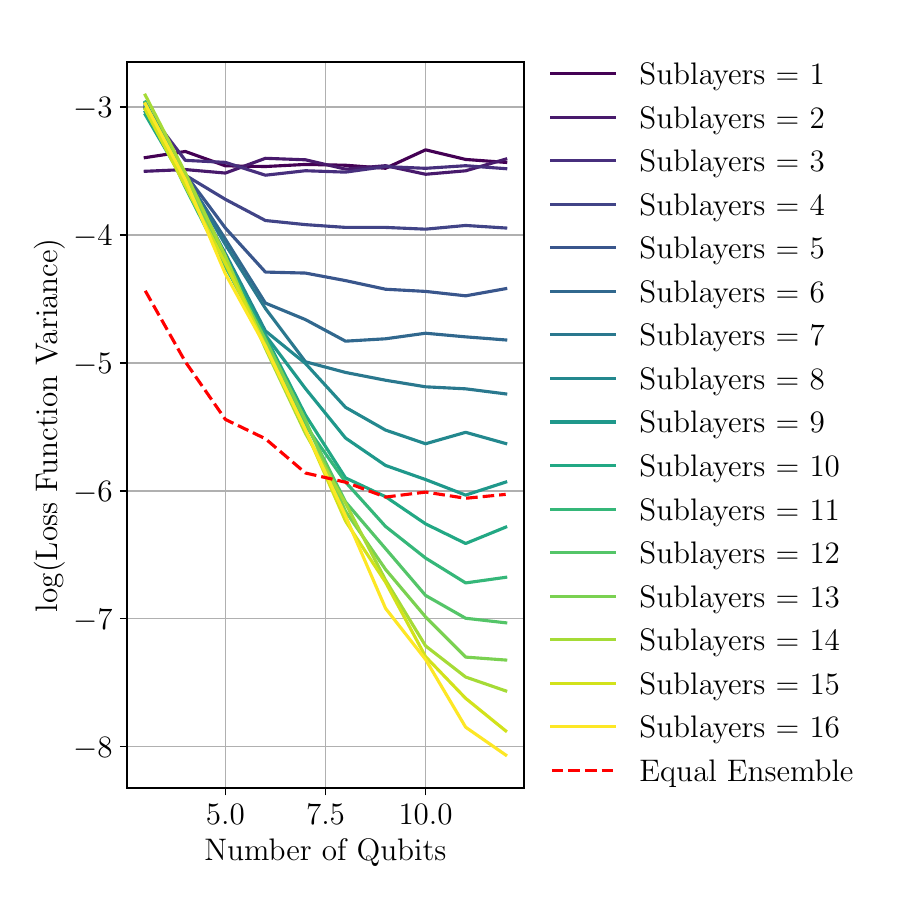}
\caption{The logarithm of the variance of the loss function for differing numbers of sublayers of the same ansatz architecture as the number of qubits increases. The variance was calculated over 500 random initialisations of the variational parameters. The ansatz was created using the generators $[XY, YX, YZ]$ as specified in Figure~\ref{fig:bp_evidence}. The red dashed line corresponds to the variance of a normalised average of all the loss functions.}\label{fig:4-layer-ensemble}
\end{figure}

Previously we showed in Equation~\ref{eq:betaprob} that the strength of the residual connection determined a lower bound of the probability of success of a given layer. We also noted that the probability of overall success decays exponentially in the number of residual connections $L$. We show in Figure~\ref{fig:Ensemble_Scaling_Beta} the impact of these effects on the expected number of repeated attempts required to achieve a successful LCU procedure. While the success probability decays exponentially in $L$ it allows the construction of ensembles of size $2^L$, which therefore leads to a linear relationship between expected attempts and ensemble size.

 \begin{figure}[h]%
\centering
\includegraphics[width=0.95\linewidth]{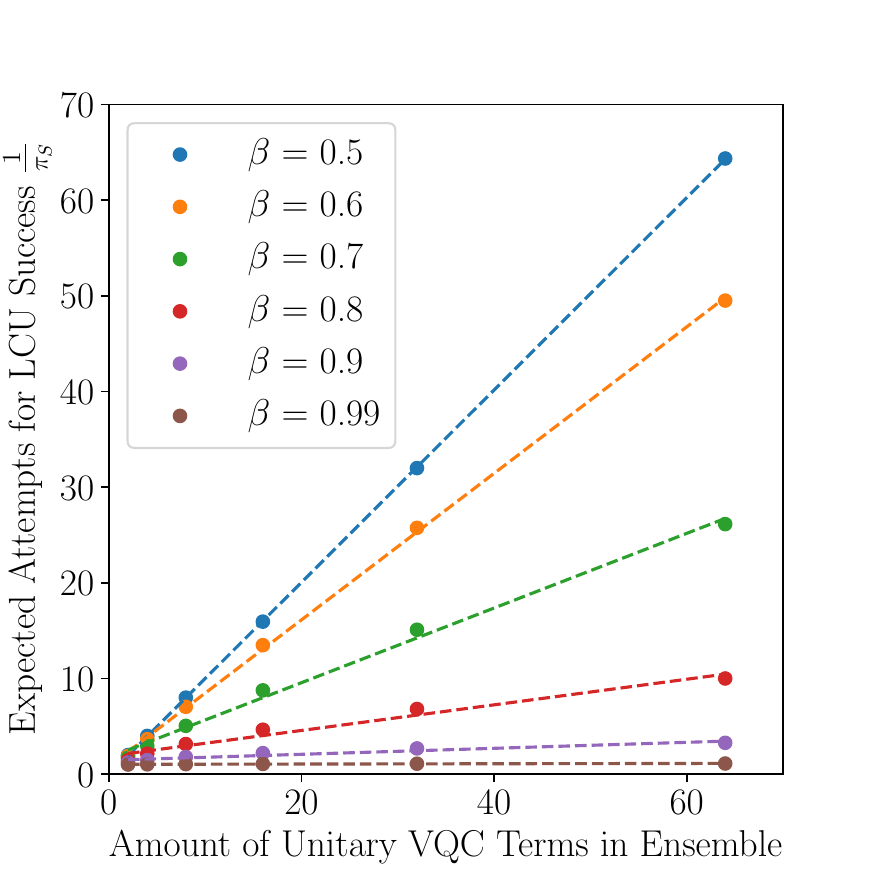}
\caption{The expected repeated attempt to implement one successful LCU procedure (equal to the reciprocal of the overall success probability $\frac{1}{\pi_S}$) plotted against the total amount of unitary VQC terms in the ensemble for various fixed residual connection strengths $\beta$. The ensemble will contain $2^L$ terms where each term corresponds to a unitary VQC model loss function, where the number of sublayers ranges in increments of $1$ from $\lambda = 1$ to $\lambda = 2^L$. In these numerics we start with an initial state $\rho = \ket{0} \bra{0}$. At each residual connection layer $l$ we calculate the probability of success of that layer $\pi_l$ and then evolve the state by the operator $(I + W^{(\lambda)}_l)$, where $W^{(\lambda)}_l$ consists of $\lambda = 2^l$ sublayers. In this case we chose each individual sublayer to be a Haar random unitary. The $\beta$ terms correspond to fixing $\beta_l = \beta, \forall l$ such that every residual connection has the same strength parameter. Due to the symmetry around $\beta = 0.5$ we restricted to plotting $0.5 \leq \beta < 1$ for readability}\label{fig:Ensemble_Scaling_Beta}
\end{figure}
 
We also note that in the case that a given residual layer provides an exponentially vanishing contribution to the overall loss, from both its unitary and non-unitary contributions, then this could potentially be used to improve the probability of success scaling of the algorithm by allowing failed implementations of that particular layer to be accepted as successes. In a scenario in which multiple ans\"{a}tze are being tested simultaneously, in which one may not know which contain barren plateaus or not, it would be possible to ignore any failures for ans\"{a}tze which do not contribute meaningfully to the loss. This means that the probability of successful implementation would decay exponentially with the number of useful quantum ResNet layers implemented $L'$ that provide non-vanishing contributions to the loss, rather than the total number of layers $L$. The implication of this however, is that one could simply discard these unused terms completely from the model and restrict to $L'$. In this case an ensemble of $2^{L'}$ models could then be manually constructed, which may be able to approximate the quantum ResNet results. Therefore, any benefit from future research in this direction would be focussed on the ability to search through a wide range of ans\"{a}tze at the same time, and find those that provide meaningful contributions, rather than any quantum advantage inherent in the model performance itself. Although in the worst case settings the cost related to the probabilistic nature of the LCU procedure may render the quantum ResNet infeasible.

Although we show that quantum ResNet frameworks can avoid barren plateaus, it is possible that many common architectures may approach a solution which is vulnerable to classical simulation. This would occur when the quantum ResNet is well approximated by an ensemble of classically simulatable shallow depth terms. Perhaps this situation should be expected, as it has previously been observed that classical deep ResNets can behave as ensembles of shallow networks, where a previous study reported that the gradient of a $110$ layer ResNet was dominated by contributions from paths of depth between $10-34$ layers \cite{veit2016residualnetworksbehavelike}. We again highlight that the key to further advancements in the quantum ResNet implementation will likely rest on constructing models in which the non-unitary terms do not exponentially vanish and remain hard to simulate classically. This provides an expanded search space compared to strictly unitary models which could be subject to a more in depth characterisation in future work. Overall the question of whether this ensemble-creating property of quantum ResNet can be used effectively therefore remains a question for further research, although the reported success of VQC ensemble models \cite{friedrich2024quantumneuralnetworkensemble} may motivate further investigation into the efficient creation of ensembles within a single quantum device.

\section{Average Pooling Layers}\label{sec:cnn}

\subsection{Quantum Average Pooling Implementation}

In a Convolution Neural Network (CNN) it is common to find a pooling layer. This layer acts to reduce the dimension of the data by considering a tile of fixed size that passes over the data and performs some operation, such as averaging, on all the datapoints in the tile. 

A classical average pooling layer consists of a pooling window of size $D \times D$ pixels that passes over an image consisting of $N \times N$ pixels, with a certain stride $K_{\text{STRIDE}}$ \cite{RawatCNN, SHARMA2018377, SaffarCNN}. We shall focus on the case where the pooling window moves across the image one pixel at a time, corresponding to a stride value $K_{\text{STRIDE}} = 1$. For each position of the pooling window, all pixels within the pooling window are averaged together, and this average is output as the value of a pixel in a new image. In classical techniques depending on the size of $K_{\text{STRIDE}}$ and $D$ the number of output pixels will be some number less than $N^2$ and hence the pooling corresponds to reducing the dimension of the image. In the quantum technique, there is no penalty for calculating all $N^2$ terms as the operations are done in parallel on the quantum state, hence we shall consider this case and leave the dimensionality reduction as a subroutine that can be implemented after the averaging. 

We consider an image of $N \times N$ pixels where each pixel is indexed by $i, j$ and the pixel colour value corresponds to $v_{i,j}$. If we focus on the average pooling window which is $D \times D$ in dimension, then classically we wish to redefine each pixel label by
\begin{equation}
  v_{i,j}' = \frac{1}{D^2} \sum_{\Delta x=0}^{D-1}\sum_{\Delta y=0}^{D-1} T_{\Delta x, \Delta y}(v_{i,j}),
\end{equation}
where $T_{\Delta x, \Delta y}$ is a pixel translation action which acts on the $i, j$ indices of $v_{i,j}$ to retrieve the values other pixels through
\begin{equation}
   T_{\Delta x, \Delta y}(v_{i,j}) = v_{i + \Delta x,j + \Delta y},
\end{equation}
so that they can be added to the average. The pixels accessed by $T_{\Delta x, \Delta y}$ define the pooling tile. For occasions in which the pooling window covers pixels outside the image, it is common to include some padding of the image \cite{Yu2023EfficientMP}, for example, any pixels outside the $N \times N$ region can be padded by zeros, giving a padded image of dimension $(N+2D) \times (N+2D)$.

There have been many investigations focused on finding quantum parallels to CNN models, whereby measurements are taken of certain qubits such that the dimensionality of the data is reduced. For example, in some setups the pooling layer corresponds to a variational ansatz followed by a measurement of a qubit \cite{Hur_2022}. In \cite{Cong_2019} an effective measurement was chosen corresponding to the symmetry and details of their specific problem concerning Quantum Phase Estimation. Although these popularised methods achieve the goal of pooling through reducing the dimensionality of the quantum state, they have not focused on giving a replication of an averaging pooling layer from classical CNN models for image data. In this section, we shall therefore demonstrate how LCU methods can be used to perform average pooling on quantum encoded image data. We shall utilise amplitude encoding with distinct coordinates $x$ and $y$ for the image, although any encoding for the data could be used, as long as the correct transformations are implemented with the LCU method.

Take an $N \times N$ image sample $(x_i, y_j, v_{i,j})$ where $v$ represents the pixel colour value and $x_i$, $y_j$ are the index coordinates of the pixel with $i, j \in [1, N]$. We shall utilise real amplitude encoding by defining a quantum state with two registers as 
\begin{equation}
  \ket{\psi} = \sum_{x_i,y_j} \frac{v_{i,j}}{\Omega} \ket{x_i}\ket{y_j} ,\label{eqn:ampl_encod}
\end{equation}
where $\Omega^2 = \sum v_{i,j}^2$ is a normalisation constant and $v_{i,j}\in\mathbb{R}$. Once we have an image encoded into a quantum state in this manner, we can consider the type of operation required by average pooling. In a simple case, we can consider only the $x$ register for $D=2$. We will need to apply an operation that takes $\ket{x} \rightarrow \frac{1}{\Omega'}(\ket{x} + \ket{x \ominus 1} )$ for some normalisation constant $\Omega'$, to perform an averaging over pixels amplitudes. This will result in the state $\ket{x}$ having its amplitude transformed to the sum of the amplitudes of the $\ket{x}$ and $\ket{x \oplus 1}$ states before being normalised (and therefore averaged). This operation is implemented by $A_{\text{POOL}} = \frac{1}{2}(I + \text{SUBTRACT}_1)$, where $\text{SUBTRACT}_1 \ket{x} = \ket{x \ominus 1}$ and $\text{ADD}_1 \ket{x} = \ket{x \oplus 1}$. These operators are implemented by decrement and increment gates. We can see that by taking
\begin{equation}
  A_{\text{POOL}}^\dagger A_{\text{POOL}} \negthinspace = \negthinspace \frac{1}{4}(2\mathbb{I}\negthinspace +\negthinspace \text{ADD}_1 \negthinspace +\negthinspace \text{SUBTRACT}_1) \negthinspace \negthinspace \neq \negthinspace \mathbb{I},
\end{equation}
that it is not a unitary operation. Hence, to implement average pooling on a quantum device, we utilise LCU methods.

 \begin{figure}[h]%
\centering
\includegraphics[width=1\linewidth]{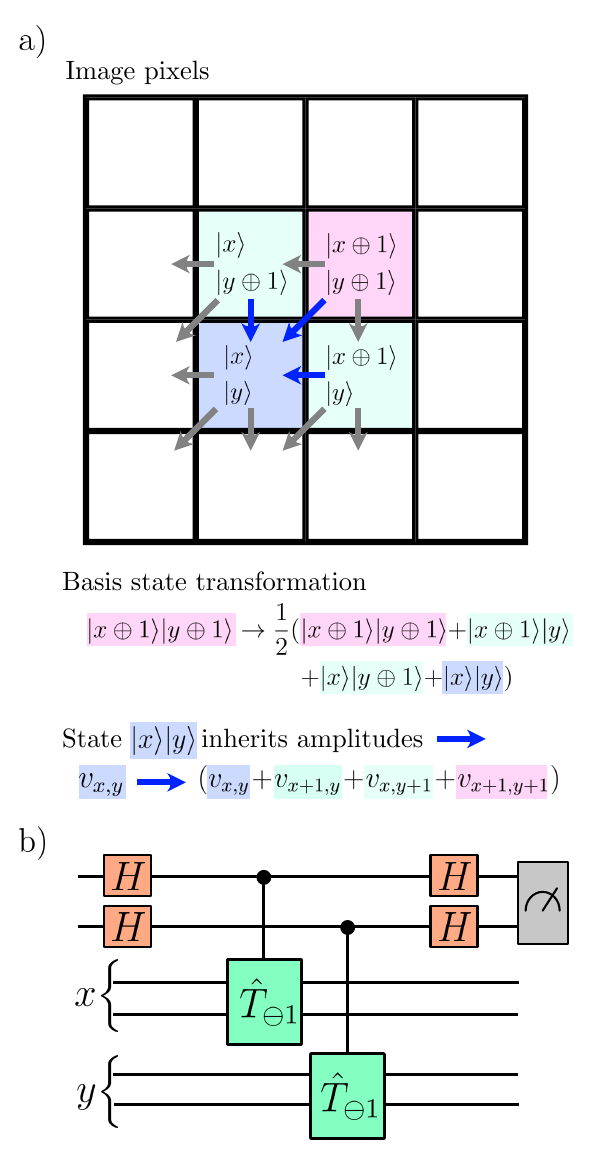}
\caption{a) Example of a $2 \times 2$ pooling window on a $4 \times 4$ pixel image indicating the desired transformation to perform average pooling. b) The circuit that implements the average pooling in this example through application of controlled $\hat{T}_{\ominus 1}$ operators to both registers. }\label{pooling_layer}
\end{figure}

\begin{theorem}[Quantum Average Pooling]
  It is possible to probabilistically implement an averaging pooling layer on an image that has been amplitude encoded into a quantum state as specified in Equation~\ref{eqn:ampl_encod}, such that the average pooling operation $A_{\text{POOL}}$ has the effect 
  \[ A_{\text{POOL}} \sum_{x_i,y_j} \frac{v_{i,j}}{\Omega} \ket{x_i}\ket{y_j} =
 \sum_{x_i,y_j} \frac{v'_{i,j}}{\Omega'} \ket{x_i}\ket{y_j}, \]
where \[ v_{i,j}' = \frac{1}{D^2} \sum_{\Delta x=0}^{D-1}\sum_{\Delta y=0}^{D-1} v_{i + \Delta x,j + \Delta y} , \]
by using the LCU framework from Section~\ref{sec:lcu-framework}. The terms $\Omega$ and $\Omega'$ are the normalisation constants for the initial and final states, respectively.
\end{theorem}

\begin{proof}

For simplicity, we consider the ancilla qubits to be composed of two registers. We initialise these registers to the following
\begin{equation}
  P_{\text{PREP}} \ket{0}^{\otimes k} \ket{0}^{\otimes k} = \frac{1}{D} \Big( \sum_{\Delta x = 0}^{D-1} \ket{\Delta x} \negthinspace \Big) \Big( \negthinspace \sum_{\Delta y = 0}^{D-1} \ket{\Delta y} \negthinspace \Big). 
\end{equation}
In this setting $\ket{\Delta x}$ and $\ket{\Delta y}$ correspond to the computational basis states of a Hilbert space of $k$ qubits $\mathcal{H} = (\mathbb{C}^2)^{\otimes k}$, which are labelled by $\Delta x$ and $\Delta y$. In this basis $\ket{0} \equiv \ket{0}^{\otimes k}$, $\ket{1} \equiv \ket{0}^{\otimes k-1}\ket{1}$ and $\ket{2} \equiv \ket{0}^{\otimes k-2}\ket{1}\ket{0}$ continuing for all values until $\ket{2^k -1} \equiv \ket{1}^{\otimes k}$. The preparation operator therefore takes the form
\begin{align}
  P_{\text{PREP}} = & \frac{1}{D}\negthinspace \Big( \sum_{\Delta x = 0}^{D-1} \ket{\Delta x} \negthinspace \Big) \Big( \sum_{\Delta y = 0}^{D-1} \ket{\Delta y} \negthinspace \Big) \bra{0}^{\otimes k} \bra{0}^{\otimes k} \nonumber \\
  &+ \sum_{i=0}^{2^k -1}\sum_{j=1}^{2^k -1} (...) \bra{i}\bra{j} + \sum_{i=1}^{2^k -1} (...) \bra{i}\bra{0}, 
\end{align}
where $(...)$ collects terms with $\bra{i}\bra{j}$ such that $i+j \geq 1$ and will therefore not be used in the algorithm and can be ignored. This can be clearly implemented by a unitary operation. It can also be shown that $T_{\Delta x, \Delta y}$ can be realised with a unitary operator. Binary addition is possible to implement on a quantum circuit \cite{Vedral_1996} as shown in Figure~\ref{binary_addition} that shows $\text{ADD}_1\ket{x} = \ket{x \oplus 1}$. In general, we shall define a $k$ adding operator as $\text{ADD}_k \ket{x} = \ket{x \oplus k}$, which may be formed by applying $ADD_1$ in series $k$ times or through a more optimised gate specific for that value of $k$. The inverse of this operator, found by taking the Hermitian conjugate, will subsequently be the subtraction operator $\text{SUBTRACT}_k \ket{x} = \ket{x \ominus k} $.

 \begin{figure}[h]%
\centering
\includegraphics[width=0.8\linewidth]{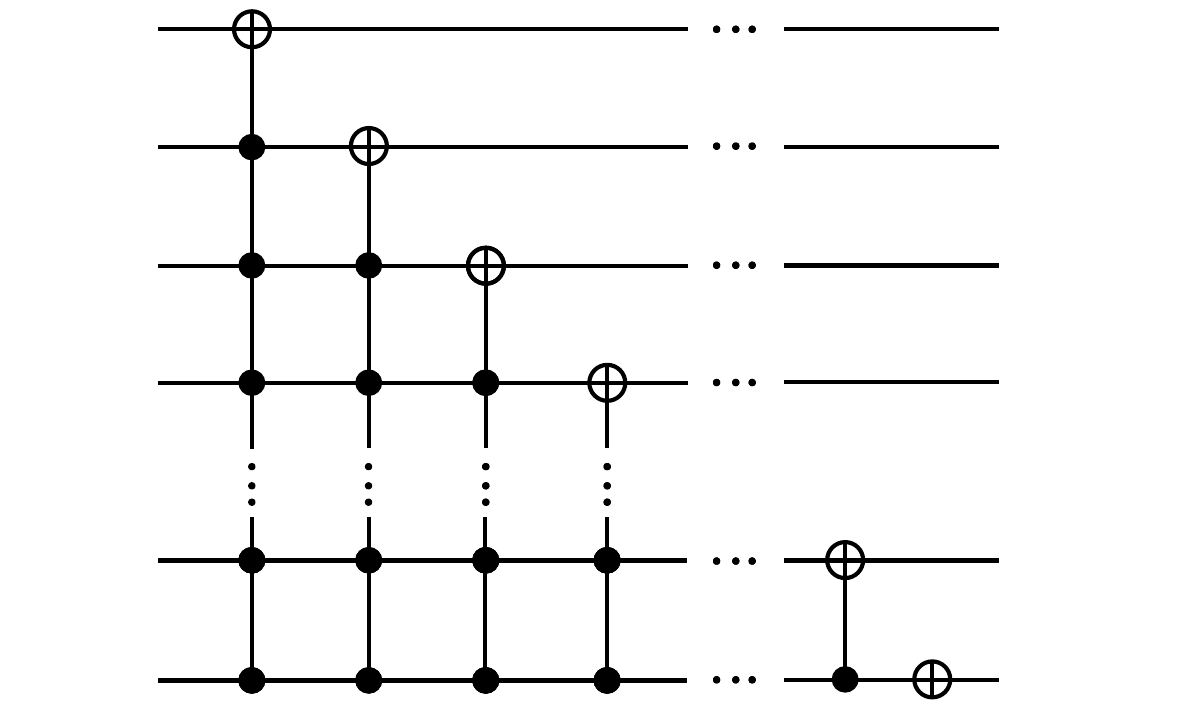}
\caption{Example circuit architecture to perform binary addition increasing all binary values by one. This takes $\text{ADD}_1 \ket{x} = \ket{x \oplus 1}$ \cite{Vedral_1996}. }\label{binary_addition}
\end{figure}

In order for $T_{\Delta x, \Delta y}(v_{i,j})$ to recover the correct amplitude, we can utilise the $\text{ADD}_k$ and $\text{SUBTRACT}_k$ operators. We can therefore define $\hat{T}_{\Delta x, \Delta y}$ through the action of these operators as
\begin{equation}
  \hat{T}_{\Delta x, \Delta y}\ket{x_i}\ket{y_j} = \ket{x_i \ominus \Delta x}\ket{y_j \ominus \Delta y} ; \Delta x, \Delta y > 0,
\end{equation}
\begin{equation}
  \hat{T}_{\Delta x, \Delta y} \ket{x_i}\ket{y_j} = \ket{x_i \oplus \rvert \Delta x \rvert }\ket{y_j \oplus \rvert \Delta y\rvert } ; \Delta x, \Delta y \leq 0.
\end{equation}
Now consider the action of $\hat{T}_{\Delta x, \Delta y}$ on the amplitude encoded image as
\begin{align}
  & \hat{T}_{\Delta x, \Delta y}\ket{\psi} = \hat{T}_{\Delta x, \Delta y} \Big( \sum_{x_i,y_j} \frac{v_{i,j}}{\Omega} \ket{x_i}\ket{y_j} \Big) \nonumber \\
  & = \sum_{x_i,y_j} \frac{v_{i,j}}{\Omega} \ket{x_i \ominus \Delta x}\ket{y_j \ominus \Delta y} \nonumber \\
  & = \sum_{x_i,y_j} \frac{v_{i,j}}{\Omega} \ket{x_{i-\Delta x}}\ket{y_{j-\Delta y}}.
\end{align}
 Relabelling the indices in the sum as $(i\rightarrow i+\Delta x)$ and $(j \rightarrow j+\Delta y)$ this can be written as
\begin{equation}
    \hat{T}_{\Delta x, \Delta y}\ket{\psi} = 
   \sum_{x_i,y_j} \frac{v_{i + \Delta x,j + \Delta y}}{\Omega} \ket{x_i}\ket{y_j}.
\end{equation}
We can now define the selection operator as
\begin{align}
  & S_{\text{SELECT}} \ket{\Delta x} \ket{\Delta y}\ket{\psi} = \ket{\Delta x}\ket{\Delta y}\hat{T}_{\Delta x, \Delta y}\ket{\psi} \nonumber \\
  & = \ket{\Delta x}\ket{\Delta y} \sum_{x_i,y_j} \frac{v_{i + \Delta x,j + \Delta y}}{\Omega} \ket{x_i}\ket{y_j},
\end{align}
which corresponds to implementing controlled $\hat{T}_{\Delta x, \Delta y}$ operators, which are respectively applied when the ancilla is in the state $\ket{\Delta x}\ket{\Delta y}$. The selection operator is therefore unitary and can be applied on a quantum circuit.

Following the LCU procedure by first preparing the ancilla qubits gives
\begin{align}
   & P_{\text{PREP}} \ket{0}^{\otimes k} \ket{0}^{\otimes k}\ket{\psi} \nonumber \\
   & = \frac{1}{D} \Big( \sum_{\Delta x = 0}^{D-1} \ket{\Delta x} \Big) \Big( \sum_{\Delta y = 0}^{D-1} \ket{\Delta y} \ \Big) \ket{\psi},
\end{align}
applying the selection operator results in
\begin{align}
     & S_{\text{SELECT}} P_{\text{PREP}} \ket{0}^{\otimes k} \ket{0}^{\otimes k}\ket{\psi} \nonumber \\
     &= \frac{1}{D} \sum_{\Delta x = 0}^{D-1} \sum_{\Delta y = 0}^{D-1} \ket{\Delta x} \ket{\Delta y} \hat{T}_{\Delta x, \Delta y} \ket{\psi}.
\end{align}
Applying the inverse preparation operator defined as
\begin{align}
  P_{\text{PREP}}^\dagger = & \frac{1}{D} \ket{0}^{\otimes k} \ket{0}^{\otimes k} \Big( \sum_{\Delta x = 0}^{D-1} \bra{\Delta x} \ \Big) \Big( \sum_{\Delta y = 0}^{D-1} \bra{\Delta y} \Big) \nonumber \\
  &+ \sum_{i=0}^{2^k -1}\sum_{j=1}^{2^k -1} \ket{i}\ket{j} (...) + \sum_{i=1}^{2^k -1} \ket{i}\ket{0}(...) , 
\end{align}
we can therefore see the state can be written as
\begin{align}
     P&_{\text{PREP}}^\dagger S_{\text{SELECT}} P_{\text{PREP}} \ket{0}^{\otimes k} \ket{0}^{\otimes k}\ket{\psi} \nonumber \\
     = & \frac{1}{D^2} \ket{0}^{\otimes k} \ket{0}^{\otimes k} \sum_{\Delta x = 0}^{D-1} \sum_{\Delta y = 0}^{D-1} \hat{T}_{\Delta x, \Delta y} \ket{\psi} \nonumber \\
     & + \sum_{i=0}^{2^k -1}\sum_{j=1}^{2^k -1} \ket{i}\ket{j} (...) + \sum_{i=1}^{2^k -1} \ket{i}\ket{0}(...).
\end{align}
We now wish to measure the ancilla qubits, discarding any states when $\ket{0}^{\otimes k} \ket{0}^{\otimes k}$ is not measured, hence we can ignore the $(...)$ terms. The probability $\pi_S$ of measuring the $\ket{0}^{\otimes k} \ket{0}^{\otimes k}$ state will therefore be
\begin{equation}
  \pi_S = \left\lvert \frac{1}{D^2} \sum_{\Delta x = 0}^{D-1} \sum_{\Delta y = 0}^{D-1} \hat{T}_{\Delta x, \Delta y} \ket{\psi} \right\rvert^2.
\end{equation}
After selecting only the cases where the ancillas are measured in the $\ket{0}^{\otimes k} \ket{0}^{\otimes k} $ state, the final state will be projected to
\begin{align}
     & \bra{0}^{\otimes k} \bra{0}^{\otimes k} P_{\text{PREP}}^\dagger S_{\text{SELECT}} P_{\text{PREP}} \ket{0}^{\otimes k} \ket{0}^{\otimes k}\ket{\psi} \nonumber \\
     & = \frac{1}{\sqrt{\pi_S}D^2}  \sum_{\Delta x = 0}^{D-1} \sum_{\Delta y = 0}^{D-1} \hat{T}_{\Delta x, \Delta y} \ket{\psi} .
\end{align}
where $\sqrt{\pi_S}$ term is required such that the state is normalised. Expanding out the operation of $\hat{T}_{\Delta x, \Delta y}$ it is possible to see that
\begin{align}
  & \frac{1}{\sqrt{\pi_S}D^2} \sum_{\Delta x=0}^{D-1}\sum_{\Delta y=0}^{D-1} \hat{T}_{\Delta x, \Delta y}\ket{\psi} 
  \nonumber \\
  & = \frac{1}{\sqrt{\pi_S}D^2} \sum_{x_i,y_j} \sum_{\Delta x=0}^{D-1}\sum_{\Delta y=0}^{D-1}\frac{v_{i + \Delta x,j + \Delta y}}{\Omega} \ket{x_i}\ket{y_j} 
   \nonumber \\
   & = \sum_{x_i,y_j} \frac{v'_{i,j}}{\Omega'} \ket{x_i}\ket{y_j}.
\end{align}
The term $\Omega' = \sqrt{\pi_S}\Omega$ is the overall normalisation factor of the final state. Hence, the average pooling operation 
\begin{equation}
v_{i,j}' = \frac{1}{D^2} \sum_{\Delta x=0}^{D-1}\sum_{\Delta y=0}^{D-1} v_{i + \Delta x,j + \Delta y},
\end{equation}
has been implemented on a quantum circuit.
\end{proof}

Through this framework the average pooling operation can be implemented as is visualised for a $2 \times 2$ pooling window in Figure~\ref{pooling_layer}. 

In terms of the image boundaries, it would be possible to apply null padding to the image by introducing extra qubits in the $\ket{0}$ states. Without any padding, the quantum pooling window will treat the image as periodic and could include pixels from opposite sides of the image when located close to the image edge. It is worth mentioning that classical CNNs perform this averaging as a means of dimensionality reduction, as the total number of averages taken is usually less than the total number of pixels in the image. In the quantum case, however, it takes no additional time to average every single pixel in the image, and therefore we presented this as the most general case. Subsequently, to truly perform the quantum analogue of pooling, certain values would be discarded to reduce the dimensionality of the problem. The exact method will depend on exactly what kind of dimensionality reduction is desired, although we discuss how this could be performed in Appendix~\ref{apn:dimensionality-reduction-pooling}.

\subsection{Algorithm Scaling}

The advantage of average pooling layers is well documented in classical neural networks \cite{Galanis2022ConvolutionalNN, gholamalinezhad2020pooling, Zafar2022}. The method suggested in our work would allow this feature to be implemented on quantum encoded images in quantum machine learning models. The utility of this in the context of a quantum advantage will therefore be dependent on whether other quantum operations within the model can achieve some advantage over classical models. For example, recent empirical studies for amplitude encoded image data suggest that quantum models may in some circumstances display an improved robustness against adversarial attacks compared to classical methods \cite{west2023benchmarking,west2023drastic}. Indeed, theoretical guarantees on quantum robustness have recently been reported \cite{dowling2024adversarial}. While this is not the focus of our work, the quantum average pooling layers we introduce may be of use in further studies regarding quantum encoded image data. We will therefore focus instead on examining the circuit complexity and probabilistic scaling of the method.

\subsubsection{Circuit Scaling}

In order to implement an average pooling layer with a pooling window size of $D \times D$ one needs to implement a linear combination $D^2$ unitaries via the LCU method, whereby these unitaries correspond to selecting all pixels within the average pooling window. This means that in a given $x / y$ register one needs to create a linear combination of $D$ unitaries. A general technique for doing this using only $\log(D)$ ancilla for any LCU circuit would be to use $D$ multi-controlled operators, whereby each ancilla basis state activates exactly one of the multi-controlled operators and no others i.e. $S_\text{SELECT} \ket{i}\ket{\psi} = \ket{i}\hat{T}_i \ket{\psi}$. Previous related work on implementing general filter masks for quantum convolutional layers used this approach where they state that the number of multi-controlled unitaries that must be implemented is equal to the size of the filter mask $D^2$ \cite{wei2021quantumconvolutionalneuralnetwork, chen2022novelarchitectureparameterizedquantum}. However, we show that by focusing on average pooling and considering that subtraction operators form a closed set, that a more efficient circuit implementation is possible which uses $\mathcal{O}(\log(D))$ single-qubit controlled operators.

Without loss of generality consider only the $x$ register. We need to ensure that this has an equal superposition of $\hat{T}_{\Delta x}$ terms such that $\Delta x$ runs from $0$ to $D - 1$. This can be achieved using $\mathcal{O}(\log(D))$ qubits and single qubit controlled applications of $\hat{T}_{k}$ operators. In this case $\hat{T}_{k}$ is implemented by the $k$-th ancilla qubit and corresponds to the $\hat{T}_{k} = \text{SUBTRACT}_{2^{k-1}}$ operation. This operation subtracts $2^{k-1}$ from the binary register and can be formed by the $\text{SUBTRACT}_1$ operation applied to the $n - k + 1$ most significant qubits and the identity operator applied to the $k - 1$ least significant qubits. If $k \in [1, L]$, where $L$ is the total number of ancilla qubits and controlled operations then we have a set of operators $T = [\hat{T}_{1}, \hat{T}_{2}, \hat{T}_{4},...,\hat{T}_{2^{L-1}}]$ where the $k$-th element is applied if the $k$-th ancilla is in the 1 state. As we have an equal superposition of all $2^L$ basis states in the ancillas, this means that we will find an equal superposition of $2^L$ terms which consist of all possible combinations of the operators in $T$. The ancilla basis states can be thought of as selecting subsets of $T$, with the resultant operation found by multiplying all operators together in that subset such that the indices decrease from left to right. As we have chosen the scaling of these elements to be $\hat{T}_{k} = \text{SUBTRACT}_{2^{k-1}}$, we can see that considering all possible combinations of these terms will result in an equal sum of $\hat{T}_{\Delta x}$ operators where $\Delta x$ runs from $0$ to $2^L - 1$. 

We can then set $D = 2^L$ to see that the number of ancillas and controlled operators must scale as $L = \mathcal{O}(\log(D))$. Therefore, the procedure requires $\mathcal{O}(\log(D))$ unitaries that correspond to subtraction operators, that act on the separate $x$ and $y$ target registers and are controlled by $\mathcal{O}(\log(D))$ ancilla registers. An example circuit which involves 4 ancilla qubits and 4 controlled unitaries is shown in Figure~\ref{fig:pooling_layer_expanded} in which 16 total combinations of operators are implemented. Note that if $D$ is less than $2^L$ then the final layer could be adjusted to be $\hat{T}_{D-2^{L-1}}$ but we may require multi-control qubits in order to prevent degeneracy and maintain and equally weighted combination of operators, see Appendix~\ref{apnd:degeneracy-prevention}.

\begin{figure}[h]%
\centering
\includegraphics[width=1\linewidth]{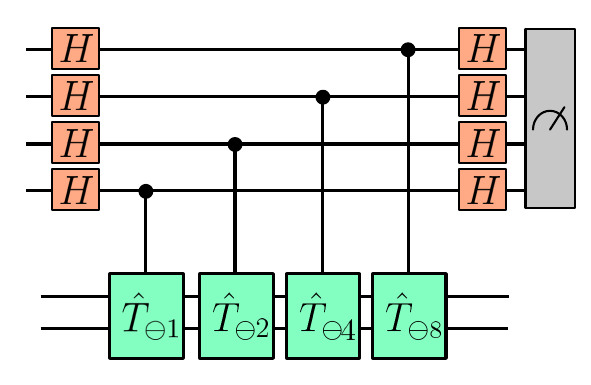}
\caption{Example efficient LCU circuit which implements a linear combination of 16 operators in the range $[\mathbb{I}, \hat{T}_{1}, \hat{T}_{2},...,\hat{T}_{15}]$. Each operator is controlled by a single qubit and corresponds to a subtraction operator, where the amount subtracted is doubled for each successive qubit. This approach is more efficient than the most general possible implementation and is possible due to the fact that subtraction operators applied in sequence form a closed set.} \label{fig:pooling_layer_expanded}
\end{figure}

To perform the averaging classically for a stride size $K$ it would require finding the average of $D^2$ values for $N^2$ pixels, giving a scaling of $\mathcal{O}(D^2 N^2)$ in general. Therefore classical averaging requires $\mathcal{O}(D^2 N^2)$ arithmetic operations, while the quantum average pooling we suggest requires $\mathcal{O}( \log (D))$ controlled subtraction operations. The exact scaling of the controlled subtraction operations depends on if ancilla qubits can be utilised, however we show this uses at worst $\mathcal{O}( (\log (N))^2)$ basic operations (see Appendix~\ref{apn:decrementscaling}) giving a total of $\mathcal{O}( \log (D)(\log (N))^2)$ basic operations in the overall LCU method while using $\mathcal{O}( \log (N))$ target qubits for the image encodings and $\mathcal{O}( \log (D))$ ancilla qubits.

Previous work built an LCU based quantum convolutional layer using a $3\times 3$ filter of the form 
\begin{equation}\label{eqn:4by4filter}
    \begin{pmatrix}
        \omega_{00} \quad \omega_{01} \quad \omega_{02} \\
        \omega_{10} \quad \omega_{11} \quad \omega_{12} \\
        \omega_{20} \quad \omega_{21} \quad \omega_{22} 
    \end{pmatrix}
\end{equation}
such that the filter passes over the image and calculates $\sum \omega_{\Delta x \Delta y} v_{i+\Delta x, j + \Delta y}$ for the pixel value in the window \cite{wei2021quantumconvolutionalneuralnetwork}. They show that average pooling can be obtained as a special case where $\omega_{ij} = 1 ;  \forall i,j$. We note that this work did not explicitly show a proof for a general $D \times D$ window, instead considering only $D=3$. They used multi-controlled operators in which each ancilla qubit state only implements exactly one unitary, meaning that their technique requires $D^2$ multi-controlled operators. They report their unitary operations as consisting of $\mathcal{O}( (\log (N))^6)$ basic operators, which would correspond to an overall scaling of $\mathcal{O}( D^2 (\log (N))^6)$. Hence, the circuit implementation that we present, which scales as $\mathcal{O}( \log (D)(\log (N))^2)$, corresponds to a polynomial advantage in $N$ and an exponential advantage in $D$ over this previous technique when considering average pooling specifically.

Furthermore, although we started with the more specific case of average pooling, our entire algorithm can be generalised to recreate the quantum convolutional filters as presented in \cite{wei2021quantumconvolutionalneuralnetwork}. This is achieved by introducing amplitudes into the ancilla qubits instead of using an equal superposition. We see this by considering a new preparation operator in which all ancillas (for both $x$ and $y$ registers) are entangled such that
\begin{align}
   & P_{\text{PREP}} \ket{0}^{\otimes k} \ket{0}^{\otimes k}\ket{\psi} \nonumber \\
   & = \frac{1}{D} \Big( \sum_{\Delta x = 0}^{D-1} \sum_{\Delta y = 0}^{D-1} \sqrt{\omega_{\Delta x, \Delta y}} \ket{\Delta x}\ket{\Delta y} \ \Big) \ket{\psi},
\end{align}
where for simplicity we say $\sqrt{\omega_{\Delta x, \Delta y}} \in \mathbb{R}$ is a real amplitude of the ancilla basis state $\ket{\Delta x}\ket{\Delta y}$. This will mean that the terms in the linear combination of $\hat{T}_{\Delta x , \Delta y}$ operators will have an associated weighting equal to $\omega_{\Delta x, \Delta y }$. Therefore the transformation on the pixel value amplitudes in the image will be 
\begin{equation}
v_{i,j}' = \frac{1}{D^2} \sum_{\Delta x=0}^{D-1}\sum_{\Delta y=0}^{D-1}  \omega_{\Delta x, \Delta y } v_{i + \Delta x,j + \Delta y}.
\end{equation}
We can therefore identify that this performs a general quantum convolutional filter as originally described in \cite{wei2021quantumconvolutionalneuralnetwork}, while utilising exponentially fewer unitary operators in terms of $D$. A comparison of this reduction in circuit complexity is shown in Figure~\ref{fig:convolutionalcomparison} when applying a $4 \times 4$ convolutional filter defined as
\begin{equation}
    \begin{pmatrix}
        \omega_{00}\quad \omega_{01}\quad \omega_{02}\quad \omega_{03} \\
        \omega_{10}\quad \omega_{11}\quad \omega_{12}\quad \omega_{13}\\
        \omega_{20}\quad \omega_{21}\quad \omega_{22}\quad \omega_{23}\\
        \omega_{30}\quad \omega_{31}\quad \omega_{32}\quad \omega_{33} 
    \end{pmatrix}
\end{equation} 
This general convolution case reduces to the previously discussed averaging case when the ancilla qubits are in an equal superposition. We highlight that in this section we do not discuss the actual pooling step itself, and hence any improvement is confined to the averaging / convolutional subroutine of the circuit. Further details on the practical implementation of the pooling step and quantum convolutional neural network framework can be found in \cite{wei2021quantumconvolutionalneuralnetwork}.

% In this work it was shown that $4$ ancilla qubits and $9$ multi-controlled unitaries can implement a $3 \times 3$ convolutional window. They suggest that a $D\times D$ filter would require $D^2$ multi-controlled unitaries.

\begin{figure}[h!]%
\centering
\includegraphics[width=0.94\linewidth]{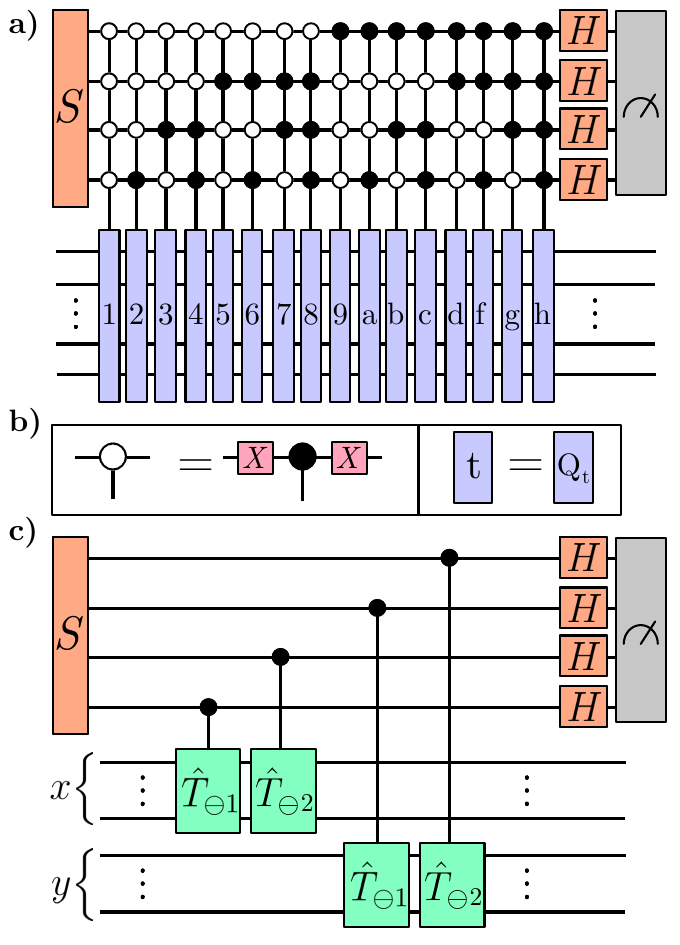}
\caption{ a) The circuit used to implement quantum convolutional filters in previous work \cite{wei2021quantumconvolutionalneuralnetwork}. For an equal comparison we show the generalisation of their method to allow the implementation of a $4 \times 4$ size convolutional window, using $16$ multi-controlled unitaries. b) White circles correspond to controls that activate if the control qubit is in the $\ket{0}$ state. The $t$ labelled controlled-gates apply the unitary operator $Q_t$ which when combined using the LCU technique reconstruct the desired convolutional filter, as outlined in \cite{wei2021quantumconvolutionalneuralnetwork}. A complete definition of the $Q_t$ unitaries within their framework is extensive and omitted here. For both circuits $S$ is the unitary used to initialise the ancilla qubit amplitudes corresponding to the desired filter, which is then achieved when the ancillas are all measured in the $\ket{0000}$ state at the end of the circuit. c) The circuit implementation of the same $4 \times 4$ convolutional filter using our framework as described in the main text. Only a single qubit controls each $\hat{T}_{\ominus t}$ which acts to subtract $t$ in binary on the $x$ or $y$ register. This arithmetic interpretation simplifies the circuit considerably as we can see each ancilla basis state applies one of $16$ different transformations $\hat{T}_{\Delta x, \Delta y}$ on the pixel registers, which provides all the transformations required to apply a $4 \times 4$ size convolutional window, using only 4 controlled unitary gates in total, rather than having to implement $16$ terms individually as multi-controlled unitary gates. For example $\ket{1111}$ will be responsible for $\Delta x = 3, \Delta y = 3$, hence the ancilla amplitude $\alpha_{1111}$ should control the relative weight of the top right pixel of the convolution filter ($\omega_{03}$ in Equation~\ref{eqn:4by4filter}), while $\alpha_{0000}$ corresponds to the identity and would control the bottom left pixel weighting ($\omega_{30}$ in Equation~\ref{eqn:4by4filter}).} \label{fig:convolutionalcomparison}
\end{figure}

We also note that our implementation is easily generalised to data types with dimension greater than two by adding additional registers for extra dimensions. For $Q$-dimensional data we require $Q \log(N)$ qubits to encode the image into $Q$ different registers. In this case, the window has a total $Q$-dimensional volume of $D^Q$, and hence using the implementation from previous work \cite{wei2021quantumconvolutionalneuralnetwork} would require $D^Q$ unitaries. In contrast, our technique would require $Q \log(D)$ operators in total with $Q \log(D)$ ancilla total qubits when considering all registers and would therefore use exponentially fewer unitaries in both data dimension $Q$ and window size $D$.

\subsubsection{Probabilistic Scaling}

When considering the fact that this is a probabilistic algorithm with a chance of failure, the situation becomes slightly more complicated, as the probability will depend on the image itself. As shown previously, the probability of success is equal to
\[ \pi_S = \left\lvert \frac{1}{D^2} \sum_{\Delta x = 0}^{D-1} \sum_{\Delta y = 0}^{D-1} \hat{T}_{\Delta x, \Delta y} \ket{\psi} \right\rvert^2. \]
If we had an image in which every pixel has the exact same colour value, we can effectively set $\hat{T}_{\Delta x, \Delta y} = \mathbb{I}$ and therefore $\pi_S = 1$. We expect real-world images to have pixels that are close in colour value to other local pixels nearby on average with the exceptions occurring at the edges of subjects within the image. We therefore expect the probability to remain relatively stable overall, although one could construct adversarial example images that result in low probabilities. To gain a practical understanding of this probability scaling we considered the MNIST fashion dataset \cite{deng2012mnist}, we empirically show the scaling of the probability of success with respect to $D$ and $N$ in Figure~\ref{fig:success_prob_L} and Figure~\ref{fig:success_prob_N}, respectively. These results indicate that the probability of success decreases but levels off after a certain point as $D$ increases, and that there is no discernible trend when increasing $N$. These results align with our intuition for real-world images. The general structure and content of real-world images (and their local pixel similarities) remain consistent as $N$ increases; therefore, there should be no discernible decrease in $\pi_S$ from increasing $N$. On the other hand, increasing $D$ means considering a larger neighbourhood of pixels. Initially, as $D$ increases, the probability of encountering edges within this neighbourhood also increases, causing $\pi_S$ to decrease. However, beyond a certain point, further increasing $D$ adds less new information because most of the image regions are already accounted for and edges covered, so $\pi_S$ levels off. The exact probability will be very image dependent; there may also exist techniques to prepare images such that the probabilities are improved, which we leave as an open question for further research. In cases where the probability of success remains high, it corresponds to pixels being similar to each other locally, this in turn means that the average pooling action on those qubits locally will be well approximated by the identity operation. It is important to note therefore that these conditions could be more prone to classical simulation techniques and further work should be carried out on this possibility before any concrete advantages are claimed.

 \begin{figure}[h]%
\centering
\includegraphics[width=1\linewidth]{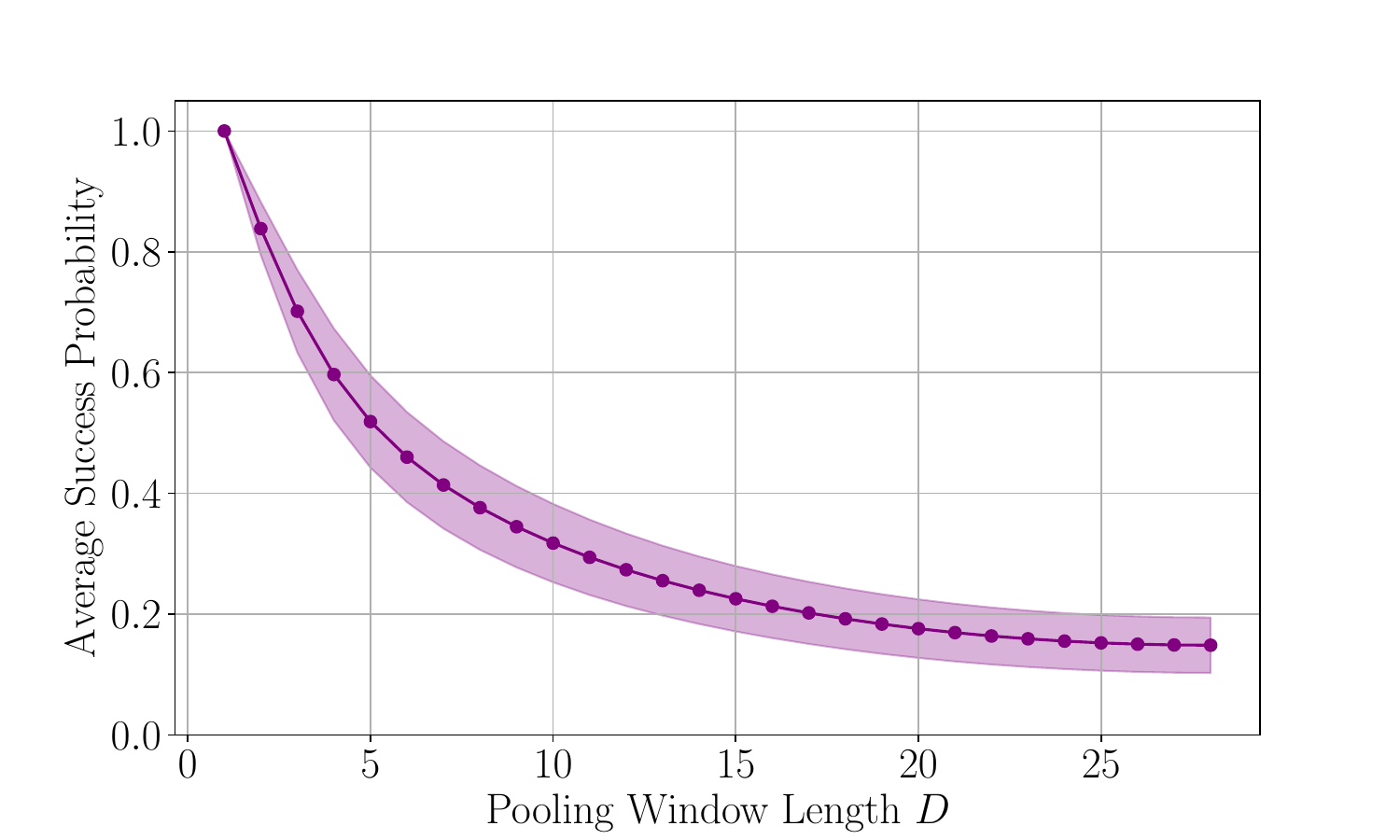}
\caption{Average success probability of the LCU procedure for implementing an average pooling layer as the pooling window dimension $D$ is increased, while the image size is set to $N=28$. The average is taken over 100 image samples from the MNIST database \cite{deng2012mnist}, with the shaded area indicating the standard deviation of the samples.} \label{fig:success_prob_L}
\end{figure}

 \begin{figure}[h]%
\centering
\includegraphics[width=1\linewidth]{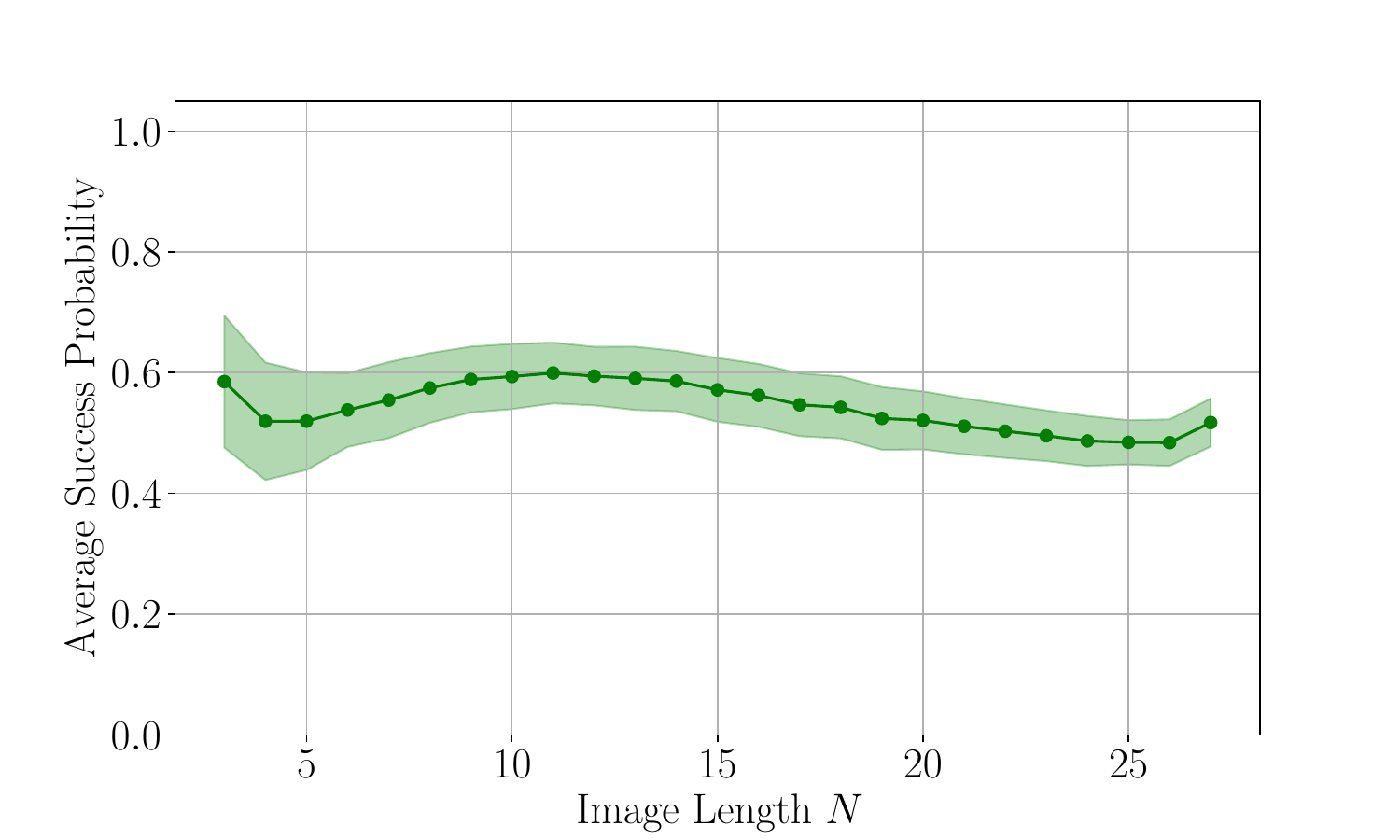}
\caption{Average success probability of the LCU procedure for implementing an average pooling layer as the image size $N$ is increased, while the pooling window size is kept constant at $D=3$. The average is taken over 100 image samples from the MNIST database \cite{deng2012mnist}, with the shaded area indicating the standard deviation of the samples.}\label{fig:success_prob_N}
\end{figure}

\section{Irreducible Subspace Projections}\label{sec:irrep_projections}

In order to discuss the irreducible subspace projections circuit, we first give a brief overview of representation theory definitions.

\subsection{Representation Theory Preliminaries}

Following standard texts on group and representation theory \cite{Fulton1991RepresentationTA, ragone2023representation} we provide a brief summary and introduction to the underlying mathematical framework utilised in this section.

The definition of a representation of a group $G$ can be given as 
\begin{definition}[Representation]
A representation of a group $G$ refers to a pair $(U, V)$, where $V$ is a vector space and $U$ is a group homomorphism $U: G \rightarrow GL(V)$, where $GL(V)$ is the general linear group of invertible matrices that acts on the vector space $V$ of the representation. The homomorphism maps group elements $g\in G$ to matrices $U_g \in GL(V)$.
\end{definition}
Often either the homomorphism $U$, the vector space $V$ or the image subgroup $U_G \subseteq GL(V)$ can be referred to as the ``representation" depending on context. 
\begin{definition}[Subrepresentation]
For a given representation $(U, V)$ a subrepresentation is a vector subspace $W \subseteq V$ for which any vector $w \in W$ remains within $W$ when acted on by any representation of the group, that is, $U_g w \in W, \forall g \in G$.
\end{definition}
\begin{definition}[Irreducible Representation]
   If a representation does not contain any non-trivial subrepresentations (the trivial spaces being the empty subspace $\{ 0 \}$ and the entire space $V$) then it is said to be an irreducible representation. 
\end{definition}
If a representation does contains a non-trivial subrepresentation, then it is called a reducible representation. Representations can often be decomposed into a direct sum of their irreducible representations, in which case they are called completely reducible. Maschke's Theorem \cite{maschke_ueber_1898} states that any finite-dimensional representation of a finite group $G$ over a field $F$ will be completely reducible (so long as the field characteristic does not divide the order of the group). It is also the case that every finite-dimensional representation of compact Lie groups is completely reducible \cite{hall_lie_2015}. If a representation $(U, V)$ is completely reducible, then there exists a basis in which we can write $U_g$ as 
\begin{equation}
  U_g \rightarrow \bigoplus_{r = 1}^R\bigoplus_{j=1}^{m_r}u^{(r)}_g = \bigoplus_{r=1}^R u^{(r)}_g \otimes \bm{I}_{m_r},
\end{equation}
where $(u^{(r)}, V_r)$ are the irreducible representations indexed by $r \in [1, R]$, where $R$ is the total number of irreducible representations. In this work we will often refer to $r$ as the irreducible representation for simplicity, although strictly it is the label that indexes the irreducible representations $(u^{(r)}, V_r)$. The quantity $m_r$ is called the multiplicity of the irreducible representation. We can also define the degree of a representation as $n_r \equiv \dim(V_r)$. In this basis $U_g$ is in a block-diagonal form. This change of basis also decomposes the representation vector space $V$ as
\begin{equation}
  V \rightarrow \bigoplus_{r =1}^R V_r \otimes \mathbb{C}^{m_r}.
\end{equation}
The conjugacy class of a group and the character of a representation for a given group element can be defined as

\begin{definition}[Conjugacy Class]
  For a group $G$ two given elements $g, h \in G$ are said to be conjugate if $\exists x \in G$ such that
\[ g = x h x^{-1}, \]
and hence the conjugacy class can be defined as  \[ C_{(g)} = \{ x g x^{-1} \mid x \in G \}, \]
  corresponding to the set of all elements that are conjugate to $g$.
\end{definition}

\begin{definition}[Character]
  For a representation $(U, V)$ of the group $G$, the character can be defined as \[ \chi_U(g) = \textup{Tr}\big( U_g\big), \]
  which is the trace of the representation of $g$ on $V$. For the irreducible representations $(u^{(r)}, V_r)$ indexed by $r$ we shall define $\chi_{u^{(r)}}(g) \equiv \chi_r(g)$ for simplicity. 
\end{definition}
In particular, as the trace permits cyclic permutation, this means that $\chi_r(xgx^{-1}) = \chi_r(g)$ and therefore $\chi_r(g)$ is constant on the conjugacy classes of $G$. i.e., if $g_1, g_2 \in C_{(g_1)}$ then $\chi_r(g_1) = \chi_r(g_2)$. It is also worth noting that as the identity element corresponds to an identity matrix, taking the trace of an identity matrix will result in the fact that $\chi_r(I) = \dim (V_r)$. Hence, the degree of an irreducible representation can be found by $n_r = \dim (V_r) = \chi_r(I)$.

We will now introduce a key result of this work regarding projecting quantum encoded states to any linear combination of irreducible subspaces of a given finite group.

\subsection{Subspace Projection Circuit}

We start by observing a result from the representation theory of finite groups which states \cite{serre_linear_reps}
\begin{theorem}[Irreducible Subspace Projection]
Let $U: G \rightarrow GL(V)$ be a representation of $G$. The canonical decomposition into an irreducible representation is given by $V = V_1 \oplus V_2 \oplus ... \oplus V_R$, where the irreducible representations have characters $\chi_1,...,\chi_R$ and degrees $n_1,...,n_R$. Then the projection $\hat{P}_r$ of $V$ onto the space $V_r$ is given by:
\begin{equation}
  \hat{P}_r = \frac{n_r}{\rvert G \rvert} \sum_{g_i \in G} \chi_r(g_i)^* U_{g_i},
\end{equation}
where $\chi_r(g_i)$ is the character of group element $g_i \in G$ for irreducible representation $r$, and $U_{g_i}$ is the matrix representation of group element $g_i \in G$ acting on the space $V$.
\label{thm:irreducible representationProj}
\end{theorem}
This is a known result in representation theory; for more details, see reference texts such as \cite{Fulton1991RepresentationTA, serre_linear_reps}. In this section, we explore a method for implementing this projection on a quantum circuit using the Linear Combination of Unitaries (LCU) technique. Specifically, we demonstrate the practicality of realising the projection described in Theorem~\ref{thm:irreducible representationProj} on a quantum device. In the most comprehensive scenario, we establish that a quantum circuit can execute combinations of such $\hat{P}_r$ projections at the same time via the LCU method. 

\begin{theorem}[Quantum Generic Irreducible Subspace Projection]
Let $\ket{\psi} \in V = (\mathbb{C}^2)^{\otimes n}$ be a quantum state encoded on $n$ qubits. Let $U: G \rightarrow SU(2^n)$ be a unitary representation of a finite group $G$, where the representatives for group elements $g_i \in G$ are denoted $U_{g_i} \in SU(2^n)$ and are unitary operators that acts on $\ket{\psi}$ and can be implemented in a quantum device. Utilising the Linear Combination of Unitaries framework on a quantum circuit, one can probabilistically apply a linear combination of irreducible representation projections
\begin{equation}\label{eqn:theorem-lcu-project}
 \sum_{r=1}^R a_r \hat{P}_r \ket{\psi} = \frac{1}{\Omega'} \sum_{r=1}^R a_r \big( \frac{n_r}{\rvert G \rvert}\sum_{g_i \in G} \chi_r(g_i)^* U_{g_i} \ket{\psi} \big),
\end{equation}
such that $\hat{P}_r$ projects $\ket{\psi} \in V$ to the subspace $V_r$ corresponding to the irreducible representation subspace indexed by $r$ with degree $n_r$. The constants $a_r$ can be freely chosen, $\Omega'$ is the normalisation constant for the final state after projection, and $R$ is the total number of irreducible representations. 
\label{thm:lcu_irrep_projection_general}
\end{theorem}

\begin{figure}[h]%
\centering
\includegraphics[width=1\linewidth]{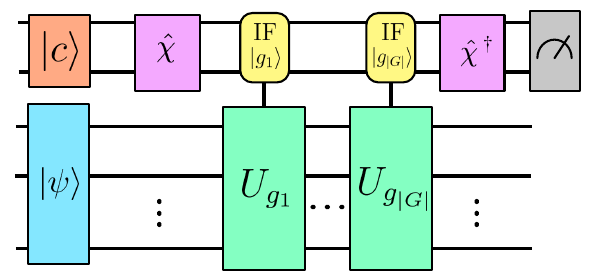}
\caption{The general irreducible representation projection circuit. The upper register consisting of $k$ qubits is prepared into a superposition state $\ket{c}$, where the basis states $\ket{b_r}$ are labelled by the irreducible representations $r$. The operator $\hat{\chi}$ then takes each irreducible representation labelled basis state to its corresponding group element labelled basis state multiplied by the correct character. Each group element labelled basis state $\ket{g_i} \equiv \ket{b_i}$ where $g_i \in G$ then controls the application of its corresponding representation on the quantum states in the lower register, represented by $U_{g_i} \in SU(2^n)$. This is achieved by using multi-qubit control gates, such that if the ancilla is in the state $\ket{g_i}$, the operator $U_{g_i}$ is applied and nothing else. This is demonstrated in the graphic using yellow boxes for these conditional control gates, which are controlled by specific multi-qubit basis states $\ket{g_i}$. These can be constructed as multi-qubit control gates along with the correct implementation of $X$ gates to ensure the controlled action $U_{g_i}$ is only applied when the ancilla register is in the correct basis state $\ket{g_i}$. We apply $\hat{\chi}^\dagger$ on the upper register and then measure and post-select for the $\ket{b_1} \equiv \ket{0}^{\otimes k}$ state. We demonstrate this will create the appropriate combinations of projections onto the desired subspaces with the proportion of each projector determined by the initialisation of the state $\ket{c}$.}\label{fig/circuit_drawing}
\end{figure}

\begin{proof}
As we intend to implement a linear combination of projections, we shall consider a slightly more general form of the LCU method. We first perform a pre-initialisation step on the $k$ ancilla qubits, initially in the basis state $\ket{b_1}$, using an operator $\hat{\gamma}$ to prepare the state $\ket{c} =\hat{\gamma} \ket{b_1}$ defined as
\begin{equation}
  \ket{c} = \frac{1}{\Omega} \sum_{r = 1}^R a_r n_r \ket{b_r},
\end{equation}
where $a_r$ can be chosen to adjust the relative amounts of a given representation, $n_r$ is the degree of the representation, and $\Omega=\sum_{r=1}^R \rvert a_r n_r \rvert^2$ is a normalisation constant for the quantum state. The states $\{ \ket{b_j} \}_{j \in [1, 2^k]} \in (\mathbb{C}^2)^{\otimes k}$ are basis states of the $2^k$ dimensional Hilbert space for the $k$ qubit ancilla register denoted by $\mathcal{H}= (\mathbb{C}^2)^{\otimes k}$. They can be taken to be computational basis states. We highlight that after this pre-initialisation the ancilla register is in a combination up to $R$ basis states, which are indexed by the irreducible representations $r$. The inclusion of the $n_r$ term here ensures that the representation weightings are correct later on. At this point, every state $\ket{b_r}$ corresponds to a different irreducible representation of the group $G$. We define $r=1$ to correspond to the trivial representation.

The next step corresponds to applying an operator that takes each representation labelled state $\ket{b_r}$ into a sum of basis states $\ket{b_i}$, $i \in [1, \rvert G \rvert]$ which are now labelled by the elements of the group $g_i \in G$, where each basis state $\ket{b_i}$ inherits a weighting according to the character $\chi_r(g_i)^*$ corresponding to the representation $r$ and the relevant group element $g_i$. To help improve readability, we define $\ket{g_i} \equiv \ket{b_i}$. These correspond to the same basis states in the ancilla register, but emphasises the fact that are labelled by the group elements $g_i$.

As we consider a combination of multiple projections at the same time, then, in order to maintain generality of our framework, instead of requiring $P_{\text{PREP},r}$ to be a different operator for each individual irreducible representation (indexed by $r$) as would be the standard case in LCU methods, we will consider a generalised preparation operator corresponding to the unitary $\hat{\chi}$ which has the effect of applying the correct character for every irreducible representation $r$ and every group element $g_i \in G$ such that
\begin{equation}
\hat{\chi}\ket{b_r} = \frac{1}{\sqrt{\rvert G \rvert}}\sum_{g_i \in G} \chi_r(g_i)^* \ket{g_i}.
\label{eqn:charoperator}
\end{equation}
where in the most general form
\begin{equation}
\hat{\chi} = \frac{1}{\sqrt{\rvert G \rvert}}\sum_{r=1}^R \sum_{g_i \in G} \chi_{r}(g_i)^* \ket{g_i}\bra{b_{r}} + \sum_{r=R+1}^{2^k}(...)\bra{b_r},
\label{eqn:charoperator-full}
\end{equation}
where the terms grouped together as $\sum_{r=R+1}^{2^k}(...)\bra{b_r}$ will not be used in the algorithm and can hence be ignored. Implementing this corresponds to constructing a matrix that contains the character of every group element for every irreducible representation. This matrix will have the form

\begin{equation}
[\hat{\chi}]_{i, j} = 
\begin{cases}
\frac{1}{\sqrt{\rvert G \rvert}}\chi_j(g_i), &\text{ if } i \in [1,\rvert G \rvert], j \in [1, R]\\
 0 & \text{ if } i \in (\rvert G \rvert, 2^k], j \in [1, R] \\
\frac{1}{\sqrt{\rvert G \rvert}} u_{i,j} &\text{ if } i \in [ 1 , 2^k], j \in (R, 2^k],
\end{cases}
\end{equation}
where the $u_{i,j}$ will never affect the algorithm, but are required as the matrices involved must be of size $2^k$. The first $R$ columns correspond to the vectors that contain the characters for each element in a given representation. Considering the character orthogonality theorem \cite{Huppert1998} which states that
\begin{equation}
  \frac{1}{\rvert G \rvert }\sum_{g_i \in G} \chi_j(g_i) \chi_k(g_i)^* = \delta_{jk},
  \label{eqn:charorththm}
\end{equation}
this means that the character vectors of irreducible representations form an orthonormal basis. The $ u_{i,r}$ terms can be chosen utilising the Gram-Schmidt procedure or otherwise to ensure that all column vectors of $\hat{\chi}$ form an orthonormal basis. A matrix whose column vectors form an orthonormal basis is unitary. Hence, our generalised preparation step $P_{\text{prep}} = \hat{\chi}$ can be implemented on a quantum circuit. 

The $S_\text{SELECT}$ operator can be implemented by applying the unitary quantum gate representation $U_g \in SU(2^n)$ to the $n$ qubit target state $\ket{\psi} \in V = (\mathbb{C}^2)^{\otimes n}$ for each corresponding group element $\ket{g_i}$ in the ancilla register
\begin{equation}
S_\text{SELECT}\ket{g_i}\ket{\psi} = \ket{g_i} U_{g_i} \ket{\psi}.
\end{equation}
As long as $G$ is a compact group there will be unitary irreducible representations $U_{g_i}$ via the Peter-Weyl theorem \cite{DeitmarWeyl}. Hence $U_{g_i}$ can be implemented in a quantum circuit for compact groups $G$. Using controlled gates $U_{g_i}$, in which the operation $U_{g_i}$ is applied only when the ancilla qubit is in the state $\ket{g_i}$, then the operator $S_\text{SELECT}$ is successfully implemented.

We have shown that the preparation and selection operators can be implemented on a quantum circuit, now all that remains is to combine the operators together to view the full action of the LCU procedure. Pre-initialising the ancillas
\begin{align}
\hat{\gamma}\ket{b_1}\ket{\psi} =\frac{1}{\Omega} \sum_{r =1}^R a_r n_r \ket{b_r} \ket{\psi},
\end{align}
followed by applying preparation operator $\hat{\chi}$ on the ancillas gives
\begin{equation}
\hat{\chi} \hat{\gamma}\ket{b_1}\ket{\psi} = \frac{1}{\Omega} \sum_{r =1}^R a_r n_r \frac{1}{\sqrt{\rvert G \rvert}} \sum_{g_i \in G} \chi_r(g_i)^* \ket{g_i}\ket{\psi}.
\end{equation}
The selection operator applied to the circuit can then be written
\begin{align}
&S_\text{SELECT} \hat{\chi} 
 \hat{\gamma}\ket{b_1}\ket{\psi} \nonumber \\
 & = \frac{1}{\Omega} \sum_{r =1}^R a_r n_r \frac{1}{\sqrt{\rvert G \rvert}} \sum_{g_i \in G} \chi_r(g_i)^* \ket{g_i} U_{g_i}\ket{\psi}.
\end{align}
 
The conjugate $\hat{\chi}$ term can be written as
\begin{equation}
\hat{\chi}^\dagger = \frac{1}{\sqrt{\rvert G \rvert}}\sum_{r=1}^R \sum_{g_i \in G} \chi_{r}(g_i) \ket{b_{r}}\bra{g_i} + \sum_{r=R+1}^{2^k}\ket{b_r}(...),
\label{eqn:charoperator-full-conj}
\end{equation}
where $\sum_{r=R+1}^{2^k}\ket{b_r}(...)$ represents terms that will not end up contributing and hence we can ignore. Applying the conjugate $\hat{\chi}^\dagger$ term, and separating out terms that will not be used further we find
\begin{align}
& \hat{\chi}^\dagger S_\text{SELECT}  \hat{\chi} 
 \hat{\gamma}\ket{b_1}\ket{\psi} = \nonumber \\
 & \frac{1}{\Omega} \negthinspace \sum_{r =1}^R \negthinspace a_r \frac{n_r}{\rvert G \rvert} \negthinspace \sum_{g_i \in G} \ket{b_1}\chi_{1}(g_i) \chi_r(g_i)^* U_{g_i}\ket{\psi} \negthinspace + \negthinspace\sum_{r=2}^{2^k}\ket{b_r}(...),
\end{align}
where we can ignore all terms $\sum_{r=R+1}^{2^k}\ket{b_r}(...)$ as they will be discarded if ever measured. Note that since $r=1$ is defined as the trivial representation, it follows that $\chi_1(g_i) =1, \forall g_i \in G$. We now measure the ancilla qubits and find that they will be found in the state $\ket{b_1}$ with a probability $\pi_S$ given by
\begin{equation}\label{eqn:probability-lcu-irrep-proj}
  \pi_S = \left\lvert \frac{1}{\Omega} \sum_{r =1}^R a_r \frac{n_r}{\rvert G \rvert} \sum_{g_i \in G} \chi_r(g_i)^* U_{g_i}\ket{\psi} \right\rvert^2.
\end{equation}
By only retaining states in which the ancilla was measured in the $\ket{b_1}$ state, we will have prepared the state 
\begin{align}\label{eqn:end_of_lcu_proof}
& \bra{b_1}\hat{\chi}^\dagger S_\text{SELECT}  \hat{\chi} 
 \hat{\gamma}\ket{b_1}\ket{\psi} \nonumber \\
 & = \frac{1}{\Omega'} \sum_{r =1}^R a_r \frac{n_r}{\rvert G \rvert} \sum_{g_i \in G} \chi_r(g_i)^* U_{g_i}\ket{\psi},
\end{align}
where $\Omega' = \sqrt{\pi_S}\Omega$ is the normalisation constant of the final state, as required. 
\end{proof}

A simple corollary regarding the special case in which during the pre-initialisation stage $a_{r'=r} = 1$ and $a_{r' \neq r} = 0$, such that only one irreducible representation is selected.

\begin{corollary}[Quantum Individual Subspace Projection]
The projection given in Theorem~\ref{thm:irreducible representationProj} can be probabilistically implemented on an $n$ qubit quantum state $\ket{\psi} \in V = (\mathbb{C}^2)^{\otimes n}$ by utilising a Linear Combination of Unitaries in a quantum circuit where
\begin{equation}
  \hat{P}_r \ket{\psi} = \frac{n_r}{\Omega' \rvert G \rvert} \sum_{g_i \in G} \chi_r(g_i)^* U_{g_i} \ket{\psi},
\end{equation}
such that $\hat{P}_r$ projects $\ket{\psi}$ from the space $V$ to an irreducible representation subspace $V_r$. 
\end{corollary}

This shows that the LCU method can be adapted through careful selection of the preparation and selection operators to reproduce the projection onto any combination of irreducible representations. An overall schematic of the representation projection circuit is shown in Figure~\ref{fig/circuit_drawing}. Note that after the projection, the state may no longer be a pure state.

\subsubsection{Circuit Scaling}

In order for $k$ ancilla qubits to have a basis state for each group element representation of $G$ it would require $2^k \geq \rvert G \rvert $. Hence, the number of qubits required scales as $k \approx \mathcal{O}(\log (\rvert G \rvert))$. However, the character implementation unitary $\hat{\chi}$ which is used in the ancilla preparation stage has a dimension that scales with $\mathcal{O}(\rvert G \rvert)$ which may be prohibitive in certain cases. For example, if $G$ is taken to be the permutation group $G = S_n$, where the representation operators $U_{\sigma_i}, \sigma_i \in S_n$ correspond to permuting $n$ qubits in the target state, then we require $k \approx \mathcal{O}(\log (n!)) \approx \mathcal{O}(n \log (n))$ ancilla qubits, however we are required to initialise a unitary of size $\mathcal{O}(n!)$, which may be difficult. A proposed solution outlined in Appendix~\ref{apn:improvement} is to utilise a smaller version of $\tilde{\chi}$ corresponding to the character table where each ancilla qubit state represents a conjugacy class rather than individual group elements. In this case, the size of the unitary required could be significantly reduced, as it would scale only with the number of conjugacy classes, denoted by $\rvert \mathcal{C} \rvert$. Hence, this alternative implementation would require ancilla qubits of order $\mathcal{O}(\log (\rvert \mathcal{C} \rvert))$ to encode the character table state, but then it would also require $\rvert \mathcal{C} \rvert$ extra registers, one per conjugacy class, with sufficient qubits to produce a superposition of states for every element in the conjugacy class. This means more qubits are required in total, but the unitary is easier to implement. The number of conjugacy classes in $S_n$, for example, that is approximately approached in the asymptotic limit is
\begin{equation}
  \rvert \mathcal{C}_n \rvert \approx \frac{1}{4n \sqrt{3}} \exp \Big( \pi \sqrt{\frac{2n}{3}} \Big) \text{ ; } n \rightarrow \infty ,
\end{equation}
which is super-polynomial but sub-exponential \cite{Fulton1991RepresentationTA}. This would be a significant improvement in scaling compared to $\mathcal{O}(n!)$. This allows a trade-off between total qubits and the ease of unitary preparation. Note that in this section, we introduced a very general framework, and improvements in circuit efficiency to implement $\hat{\chi}$ may well be possible for specific groups, representations, and data types.

\subsubsection{Probabilistic Scaling}

The probability of success $\pi_S$ is shown in Equation~\ref{eqn:probability-lcu-irrep-proj} which is equal to the probability of measuring the ancillas in the $\ket{b_1}$ state. This can be simplified by utilising the result of Theorem~\ref{thm:lcu_irrep_projection_general} to observe that
\begin{align}
 \pi_S = \left\lvert \frac{1}{\Omega} \sum_{r =1}^R a_r \hat{P}_r \ket{\psi} \right\rvert^2.
\end{align}
Any quantum state can be written as a decomposition in terms of its components on the irreducible subspaces
\begin{equation}
  \ket{\psi} =\bigoplus_{r=1}^R \bigoplus_{j=1}^{m_r} \ket{\psi_r}_j,
\end{equation}
where $m_r$ is the multiplicity of representation $r$. We define $\ket{\psi_r} \equiv \bigoplus_{j=1}^{m_r} \ket{\psi_r}_j$ and see that it denotes the component of $\ket{\psi}$ that occupies the subspaces of the irreducible representation $r$. After $\ket{\psi}$ is projected by $\hat{P}_r$ onto the irreducible subspace of $r$ it will equal $\ket{\psi_r}$. Noting that $\ket{\psi_r}$ will be orthogonal for different $r$ values, we can therefore write
\begin{align}
 \pi_S = \frac{1}{\Omega^2} \sum_{r =1}^R \rvert a_r \rvert^2 \langle \psi_r \rvert \psi_r \rangle.
\end{align}
Therefore, the probability of success depends on the amount that the initial state $\ket{\psi}$ occupies the relevant irreducible subspaces in the projection. The normalisation condition in the general state means that $\langle \psi \rvert \psi \rangle = \sum_r \langle \psi_r \rvert \psi_r \rangle = 1$. In the case that we project fully to one specific space $a_r = 1, a_{r' \neq r} =0$ then 
\begin{align}
 \pi_S = \langle \psi_r \rvert \psi_r \rangle,
\end{align}
which will depend on how much of the state $\ket{\psi}$ lies in the subspace of irreducible representation $r$. If $\ket{\psi}$ lies fully within the space, then $\pi_S = \langle \psi_r \rvert \psi_r \rangle = 1$. Conversely, if $\ket{\psi}$ has no components in the subspace for irreducible representation $r$ then $\pi_S = \langle \psi_r \rvert \psi_r \rangle = 0$ and the algorithm is impossible to run. An advantage of the additional control introduced by the $a_r$ parameters is that we can freely choose the weightings to potentially improve the probability of success. In general, the probability of success will depend greatly on data encoding and choice of representations.

\subsection{Symmetry Invariant Encodings for Point Cloud Data}

In this subsection we discuss projections to a single irreducible subspace in order to highlight how these subspaces can correspond to symmetries of the input data. However, it should be noted that these projections could take the quantum state to a polynomially sized space which may be classically simulatable, as has previously been shown to be the case with equivariant variational models \cite{anschuetz_efficient_2023, goh2023liealgebraic}. The following subsection will address this issue by instead considering all subspaces, and hence maintaining an exponentially large space overall, but allowing certain subspaces to be amplified relative to the others. Investigating single subspace projections remains useful in this context as it allows the identification of which irreducible representation subspace correspond to symmetries of the underlying data. We therefore proceed to highlight how certain irreducible representation subspaces correspond to symmetries of the underlying data when it has been encoded into a quantum state. In particular we focus on point cloud data due to its inherent permutation and rotation symmetry. 

A point cloud is a collection of three-dimensional vectors (the points) that when viewed as a collective represent an image. Amongst other applications they are often associated with computer vision algorithms, as a primary method for performing three-dimensional imaging is the use of the Light Detection and Ranging (LiDAR) system, which produces point clouds as its data output \cite{raut2023endtoend}. We can consider a point cloud $P$ as a set of three-dimensional vectors $P = \{ \mathbf{p}_1, \mathbf{p}_2 , ..., \mathbf{p}_n \}$ that overall forms an image. We can consider a form of quantum encoding in which each individual point $\mathbf{p}_i$ is encoded into a quantum state $\ket{\mathbf{p}_i}$ leading to a separable quantum state for the overall point cloud as
\begin{equation}
  \ket{P} = \ket{\mathbf{p}_1} \otimes \ket{\mathbf{p}_2} \otimes \dots \otimes \ket{\mathbf{p}_n}.
\end{equation}
The states $\ket{\mathbf{p}_i}$ could be single qubits, in which case the total number of qubits in this target register would be $n$. However, in general, each $\ket{\mathbf{p}_i}$ could be a state encoded on $t$ qubits, in which case we can consider each $\ket{\mathbf{p}_i}$ as a $2^t$-dimensional qudit and $n$ would denote the number of qudits in the system that encodes $\ket{P}$.

\subsubsection{Permutation Invariance}\label{sec:perm-invariance}

The ordering of the points within the set $P$ does not affect the overall image; however, in general a machine learning algorithm may output different results depending on the ordering of the points in the data input array. A solution to prevent this by utilising permutation symmetric encodings for point clouds was suggested in previous work \cite{heredge2023permutation} where a quantum superposition of all possible permutations was used, leading to a permutation invariant quantum encoding such that
\begin{equation}\label{eqn:heredgeperminvariantencoding}
  \ket{P}_{\text{perm}} \negthinspace = \negthinspace \frac{1}{\Omega'} \sum_{\sigma \in S_n} \ket{\mathbf{p}_{\sigma^{-1}(1)}} \otimes \ket{\mathbf{p}_{\sigma^{-1}(2)}} \otimes \dots \otimes \ket{\mathbf{p}_{\sigma^{-1}(n)}} .
\end{equation}
This previous work can be viewed as a special case of the LCU method described previously, in which the group is $S_n$ and the state $\ket{P}$ is projected to the symmetric irreducible representation subspace (in which the character for all group elements is equal to one). The representations of the group elements $\sigma \in S_n$ that act on $\ket{P}$ are denoted by $U_\sigma$ and correspond to SWAP gates that permute the constituent states $\ket{\mathbf{p}_i}$ accordingly. This projection would correspond to applying the projector
\begin{equation}
  \sum_{\sigma \in S_n} U_{\sigma} \ket{P}  = \frac{1}{\Omega'} \sum_{\sigma \in S_n} \ket{\mathbf{p}_{\sigma^{-1}(1)}} \otimes\dots \otimes \ket{\mathbf{p}_{\sigma^{-1}(n)}},
\end{equation}
providing the result proposed in Equation~\ref{eqn:heredgeperminvariantencoding} and recreating the work of \cite{heredge2023permutation} as a special case of a more general framework.

\subsubsection{Rotational Invariance}\label{sec:rotat-invariant-encoding}

The novelty of the projection technique we propose in the LCU framework is that it allows expansion to many other possible symmetries beyond permutation invariance. Here we shall outline another example, rotationally invariant encodings of point clouds that contain four points utilising an $SU(2)^{\otimes 4}$ invariant encoding. This is motivated by the fact that $SU(2)$ is the double cover of $SO(3)$, which is the group that corresponds to rotations in three dimensions, which means that each element of $SO(3)$ corresponds to exactly two elements in $SU(2)$.

In this framework, we consider point clouds in which the points $\mathbf{p}_i$ are represented in spherical coordinates $\mathbf{p}_i = (r_i, \theta_i, \phi_i)$. We shall ignore the radial component $r_i$ for the sake of simplicity, although this could be included in any practical implementation in an appropriate manner.

We propose an encoding in which for each point $\mathbf{p}_i = (r_i, \theta_i, \phi_i)$ is encoded into the following quantum state
\begin{equation}\label{eqn:rotat-inv-encoding}
  \ket{\mathbf{p}_i} = \cos \left( \frac{\theta_i}{2} \right) \ket{0} + e^{i \phi} \sin \left( \frac{\theta_i}{2} \right) \ket{1} ,
\end{equation}
such that the $\theta_i$ and $\phi_i$ angle of the point is encoded into the respective angles of a single qubit in the Bloch sphere representation.

Rotation of a point cloud corresponds to rotating every individual point by the same amount and in the same direction such that $(r_i, \theta_i, \phi_i) \rightarrow (r_i, \theta_i + \Delta \theta, \phi_i + \Delta \phi) \forall i \in [1 , n]$. Considering this in terms of the quantum encoded state above, it would correspond to the application of the same rotation $U_{SU(2)} \in SU(2)$ on each qubit $\ket{\mathbf{p}_i}$, rotating every qubit state about the Bloch sphere by the same amount. In order to have a quantum encoded point cloud that is invariant to rotations of the point cloud, it would require the following invariance
\begin{align}  & \ket{\mathbf{p}_1}  \ket{\mathbf{p}_2} \ket{\mathbf{p}_3} \ket{\mathbf{p}_4} \nonumber \\
& = U_{SU(2)} \ket{\mathbf{p}_1} U_{SU(2)} \ket{\mathbf{p}_2} U_{SU(2)} \ket{\mathbf{p}_3} U_{SU(2)} \ket{\mathbf{p}_4},
\end{align}
which would be invariance under the action $SU(2)^{\otimes 4}$ .

We consider irreducible representations of the group $S_4$, where the representations of group elements $\sigma \in S_4$ correspond to the SWAP gates $U_\sigma = \text{SWAP}_\sigma$ that permute the constituent points $\ket{\mathbf{p}_i}$ accordingly. We are able to use this because Schur-Weyl duality (see Appendix~\ref{sec:rot-invariance-schur-weyl} for a detailed discussion) provides a relation between irreducible finite-dimensional representations between the symmetric group and general linear group. This means that for the one-dimensional trivial representation of $SU(2)^{\otimes n}$ of multiplicity $m$, there will correspond an $m$-dimensional representation of $S_n$ covering the same subspace \cite{Kirby_2018, kirby_thesis, ragone2023representation}. The character table for $S_4$ is shown in Figure~\ref{tab:char_table_s4}.

\begin{table}[t!]
\centering
\caption{Character table for the group $S_4$ which corresponds to the group of all permutations of the set $\{ 1, 2, 3, 4 \}$. Each conjugacy class is characterised by the cycle structure of its permutations, where a cycle of length $k$ corresponds to a sequence of $k$ elements being permuted cyclically. The identity if represented by $(I)$, 2-cycles are represented by $(12)$ and correspond to swapping any two elements. 3-cycles and 4-cycles are represented by $(123)$ and $(1234)$, respectively. Disjoint cycles are those that affect different sets of elements and can be represented together, such as $(12)(34)$. There are five irreducible representations with their characters denoted by $\chi_1, \chi_2, \chi_3, \chi_4, \chi_5 $.
 \vspace{0.2cm}}
  \label{tab:char_table_s4}
\begin{tabular}{|>{\raggedright\arraybackslash}p{1.6cm}||>{\raggedright\arraybackslash}p{0.6cm}|>{\raggedright\arraybackslash}p{0.7cm}|>{\raggedright\arraybackslash}p{1.2cm}|>
{\raggedright\arraybackslash}p{0.9cm}|>
{\raggedright\arraybackslash}p{1.0cm}|}
\hline
 Conjugacy Class & $C_{(I)}$& $C_{(12)}$&$C_{(12)(34)}$& $C_{(123)}$ & $C_{(1234)}$ \\
 \hline
 Number of Elements &
 1 & 6 & 3 & 8 & 6 \\\hline  
 $\chi_1$& 1 & 1 & 1 & 1 & 1\\
 $\chi_2$& 1 & -1 & 1 & 1 & -1\\
 $\chi_3$& 2 & 0 & 2 & -1 & 0\\ 
 $\chi_4$& 3 & -1 & -1 & 0 & 1\\
 $\chi_5$& 3 & 1 & -1 & 0 & -1\\
\hline
\end{tabular}
\end{table}

If we utilise the irreducible representation projection procedure to project the state $\ket{\mathbf{p}_1}\otimes \ket{\mathbf{p}_2}\otimes \ket{\mathbf{p}_3}\otimes \ket{\mathbf{p}_4} $ to the subspace of irreducible representation $r = 3$, with characters $\chi_3$, then this corresponds to projecting to the basis states \cite{Kirby_2018, kirby_thesis, ragone2023representation}
\begin{align}
  \ket{d_1} = \frac{\sqrt{3}}{6}( & 2\ket{0011} + 2\ket{1100} - \ket{0101} \nonumber \\
  & - \ket{1010} - \ket{0110} - \ket{1001}),
\end{align}
 and
 \begin{equation}
   \ket{d_2} = \frac{1}{2}(\ket{0101} + \ket{1010} - \ket{0110} - \ket{1001}).
 \end{equation}

When writing the action of $(U \otimes U \otimes U \otimes U)$ in the Schur basis, as shown in Figure~\ref{fig:schurdecomp}, it can be seen that the above basis states $\ket{d_1}, \ket{d_2}$ are invariant under the action of $(U \otimes U \otimes U \otimes U)$ as they are in the subspace associated with the trivial representation of $SU(2)^{\otimes 4}$ \cite{kirby_thesis, Kirby_2018}. We are able to reach this space through projecting in $S_4$ due to Schur-Weyl duality, since the irreducible subspaces of $S_4$ and $SU(2)^{\otimes 4}$ are simultaneously block diagonalised. This means that we have effectively projected the quantum encoded point cloud state to a subspace which is invariant with respect to three-dimensional rotations of the point cloud. This is a desirable symmetry for point clouds, as they are naturally invariant to three-dimensional rotations.

In order to demonstrate this numerically, we show in Figure~\ref{overlap_rotationals} that the rotationally invariant encoding introduced by this projection gives a constant overlap equal to unity when comparing two identical but rotated point clouds. This constant overlap of unity is maintained as the degree of rotation is increased, demonstrating that the encoding is producing rotationally invariant encoded states. It is also demonstrated that, without applying this symmetry projection, a non-invariant encoding produces an overlap between identical but rotated point clouds that decreases as the magnitude of the rotation is increased. This is significant as it means that without a rotationally invariant encoding, the exact same point cloud, rotated by $\pi$ radians, would produce a quantum state that appears completely different and will give zero overlap with the original orientation of the point cloud.

 \begin{figure}[h]%
\centering
\includegraphics[width=1\linewidth]{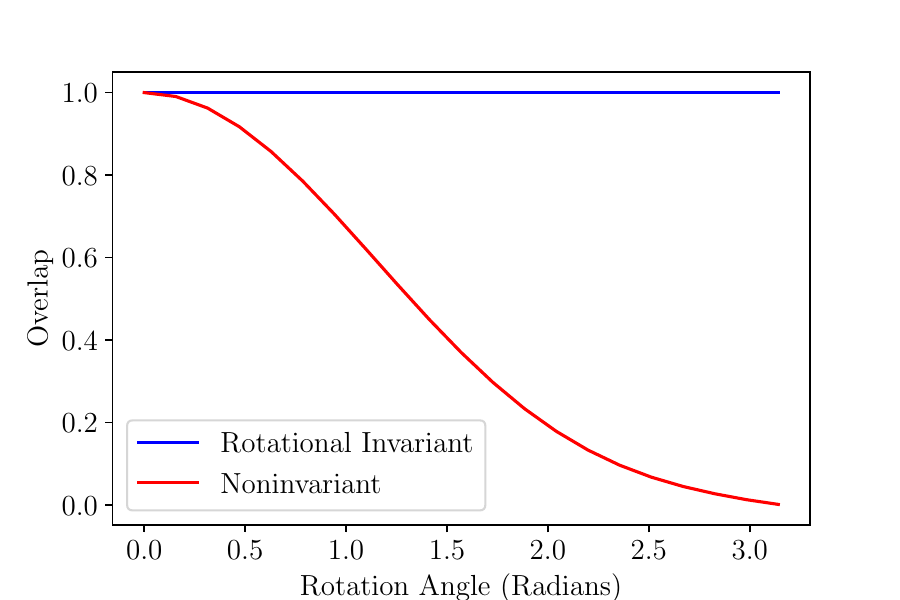}
\caption{The overlap between two identical Point Clouds that have been rotated about a random direction in three dimensions as the amount of rotation is increased, for the rotational invariant encoding and generic noninvariant encoding. To generate this plot, we randomly generated $n=4$ points on a sphere of radius $1$ to act as the point cloud. We then encoded the $(\theta, \phi)$ radial coordinates of each point into a qubit using the specification of Equation~\ref{eqn:rotat-inv-encoding}. This results in an initial state $\ket{P}_{\text{init}} = \ket{\mathbf{p}_1}  \ket{\mathbf{p}_2} \ket{\mathbf{p}_3} \ket{\mathbf{p}_4}$. 
We then repeat the procedure after the point cloud data is transformed by a three-dimensional rotation in random direction by some angle $\Theta$ to produce non-invariant state $\ket{P}_\Theta$. To generate rotationally invariant states $\ket{P_{\text{rot}}}$ we apply the result in Theorem~\ref{thm:irreducible representationProj} to project $\ket{P}_\Theta$ states to the irreducible representation denoted by $r=3$ as specified in Table~\ref{tab:char_table_s4}, which via Schur-Weyl duality lie in the trivial representation subspace of $SU(2)^{\otimes 4}$ and are hence rotationally invariant under our data encoding set-up. This means that we now have rotationally invariant encoded states $\ket{P_{\text{rot}}}_\Theta$ for all $\Theta$ values. The blue line plots the overlap $\langle {P}_{\text{init}} \rvert {P_{\text{rot}}} \rangle_\Theta$ and the red line plots the overlap $\langle {P}_{\text{init}} \rvert {P} \rangle_\Theta$ as the point cloud data rotation angle $\Theta$ is varied over $[0,\pi]$. As the overlap $\langle {P}_{\text{init}} \rvert {P_{\text{rot}}} \rangle_\Theta$ is constant we therefore numerically confirm that projection to the irreducible representation denoted by $r=3$ creates rotationally invariant states.}\label{overlap_rotationals}
\end{figure}

\subsection{Symmetric Subspace Amplification}\label{sec:subspace-amplification}

The rotationally invariant encoding discussed previously, as well as previous work on permutation invariant encodings \cite{heredge2023permutation} have been implemented in a binary fashion: the initial quantum state is either projected to a fully permutation invariant state or remains unchanged. Such an implementation significantly reduces the dimensionality of the encoding, yielding benefits in certain scenarios through improved generalisation of the model. This section explores the potential advantages of introducing a continuous spectrum of invariance with respect to some symmetry, rather than adhering to binary extremes, and investigates whether optimising this aspect as a hyperparameter could enhance classification performance in quantum machine learning (QML) models. 

Recent studies suggest that permutation equivariant variational circuits may be classically tractable under specific conditions \cite{anschuetz_efficient_2023, goh2023liealgebraic}, indicating that a drastic reduction in dimensionality might increase the likelihood of classical simulatability. This raises concerns that fully projecting to a polynomially sized subspace could render algorithms susceptible to classical simulation. To address these issues, we propose a novel approach that involves partially amplifying the portion of the quantum state into a given invariant space. Specifically, we will consider amplifying the permutation invariant subspace, although any irreducible subspace could be chosen. Apart from edge cases, the quantum state will still in general be exponential in dimension but the permutation symmetric subspace will have a higher weighting compared to all other subspaces; the amount of permutation symmetry amplification is governed by a hyperparameter $\alpha$ which will therefore also affect the expressivity of the encoding. This method provides enhanced control over training performance and facilitates the tuning of the model to achieve optimal results. For data types with inherent permutation symmetry, such as point cloud data, the optimal global classification function must also exhibit permutation invariance. However, the specific outcome of a QML model depends on the circuit architecture, classification protocol, and inherent limitations of QML models, suggesting that a purely permutation invariant QML algorithm may not necessarily be the closest to the global solution. We hypothesise that by incrementally adjusting the degree of permutation invariance in the model, we can converge towards a value that is closer to the true global optimum performance than is possible with an non-invariant, or fully permutation invariant model. This concept is illustrated in Figure~\ref{fig/explanation_figure}.

\begin{figure}[h]%
\centering
\includegraphics[width=1\linewidth]{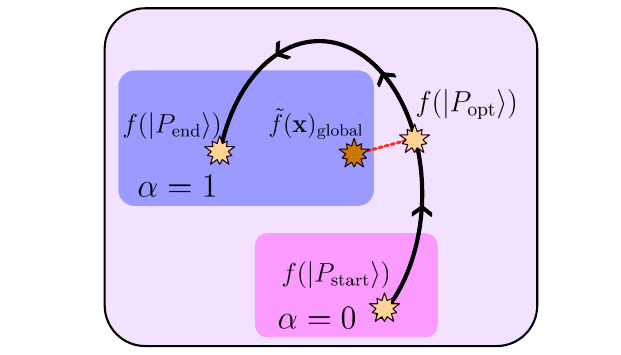}
\caption{Diagram providing a conceptual overview of the symmetric subspace amplification motivation. We start at $\alpha = 0$ with a state encoding $\ket{P_{\text{start}}} = \bigotimes_{i=1}^n \ket{\mathbf{p}_i}$ corresponding to an initial state without symmetry amplification, which is input to some classification model $f(\cdot)$. As $\alpha$ increases, the state moves through the solution space until finally we have a fully symmetric state at $\alpha = 1$ where the input state will be equal to $\ket{P_{\text{end}}} = \frac{1}{\Omega'} \sum_{\sigma \in S_n}\ket{\textbf{p}_{\sigma_1}}\ket{\textbf{p}_{\sigma_2}}...\ket{\textbf{p}_{\sigma_n}}$ as initially proposed in \cite{heredge2023permutation}. If we are dealing with data that is known to be permutation symmetric, then the global optimal model for the input data $\tilde{f}(\mathbf{x})_{\text{global}}$ (which could be classical or quantum) must be within the fully permutation symmetric subspace of models. However, the precise model that yields the global optimum remains unknown and may not correspond to the function $f(\cdot)$ used in this particular case, suggesting that an intermediate value of $\alpha$ might align more closely with the global optimum than the fully symmetric projection at $\alpha = 1$. In the figure depicted, the red line illustrates the deviation between the global model solution and the nearest solution obtained at an intermediate value of $\alpha$ in which the input state is denoted $\ket{P_{\text{opt}}}$ .
 }\label{fig/explanation_figure}
\end{figure}

The framework introduced previously allows the implementation of any linear combination of projections to an irreducible representation subspace with full control of the ratios of these projections. One can consider decomposing a state into its irreducible representation components, as shown in Figure~\ref{fig/state_amplification} and gaining control over the relative ratios between the amplitude components, allowing them to be adjusted to suit the model and data.

\begin{figure}[h]%
\centering
\includegraphics[width=0.8\linewidth]{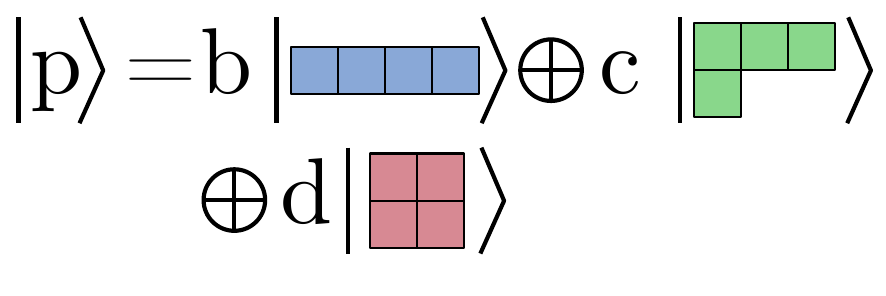}
\caption{A state can be viewed as a linear combination of its projection onto its irreducible representations. Utilising this projection algorithm we are able to change the relative weightings of $b, c$ and $d$, effectively amplifying and reducing the portion of the original state that lies in the respective irreducible representation subspace. }\label{fig/state_amplification}
\end{figure}

Via the LCU framework we previous showed in Theorem~\ref{thm:lcu_irrep_projection_general} that it is possible to implement a projection onto a quantum state in the following form
\begin{equation}
  \sum_r a_r \hat{P}_r \ket{\psi},
\end{equation}
where the $a_r$ parameters can be controlled via the initial state of the ancilla qubits as
\begin{equation}
  \ket{c} = \frac{1}{\Omega}\sum_{r=1}^{R} a_r n_r \ket{b_r}.
\end{equation}
In order to implement the symmetric subspace amplification algorithm, one can define a parameter $\alpha$ that parameterises the amount of symmetry in the encoding. Then we can assign the following $a_r$ values 
\begin{equation}\label{a_values}
  a_r = 
  \begin{cases}
    a_1 = 1 &\text{ for } r=1 \\
    a_r = (1 - \alpha) & \text{ for } r \neq 1 
  \end{cases}
\end{equation}
Hence, the initial ancilla state will be of the form 
\begin{equation}
  \ket{c} = \frac{1}{\Omega}\big(n_1 \ket{b_1} + (1 - \alpha)\sum_{r=2}^{R} n_r  \ket{b_r} \big).
\end{equation}
The $\alpha$ term can be varied to adjust the amount of permutation symmetry. When $\alpha = 0$ this corresponds to an equal projection on all irreducible representations, since the resulting operation is an equal projection onto all subspaces (weighted by the dimension of the subspace) and hence will leave the state unchanged. In contrast, the case where $\alpha = 1$ will result in projection onto the symmetric subspace only. As $\alpha$ increases, this corresponds to continuously amplifying the portion of the state that lies in the symmetric subspace.

In order to investigate the potential of symmetric subspace amplification, we implemented the model proposed here utilising the aforementioned technique alongside Qiskit \textit{statevector\_simulator} \cite{qiskit}. This was used to encode point cloud data and perform classification on the sphere and torus point cloud dataset as specified in \cite{heredge2023permutation}. More details can be found in Appendix~\ref{sec:symmetric-amplificaiton-numeric-details}. 

We record the average test accuracy over 10 experiments for a given $\alpha$ and then vary $\alpha$ to see which value is optimal for the data. The results presented in Figure~\ref{fig/accuracy_vs_mixing_with_dimension} illustrate that increasing the degree of symmetry yields the highest test accuracy at an intermediate value, demonstrating the potential advantage of this technique. Moreover, it is observed that while a fully symmetric projection ($
\alpha=1$) surpasses the performance of the non-symmetrised configuration ($\alpha = 0$), there is a significant reduction in dimensionality as we approach a fully symmetric state. This reduction corresponds to projecting onto a polynomial size space for the case of symmetric subspaces \cite{heredge2023permutation}. The $\alpha$ parameter is the ratio of the symmetric subspace to all other subspaces. In general one could select any relative weighting for all the irreducible subspaces, at the cost of having more parameters to optimise. We envisage that in practice a grid search would be used to find these hyperparameter values and that they would be optimised using a cross validation set, allowing final classification performance to be assessed on an unseen testing set. 

A key concern with fully symmetric states is their potential for efficient classical simulation due to this polynomially sized space. In contrast, the technique of partial symmetrisation not only achieves higher accuracy but also resides in a higher-dimensional space, possibly making it more challenging to simulate classically. While computer vision data can consist of thousands of points, there are example datasets that may be more suitable in the near term, such as particle physics collision classification tasks in which there may be less than 10 particles (points) in a given event \cite{Heredge21}. This would likely be a nearer term goal for the application of point cloud specific algorithms, although the method described in this paper could be applied to many symmetries across various different data types.

\begin{figure}[h]%
\centering
\includegraphics[width=1\linewidth]{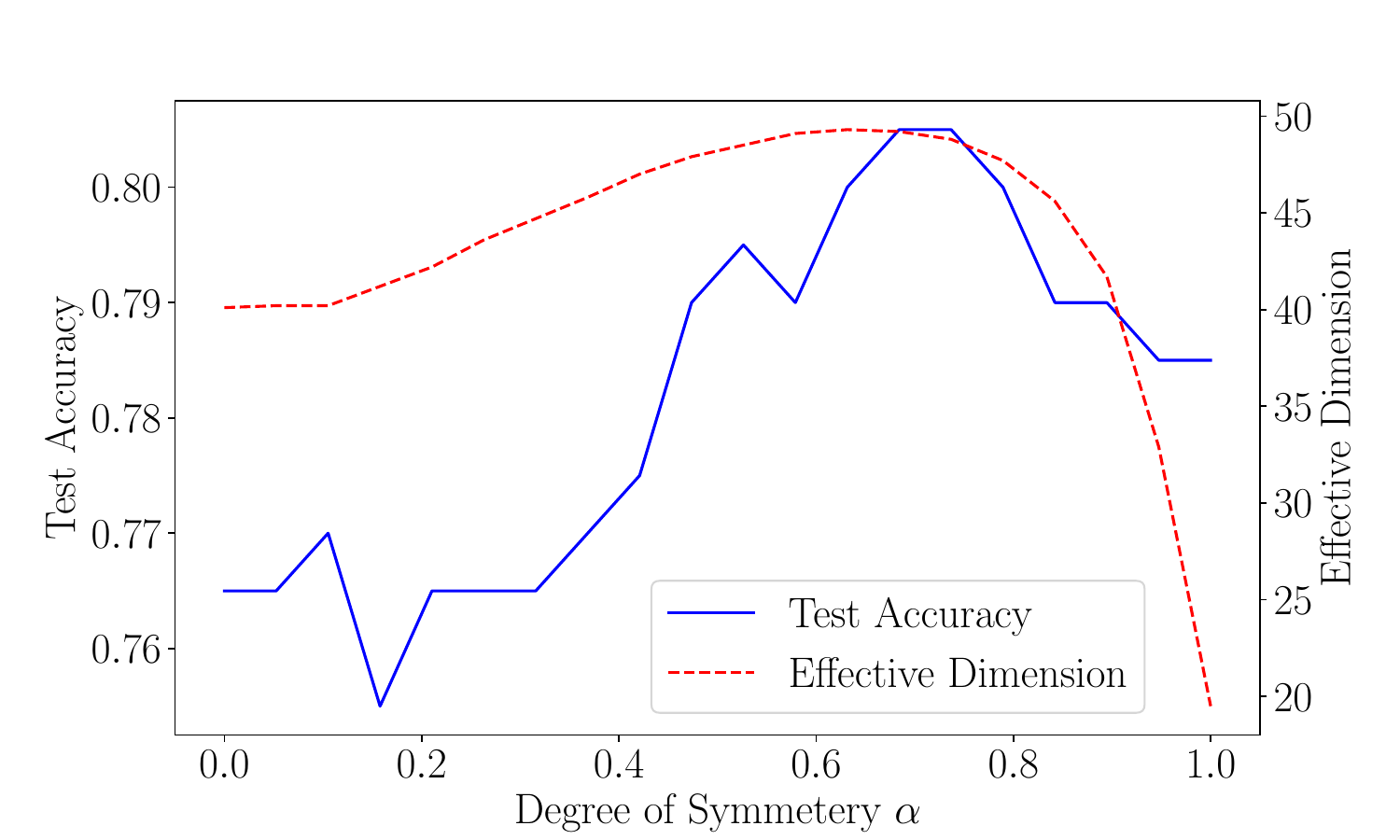}
\caption{Test accuracy and effective dimension as the degree of permutation symmetry is varied. Result is the average over 10 repeated experiments, each with 100 data points each of the sphere and torus classification dataset using quantum kernel estimation followed by a Support Vector Machine classification \cite{havlicek_supervised_2019, heredge2023permutation, Heredge21}. The effective dimension is found by explicitly calculating the higher dimensional encoding, performing principal component analysis, and calculating the amount of dimensions required to explain 95\% of the variance of the dataset. We show that an intermediate value of permutation symmetry around $\alpha = 0.7$ results in the best performance for this classification.}\label{fig/accuracy_vs_mixing_with_dimension}
\end{figure}

\section{Conclusion}

In this study, we have demonstrated the applicability and flexibility of the Linear Combination of Unitaries method in enhancing Quantum Machine Learning architectures. Our work provides several implementations of classical machine learning structures within the quantum domain, achieving not only a foundational translation of these concepts but also demonstrating potential for computational advantages.

The implementation of quantum ResNet demonstrates a procedure that can avoid the trainability issues associated with barren plateaus by allowing shallow depth components to survive in the final loss function. It is still a question of further research as to understanding the exact dynamics of the non-unitary terms. We demonstrated that by parameterising the strength of the residual connections through $\beta_l$ terms, one can increase the lower bound probability of success of the LCU method. This provides a possible avenue to tackle the key problem of the LCU method, which is that the implementations can only be performed probabilistically.

We also demonstrated an implementation of quantum native average pooling layers commonly applied in convolutional neural networks. These layers have demonstrated success in classical techniques and therefore have the potential to improve quantum convolutional neural networks. While classically one would need to calculate $\mathcal{O}(N^2)$ averages, the quantum parallelism inherent in our technique applies the averaging to all pixels (which are amplitude encoded into a quantum state) simultaneously. By generalising our work further, we can recreate quantum convolutional filters as has been previously proposed \cite{wei2021quantumconvolutionalneuralnetwork}, while utilising exponentially fewer controlled unitary operations.

Finally, we demonstrated an integration of irreducible representation projections into a quantum framework, allowing for the encoding of data symmetries directly into the quantum state, which has the potential to enhance the model's generalisation capabilities across various data structures. This general irreducible subspace projection framework includes previous work regarding permutation invariant encodings for point clouds \cite{heredge2023permutation} as a special case, while in Section~\ref{sec:rotat-invariant-encoding} we introduced a novel rotationally invariant encoding for point cloud data using the $S_4$ group. We further introduced in Section~\ref{sec:subspace-amplification} a method of parameterising the amount of symmetry in an encoding as a tuneable hyperparameter. This allows the quantum state to remain in an exponentially large space, while certain subspaces are amplified relative to each other, helping to mitigate some of the simultability concerns associated with fully projecting to a polynomially sized subspace. This added flexibility resulted in an improved performance for point cloud data when encoding an intermediate amount of permutation symmetry when compared to either a fully permutation invariant or noninvariant encoding. These were illustrative examples of how the framework could be utilised and there could be many possible symmetries and datasets for which this framework could be similarly adapted and optimised.

As the field of quantum computing continues to mature, the methods and frameworks presented here offer promising avenues for developing more sophisticated quantum algorithms that leverage both the computational benefits of quantum mechanics and the established successes of classical machine learning architectures. By utilising LCU methods we have shown non-unitary operations can be applied in a QML setting; this research not only extends the theoretical possibilities and flexibility of QML algorithms but also provides practical frameworks for their application, setting the stage for further evolution of quantum machine learning models.

\section*{Author Contributions}

J. Heredge envisaged the project, initially focused on irreducible subspace projections, with M. Sevior and L. Hollenberg as an extension of their previous work \cite{heredge2023permutation}. J. Heredge devised the ResNet and Average Pooling layers applications. The theorems and numerical simulations were developed by J. Heredge with the help of M. West. M. Sevior and L. Hollenberg supervised the overall project. All authors contributed in reviewing the manuscript.

\section*{Acknowledgements}

This research was supported by the Australian Research Council from grant DP210102831. J. Heredge acknowledges the support of the Australian Government Research Training Program Scholarship and the support of the N.D. Goldsworthy Scholarship. This work was supported by the University of Melbourne through the establishment of an IBM Quantum Network Hub at the University. We thank Dougal Davis for early discussions on representation theory that helped aid the development of the irreducible subspace projection work. We also thank Giulio Crognaletti for discussions on current quantum ResNet literature.

\begin{appendices}

\section{Conjugacy Class Implementation}\label{apn:improvement}

In the implementation of irreducible subspace projections detailed in Section~\ref{sec:irrep_projections}, an operator $\hat{\chi}$ was required that contained the character of every element in the group $G$. This meant that the dimension of the matrix scaled with $\rvert G \rvert$, which can be problematic for some groups, such as the permutation group $S_n$ where $\rvert S_n \rvert = n!$. To tackle this, in this section we will consider a unitary operator $\tilde{\chi}$ that contains the character for each conjugacy class in the group, with additional ancilla qubit registers per conjugacy class prepared in an equal superposition that then acts to distribute these characters to the corresponding group elements. This effectively only requires encoding the character table of the group into a unitary operator, and hence $\tilde{\chi}$ scales with the number of conjugacy classes of $G$. This has the potential to drastically reduce the size of the unitary that must be applied, at the expense of requiring additional ancilla qubits. We show that utilising this circuit structure recreates the result of Theorem~\ref{thm:lcu_irrep_projection_general} and is therefore a valid alternative approach.

The character orthogonality theorem \cite{Huppert1998} stated previously in Equation~\ref{eqn:charorththm} relates to a sum over all the elements of the group. As the character is the same for all elements in a conjugacy class, this can be adjusted to a sum over all conjugacy classes $\nu_i \in \mathcal{C}$ by introducing a term $d_{\nu_i} = \rvert \nu_i \rvert$ that represents the number of group elements in the conjugacy class $\nu_i$. Therefore, the character orthogonality theorem can be written
\begin{equation}
  \frac{1}{\rvert G \rvert }\sum_{{\nu_i} \in \mathcal{C}} d_{\nu_i} \chi_j({\nu_i}) \chi_k({\nu_i})^* = \delta_{jk}.
\end{equation}
Therefore, to create a matrix with orthonormal columns (to ensure $\tilde{\chi}$ is unitary) the characters for a conjugacy class $\nu$ must be weighted by square root of the number of elements in the conjugacy class $\sqrt{d_{\nu_i}}$. Hence, we require that
\begin{equation}
\tilde{\chi}\ket{b_r} = \frac{1}{\sqrt{\rvert G \rvert}}\sum_{{\nu_i} \in \mathcal{C}} \sqrt{d_{\nu_i}} \chi_r({\nu_i})^* \ket{{\nu_i}},
\label{eqn:charoperator-alternative}
\end{equation}
where in this case the conjugacy classes $\nu_i \in \mathcal{C}$ are labelling the basis states $\ket{\nu_i} \equiv \ket{b_i}$ defined as $\{ \ket{b_j} \}_{j \in [1, 2^k]} \in (\mathbb{C}^2)^{\otimes k}$ the basis states of the $2^k$ dimensional Hilbert space for the $k$ qubit ancilla register denoted by $\mathcal{H}= (\mathbb{C}^2)^{\otimes k}$. These can be taken to be the computational basis states.

If we consider $S_3$ with its character table given in Table~\ref{tab:char_table_s3} then we see that in this case the appropriate unitary matrix to be constructed would be
\begin{equation}
  \tilde{\chi} = \frac{1}{\sqrt{6}}
  \begin{bmatrix}
    1 & 2 & 1 & 0 \\
    \sqrt{3} & 0 & -\sqrt{3} & 0\\
    \sqrt{2} & -\sqrt{2} & \sqrt{2} & 0 \\
    0 & 0 & 0 & 1
  \end{bmatrix}
\end{equation}
\begin{table}[]
\caption{Character table for the group $S_3$. There are three irreducible representations with their characters denoted by $\chi_1, \chi_2, \chi_3 $. The conjugacy classes of $S_3$ are $C_{(I)} = \{ I \}$, $C_{(12)} = \{ \sigma_{12}, \sigma_{13}, \sigma_{23} \}$ and $C_{(123)} = \{ \sigma_{231}, \sigma_{312} \}$. 
   \vspace{0.2cm}}
  \centering
    \begin{tabular}{|l||l|l|l|}
    \hline
     Conjugacy Class $\nu_i$& $C_{(I)}$& $C_{(12)}$& $C_{(123)}$\\
     \hline
     Number of Elements $d_{\nu_i}$&
     1 & 3 & 2 \\\hline  
     $\chi_1$& 1 & 1 & 1\\
     $\chi_2$& 2 & 0 & -1\\
     $\chi_3$& 1 & -1 & 1\\     
    \hline
    
    \hline
    \end{tabular}
  \label{tab:char_table_s3}
\end{table}

This successfully allows for the construction of a unitary matrix $\tilde{\chi}$ that scales in size with the number of conjugacy classes as opposed to the number of group elements. We therefore have satisfied the constraint that $\tilde{\chi}$ is unitary and can be implemented on a quantum device. In general $\tilde{\chi}$  can be written
\begin{equation}
\tilde{\chi} = \frac{1}{\sqrt{\rvert G \rvert}}\sum_{r=1}^R \sum_{\nu_i \in \mathcal{C}} \sqrt{d_{\nu_i}} \chi_r({\nu_i})^* \ket{{\nu_i}} \bra{b_{r}} + \sum_{r=R+1}^{2^k}(...)\bra{b_r},
\label{eqn:chartildeoperator-full}
\end{equation}
where $(...)$ collects terms associated with $\ket{b_r}, r > R$ which will not be used in the construction.

Similarly to the main text, the first step is the pre-initialisation of the ancilla qubits to some combination of states that index the different irreducible representations
\begin{align}
\hat{\gamma}\ket{b_1}\ket{\psi} =\frac{1}{\Omega} \sum_{r =1}^R a_r n_r \ket{b_r} \ket{\psi},
\end{align}
where $\ket{b_1}$ is the initial state of the ancilla qubit register, usually assumed to be $\ket{0}^{\otimes k}$. By applying $\tilde{\chi}$ to this pre-initialised ancilla state, one can write
\begin{align}
 &\tilde{\chi}\hat{\gamma}\ket{b_1} \ket{\psi} = \tilde{\chi} \frac{1}{\Omega} \sum_{r =1}^R a_r n_r \ket{b_r} \ket{\psi} \nonumber \\
 & = \frac{1}{ \Omega \sqrt{ \rvert G \rvert}}\sum_{r =1}^R a_r n_r \sum_{\nu_i \in \mathcal{C}} \sqrt{d_{\nu_i}} \chi_r({\nu_i})^* \ket{{\nu_i}} \ket{\psi} ,
\end{align}
where as previously stated $\ket{\nu_i} \equiv \ket{b_i}$ are the same basis states, usually taken to be the computational basis states, but we have switched from $\ket{b_r}$ being indexed by the irreducible representation $r$, to the basis states being labelled by the conjugacy class $\nu_i$. This contrasts to $\hat{\chi}$ in the main text, where the group elements label the ancilla states. After application of $\tilde{\chi}$ the ancilla qubits now have basis states corresponding to the conjugacy classes $\ket{\nu_i}$. The selection operator requires the creation of additional qubit ancilla registers, denoted by $\ket{\omega_{\nu_i}}$, for each conjugacy class which are encoded into an equal superposition of $d_{\nu_i} = \rvert {\nu_i} \rvert$ states
\begin{equation}
  \ket{\omega_{\nu_i}} = \frac{1}{\sqrt{d_{\nu_i}}}\sum_{l_{\nu_i} = 1}^{d_{\nu_i}} \ket{l_{\nu_i}},
\end{equation}
where $\ket{l_{\nu_i}}$ are basis states $\{ \ket{l_{\nu_i}} \} \in (\mathbb{C}^2)^{\otimes k_{\nu_i}}$ of the $2^{k_{\nu_i}}$ dimensional Hilbert space for the $k_{\nu_i}$ qubit ancilla register $\ket{\omega_{\nu_i}}$, which can be taken to be the computational basis states. The groups elements $l_{\nu_i}$ that are contained in the conjugacy class $\nu_i$ now label the basis states in register $\ket{\omega_{\nu_i}}$. Each register has the requirement $k_{\nu_i} \geq \log(d_{\nu_i})$ to ensure that there are sufficient qubits in the register such that each element in ${\nu_i}$ has a corresponding basis state $\ket{l_{\nu_i}}$. The $\ket{\omega_{\nu_i}}$ are prepared on a $k_{\nu_i}$ qubit register which is initially in a basis state $\ket{b_1^{\nu_i}}$ through the operation
\begin{equation}
    \hat{\omega}_{\nu_i} = \frac{1}{\sqrt{d_{\nu_i}}}\sum_{l_{\nu_i} = 1}^{d_{\nu_i}} \ket{l_{\nu_i}} \bra{b_1^{\nu_i}} + \sum_j(...)\bra{b_j^{\nu_i}},
\end{equation}
where $(...)$ summarises terms associated with $\bra{b_j^{\nu_i}}, j \geq 2$ which will not be present in any calculations.

Overall the operations up to this point can be written as
\begin{align}
 \bigotimes_{{\nu_i}'} \ket{\omega_{{\nu_i}'}} & \otimes \tilde{\chi}\hat{\gamma}\ket{b_1} \otimes \ket{\psi} \nonumber \\
 = \frac{1}{\Omega \sqrt{\rvert G \rvert}} & \Big(  \big( \bigotimes_{{\nu_i}'} \frac{1}{\sqrt{d_{{\nu_i}'}}} \sum_{l_{{\nu_i}'} = 1}^{d_{{\nu_i}'}} \ket{l_{{\nu_i}'}} \big) \nonumber \\
 & \otimes \sum_{r =1}^R a_r n_r \sum_{\nu_i \in \mathcal{C}} \sqrt{d_{\nu_i}} \chi_r({\nu_i})^* \ket{{\nu_i}} \ket{\psi} \Big).
\end{align}
The controlled operations are performed similarly to the main text, except in this case the group element $g \in G$ is indexed by both ${\nu_i}$ which determines the conjugacy class and $l_{\nu_i}$ which indexes the element within the given conjugacy class. In this case $S_\text{SELECT} \equiv \hat{S}$ will be controlled by both the conjugacy class state $\ket{{\nu_i}}$ and the element selector within the conjugacy class $\ket{l_{\nu_i}}$ of the register $\ket{\omega_{\nu_i}}$ associated with the conjugacy class ${\nu_i}$. The action of this selection operator can be defined as
\begin{equation}
  \hat{S} \ket{l_{{\nu_i}}} \ket{{\nu_i}} \ket{\psi} = \ket{l_{{\nu_i}}} \ket{{\nu_i}} U_{{\nu_i}, l_{{\nu_i}}} \ket{\psi},
\end{equation}
which is implemented by each unitary $U_{{\nu_i}, l_{{\nu_i}}}$ being controlled by both the $ \ket{l_{{\nu_i}}} $ and $\ket{{\nu_i}}$ states. Applying this operator one can see the effect is 
\begin{align}
 & \hat{S} \big( \bigotimes_{{\nu_i}'} \ket{\omega_{{\nu_i}'}} \otimes \tilde{\chi}\hat{\gamma}\ket{b_1} \otimes \ket{\psi} \big)
 \nonumber \\
 &= \frac{1}{\Omega \sqrt{\rvert G \rvert}} \Big( \sum_{\nu_i} \big( \bigotimes_{{\nu_i}' \neq {\nu_i}} \frac{1}{\sqrt{d_{{\nu_i}'}}}\sum_{l_{{\nu_i}' } = 1}^{d_{{\nu_i}' }} \ket{l_{{\nu_i}'}}\big) \nonumber \\
 & \otimes \sum_{r =1}^R \sum_{l_{\nu_i} = 1}^{d_{\nu_i}} a_r n_r\chi_r({\nu_i})^* \ket{l_{\nu_i}}\ket{{\nu_i}} U_{{\nu_i}, l_{\nu_i}} \ket{\psi} \Big) \nonumber \\
 & = \sum_{\nu_i} \bigotimes_{{\nu_i}' \neq {\nu_i}} \ket{\omega_{{\nu_i}'}}
 \sum_{r =1}^R \sum_{l_{\nu_i} = 1}^{d_{\nu_i}}  \frac{a_r n_r \chi_r({\nu_i})^*}{\Omega \sqrt{\rvert G \rvert}}  \ket{l_{\nu_i}}\ket{{\nu_i}} \nonumber \\
 & \hspace{1.4cm} \otimes U_{{\nu_i}, l_{\nu_i}} \ket{\psi} ,
\end{align}
where we are being flexible with the precise positioning in the tensor product of the unused registers $\ket{\omega_{{\nu_i}' \neq {\nu_i}}}$ for each ${\nu_i}'$ term for the sake of readability. The notation we are using is that for a given $\nu_i$ element in the sum over all conjugacy classes, all $\ket{\omega_{\nu_i'}}$ registers which will not affect the $S_{\text{SELECT}}$ operator because ${\nu_i}' \neq {\nu_i}$ are written first, with the case ${\nu_i}' = {\nu_i}$ written to the right of this. In reality the register position will be different for each $\nu_i$, but this becomes difficult to denote in the notation.

We can now uncompute all $\hat{\omega}_{\nu_i}$ registers by noting that
\begin{equation}
    \hat{\omega}_{\nu_i}^\dagger = \frac{1}{\sqrt{d_{\nu_i}}}\sum_{l_{\nu_i} = 1}^{d_{\nu_i}}  \ket{b_1^{\nu_i}} \bra{l_{\nu_i}} + \sum_j(...)\ket{b_j^{\nu_i}}.
\end{equation}
This means that if we require these registers to be measured in the $\ket{b_1^{\nu_i}}$ then we can ignore the terms summarised by $(...)$ associated with $\ket{b_j^{\nu_i}}, j \geq 2$. Application of $\hat{\omega}_{\nu_i}^\dagger$ to all such registers gives
\begin{align}
 & \bigotimes_{{\nu_i}} \hat{\omega}_{\nu_i}^\dagger \hat{S} \big( \bigotimes_{{\nu_i}'} \ket{\omega_{{\nu_i}'}} \otimes \tilde{\chi}\hat{\gamma}\ket{b_1} \otimes \ket{\psi} \big)
 \nonumber \\
 & =  \sum_{\nu_i} \bigotimes_{ {\nu_i}} \ket{b_1^{\nu_i}}
 \sum_{r =1}^R \sum_{l_{\nu_i} = 1}^{d_{\nu_i}}  \frac{a_r n_r \chi_r({\nu_i})^*}{\Omega \sqrt{ d_{\nu_i}\rvert G \rvert}}  \ket{{\nu_i}}  \otimes U_{{\nu_i}, l_{\nu_i}} \ket{\psi}  \nonumber \\
 & + \sum_{j \geq 2}(...)\ket{b_j^{\nu_i}}
 \end{align}

 We measure the registers which previously contained $\hat{\omega}_{\nu_i}$ states and discard unless we measure all $\ket{b_1^{\nu_i}}$ states $\forall \nu_i$. This discards the terms collected by the $(...)$ in the previous equation.

The final steps consist of applying $\tilde{\chi}^\dagger$ and measuring the conjugacy class $\ket{{\nu_i}}$ ancillas to be in the initial state $\ket{b_1}$. Due to the fact that the state $\ket{b_1}$, which is commonly assumed to be $\ket{0}^{\otimes k}$, is used to index the first representation, we can write
\begin{equation}
\bra{b_1}\tilde{\chi}^\dagger = \frac{1}{\sqrt{\rvert G \rvert}}\sum_{{\nu_i} \in \mathcal{C}} \sqrt{d_{\nu_i}} \chi_1({\nu_i}) \bra{{\nu_i}}.
\label{eqn:charoperator-alternative-bra}
\end{equation}
where noticeably $\chi_1({\nu_i}) = 1, \forall {\nu_i}$ due to the fact that we are free to define $r=1$ to always correspond to the trivial representation when constructing $\tilde{\chi}$. Applying this operator to the conjugacy class ancilla qubits in the circuit therefore results in

\begin{align}
 &\Big( \bra{b_1}\tilde{\chi}^\dagger \Big) \cdot \Big( \bigotimes_{{\nu_i}} \bra{b_1^{\nu_i}} \hat{\omega}_{\nu_i}^\dagger   \hat{S} \big( \bigotimes_{{\nu_i}'} \ket{\omega_{{\nu_i}'}} \otimes \tilde{\chi}\hat{\gamma}\ket{b_1} \otimes \ket{\psi} \big) \Big)
 \nonumber \\
 & = \frac{1}{\Omega' }\sum_{\nu_i}
 \sum_{r =1}^R \sum_{l_{\nu_i} = 1}^{d_{\nu_i}}  a_r \frac{n_r}{\rvert G \rvert} \chi_r({\nu_i})^*  U_{{\nu_i}, l_{\nu_i}} \ket{\psi},
\end{align}
where $\Omega' = \Omega \sqrt{\pi_S}$ is the normalisation constant of the final state. Due to the fact that conjugacy classes partition the group $G$, the labels of the conjugacy class ${\nu_i}$ along with the labels of the elements within each conjugacy class $l_{\nu_i}$ will uniquely index each group element. In addition, all elements in the same conjugacy class have the same character i.e., $\chi_r(g_1) = \chi_r(g_2)$ for $g_1, g_2 \in {\nu_i}$. Hence, we can relabel the above expression in terms of group elements $g$, rather than ${\nu_i}$ and $l_{\nu_i}$, to explicitly see that it can be written as
\begin{align}
  \frac{1}{\Omega }
 \sum_{r =1}^R a_r \frac{n_r}{\rvert G \rvert} \sum_{g \in G} \chi_r(g)^*  U_{g} \ket{\psi}.
\end{align}
This is the precise form of Equation~\ref{eqn:theorem-lcu-project} in Theorem~\ref{thm:lcu_irrep_projection_general}. Hence, we have provided another framework that satisfies the theorem in the main text. As the final step consisted only of reordering the indexing, it follows that probability of success $\pi_S$ will be the same as in the main text.

The dimension of $\tilde{\chi}$ scales with the number of conjugacy classes of the group $G$ as opposed to $\hat{\chi}$ in the main text, which scales with the number of elements in the group. This means that $\tilde{\chi}$ could potentially be much easier to implement than $\hat{\chi}$. However, a caveat is that an additional ancillary qubit register $\ket{\omega_{\nu_i}}$ is required for each conjugacy class ${\nu_i}$, where each register will contain a number of qubits of the order $\mathcal{O}(\log(d_{\nu_i}))$. This presents a trade-off between unitary implementation difficulty and the number of ancilla qubits required. The true difficulty in implementing either unitary will ultimately depend on the group, but we provide both of these general techniques as possible starting points for future implementations.

\section{Rotational Invariance via Schur-Weyl}\label{sec:rot-invariance-schur-weyl}

In this section, we demonstrate how certain irreducible representations of $S_n$ could correspond to rotationally invariant encoding states when using the encoding specified in Section~\ref{sec:rotat-invariant-encoding}. Following the work and explanations of \cite{Kirby_2018, ragone2023representation} and in particular the thesis \cite{kirby_thesis} we reproduce an overview of Schur-Weyl duality and the concept of the Schur basis and add discussion at points linking this to how certain irreducible representation subspaces would correspond to rotationally invariant point cloud encodings in the case considered in Equation~\ref{eqn:rotat-inv-encoding}.

\subsection{Introduction to Schur Transforms}

The rotational invariant point cloud encoding we propose in Section~\ref{sec:rotat-invariant-encoding} relies on a property of the probabilistic irreducible subspace projection circuit when considering the group $S_n$, which in effect performs a projection to basis states of the Schur basis. Circuits that implement quantum Schur transforms have previously been introduced, such as those that implement them through the unitary group by iteratively applying the Clebsch-Gordan transforms \cite{bacon_efficient_2006} and alternative methods that implement it by considering the action of the symmetric group using quantum Fourier transforms \cite{krovi_efficient_2019}. Indeed, further investigation on whether the use of these existing algorithms could be adapted to more efficiently refine our procedure would be an interesting further research direction, as our implementation derives from the most general case of Theorem~\ref{thm:lcu_irrep_projection_general} which applies to any group $G$ and therefore may not be the most efficient in practice. 

The Schur transform is related to the concept of Schur-Weyl duality. In general, we can consider $n$ qudits of dimension $d$ inside a vector space $(\mathbb{C}^d)^{\otimes n}$ with a computational basis written as $\ket{i_1}\otimes\ket{i_2}\otimes...\otimes\ket{i_n}$. There are two representations on this space that are related to each other by Schur-Weyl duality. One of the representations is that of the Symmetric Group $S_n$, whose elements consist of all possible permutations of $n$ objects. The representation of $S_n$ on this vector space would simply consist of permuting the qudits. For a given group element $\sigma \in S_n$ one can write the action of the representation $P(\sigma)$ on this vector space as
\begin{equation}
  P(\sigma)\ket{i_1}\otimes...\otimes\ket{i_n} = \ket{i_{\sigma^{-1}(1)}}\otimes...\otimes\ket{i_{\sigma^{-1}(n)}}.
\end{equation}
Hence, each $\sigma \in S_n$ simply corresponds to a different permutation of the $n$ qudits. 

The second group to consider is the group of all $d \times d$ unitary operators $U \in \mathcal{U}_d$. As each qudit forms a $d$-dimensional vector, we see that the natural representation $Q(U)$ of the group $\mathcal{U}_d$ acting on the vector space $(\mathbb{C}^d)^{\otimes n}$ would correspond to the $n$-fold product action, where the same $U$ is applied to each qudit. This can be written as
\begin{equation}
  Q(U)\ket{i_1}\otimes...\otimes\ket{i_n} = U\ket{i_1}\otimes...\otimes U\ket{i_n}.
\end{equation}
The two actions $Q(U)$ and $P(\sigma)$ are fully reducible and hence can be written as a direct sum of their irreducible representations
\begin{equation}
  P(\sigma) \rightarrow \bigoplus \bm{I}_{n_\alpha} \otimes p_\alpha(\sigma),
\end{equation}

\begin{equation}
  Q(U) \rightarrow \bigoplus \bm{I}_{m_\beta} \otimes q_\beta (U).
\end{equation}
Note that the actions $P(\sigma)$ and $Q(U)$ commute with each other. The Schur-Weyl duality therefore states that there exists a basis which simultaneously decomposes the action of $P(\sigma)$ and $Q(U)$ into irreducible representations as
\begin{equation}
  Q(U)P(\sigma) \rightarrow \bigoplus_{\lambda} q_\lambda (U) p_\lambda (\sigma).
\end{equation}
We can write this Schur basis as $\ket{\lambda} \ket{q_\lambda} 
 \ket{p_\lambda}_{Sch}$ which decomposes the action of $P(\sigma)$ and $Q(U)$ as
 \begin{equation}
   Q(U) \ket{\lambda} \ket{q_\lambda} \ket{p_\lambda}_{Sch} = \ket{\lambda} (q^d_\lambda (U) \ket{q_\lambda}) \ket{p_\lambda}_{Sch},
 \end{equation}
 \begin{equation}
   P(\sigma) \ket{\lambda} \ket{q_\lambda} \ket{p_\lambda}_{Sch} = \ket{\lambda} \ket{q_\lambda} (p_\lambda (\sigma) \ket{p_\lambda})_{Sch}.
 \end{equation}
This action will decompose the vector space $(\mathbb{C}^d)^{\otimes n}$ as
\begin{equation}
  (\mathbb{C}^d)^{\otimes n} \rightarrow \bigoplus_{\lambda \in \Lambda} \mathcal{Q}_\lambda ^d \otimes \mathcal{P}_\lambda.
\end{equation}
Note that since these subspaces $\mathcal{Q}_\lambda ^d$ and $\mathcal{P}_\lambda$ are irreducible, the action of the representation $Q(U)$ or $P(\sigma)$ on a vector in spaces will keep the vector in that space. We will focus on building the Schur basis through considering the action of $P(\sigma)$ for $\forall \sigma \in S_n$.

\subsection{Rotational Invariance for Two Points is Trivial}

If we wish to find states that are rotationally invariant as we rotate points in a 3 dimensional point cloud, we should first consider what this means on the point cloud data. Each point can be represented by a vector in $\mathcal{R}^3$. If the entire point cloud is rotated, this amounts to every point being transformed by the same rotation, which is equivalent to applying a rotation matrix from the group $SO(3)$ to each point.

A natural way we choose to represent rotations of points in a quantum encoding is as rotations on the Bloch sphere representation of a qubit as was detailed in Section~\ref{sec:rotat-invariant-encoding}. This is further motivated by the fact that $SU(2)$ is the double cover of $SO(3)$, which means that every point in $SO(3)$ corresponds to two points in $SU(2)$, which means that we can use $SU(2)$ to model rotations in three-dimensional space classified by $SO(3)$. To produce an encoding like this, we simply need to apply the parameterised rotation gates $R_x(\theta)$ and $R_z(\phi)$ to an initial $\ket{0}$ state, where $(\theta, \phi)$ are the angle coordinates of a point in radial coordinates. We are now left with the angular direction of the point represented in the Bloch sphere of a single qubit. If we do this for all $n$ points in the point cloud, we will have $n$ qubits. A rotation of the entire point cloud can now be represented by the application of any $SU(2)$ unitary to every qubit simultaneously. This is a tensor representation of the $SU(2)$ group and can be written explicitly as $\bigotimes_{i=1}^n U$ for $U \in \mathrm{SU}(2)$. 

To further study this, we shall resort to the case consisting of only two points $n=2$. We again follow standard explanations from reference texts \cite{Kirby_2018, ragone2023representation, kirby_thesis}  while relating these points back to the rotationally invariant encoding we propose throughout the discussion. In this case rotation of the point cloud corresponds to applying the representation $U \otimes U; U \in SU(2)$, to the two qubits. Note that this construction commutes with the SWAP operator $[U \otimes U, \mathrm{SWAP}] = 0$, where SWAP corresponds to permuting the two qubits. This fact means that $U \otimes U$ and SWAP can be simultaneously diagonalised in a common basis.

If we consider two qubits such that $V=(\mathbb{C}^2)^{\otimes 2}$ and the group $S_2$ then a natural representation on two qubits would be $\{\mathbb{I}, SWAP\}$. If we change from the computational basis to the Schur basis $ \{ \ket{00}, \frac{1}{\sqrt{2}}(\ket{01} + \ket{10}), \ket{11}, \frac{1}{\sqrt{2}}(\ket{01} - \ket{10}) \} $, then we can see that this is a basis of eigenvectors for the SWAP operator that is completely diagonalised in this basis. We can show this explicitly as
\begin{align}
&SWAP \ket{00} = \ket{00}, \nonumber \\
&SWAP \frac{1}{\sqrt{2}}(\ket{01} + \ket{10}) = \frac{1}{\sqrt{2}}(\ket{01} + \ket{10}), \nonumber \\
&SWAP \ket{11} = \ket{11},\nonumber \\
&SWAP \frac{1}{\sqrt{2}}(\ket{01} - \ket{10}) = -\frac{1}{\sqrt{2}}(\ket{01} - \ket{10}).
\end{align}
As both $\{ \mathbb{I}, SWAP \}$ have the same action when applied to the basis states $\{ \ket{00}, \ket{11}, \frac{1}{\sqrt{2}}(\ket{01} + \ket{10}) \}$, this means it corresponds to the trivial representation $\textbf{1}$ of $S_2$. The action on this subspace leaves states unchanged and hence they have an eigenvalue of $1$. The action of SWAP when applied to the basis state $\{ \ket{01} - \ket{10} \}$, leads to a sign flip (eigenvalue of -1) which is therefore denoted as the $\textbf{sign}$ representation of $S_2$, where the sign of the state gets flipped by SWAP and remains unchanged under $\mathbb{I}$.

We can simultaneously block diagonalise $U \otimes U$ and SWAP as they both commute $[U \otimes U, SWAP]=0$. Hence $U \otimes U$ is block diagonalised in this basis. This means that if we have a state in the space spanned by $\{ \ket{00}, \frac{1}{\sqrt{2}}(\ket{01} + \ket{10}), \ket{11} \}$ that it will in this space under the operation of $U \otimes U$. It is possible to identify for $U \otimes U$ a symmetric eigenspace, denoted $\mathrm{Sym}^2(\mathbb{C}^2)$, with eigenvalue $\lambda_{\mathrm{sym}} = 1$ along with an anti-symmetric eigenspace, denoted $\mathrm{\wedge}(\mathbb{C}^2)$, with an eigenvalue $\lambda_{\mathrm{ant}} = -1$.

If we want to find a state that is invariant under $S_2$, which in this case corresponds to permuting the order of the qubits, then the state needs to be in the eigenspace of the trivial representation of $S_2$, this corresponds to the state formed with the basis $\{ \ket{00}, \ket{01} + \ket{10}, \ket{11} \}$, as in this basis both $\mathbb{I}$ and SWAP both leave the state unchanged,hence the state is permutation invariant. Projections to this subspace is how permutation invariant encodings are performed, as initially suggested in \cite{heredge2023permutation}.

However, while the symmetric space is closed under the action $U \otimes U$, the application of the unitaries will still change the state. If we have a unitary matrix given by
\begin{equation}
  U = 
  \begin{pmatrix}
    a & -b^* \\
    b & a*
  \end{pmatrix},
\end{equation}
where $\rvert a \rvert^2 + \rvert b\rvert^2 = 1$. If one applies a Schur transform to $U \otimes U$ which transforms it into the basis described previously it has been shown that it takes a block diagonal form \cite{kirby_thesis, Kirby_2018} 
\begin{equation}
  \mathrm{Sch} (U \otimes U)\mathrm{Sch}^\dagger \negthinspace = \negthinspace 
  \begin{pmatrix}
    a^2 & -\sqrt{2}ab^* & (b^*)^2 & 0 \\
    \sqrt{2}ab & \rvert a \rvert^2 - \rvert b \rvert^2 & -\sqrt{2}a^*b^* & 0 \\
    b^2 & \sqrt{2}a^*b & (a^*)^2 & 0 \\ 
    0 & 0 & 0 & 1    
  \end{pmatrix}
\end{equation}
and the SWAP operator can be written as
\begin{equation}
  \mathrm{Sch} (\mathrm{SWAP}) \mathrm{Sch}^\dagger \negthinspace = \negthinspace
  \begin{pmatrix}
    1 & 0 & 0 & 0 \\
    0 & 1 & 0 & 0 \\
    0 & 0 & 1 & 0 \\    
    0 & 0 & 0 & -1    
  \end{pmatrix}
\end{equation}

The symmetric subspace consists of the upper left block. $U \otimes U$ will transform the state, but keep it within the symmetric subspace. This also means that any variational circuit of the form $U \otimes U$ will similarly keep the state symmetric. This is the intuition behind geometric QML techniques in which variational layers can be constructed, such as $U \otimes U$ in this case, that preserve the symmetry of the state they act on while still allowing variational parameters to be adjusted \cite{meyer22}. Although in this work we do not consider creating variational circuits that preserve symmetry, instead we focus on projecting quantum encoded states to irreducible subspaces, which strictly enforces the symmetry invariance in the model as all information in the quantum state is deleted except the portion lying in the chosen irreducible subspace before the training and classification step are considered. It is important to note that the multiplicity of the subspace of $S_n$ corresponds to the dimension of the representation of $U \otimes U$ in that subspace and vice versa.

If we require a state to be rotationally invariant using the encoding previously described, we require that the state lies in the subspace corresponding to the trivial representation of $SU(2)$, as under the trivial representation the state will be unchanged by the action $U \otimes U$. In the case $n=2$, this would correspond to the basis consisting of a single eigenvector $\frac{1}{\sqrt{2}}(\ket{01} - \ket{10})$. This is now the trivial representation for $SU(2)$ but corresponds to the $\textbf{sign}$ representation in $S_2$. This means that the state is invariant under $SU(2)$ but the SWAP operator changes its sign, hence it is not invariant under permutation. As the dimension and multiplicity of the space is 1, the only normalised state that can exist in this space would be the state $\frac{1}{\sqrt{2}}(\ket{01} - \ket{10})$. Therefore, any projection onto the antisymmetric space for $n=2$ would correspond to projection to the state $\frac{1}{\sqrt{2}}(\ket{01} - \ket{10})$, which would provide rotational invariance, but result in a trivial encoding as all information about the data would be lost. Hence we show that as $U \otimes U$ and $S_2$ can be simultaneously block diagonalised, that we can perform projections for $S_2$ that take the state to a space that is the trivial representation of $U \otimes U$ and is therefore a rotationally invariant state. Note that it would not be possible in this set-up to have a state which is both permutation invariant and rotationally invariant, as they correspond to orthogonal irreducible subspaces.

We also note that if we consider the case where one attempts to create a equivariant variational quantum model that is rotationally invariant in a similar manner to how permutation equivariant models are created in this context, then the variational layer would have to be composed of SWAP gates. These clearly do not have the ability to be parameterised by any variational parameters and therefore it is not possible to create such a variational circuit if strictly following the line of reasoning presented here, although they could be implemented with a different encoding and approach. This highlights a further difference between the rotationally invariant projections proposed here, which enforces symmetry in the state itself, and equivariant VQC models.

\subsection{Rotationally Invariant Encoding for $S_4$}

The case presented for $n=2$ demonstrated a rotationally invariant projection is possible, but it always corresponds to the state $\frac{1}{\sqrt{2}}(\ket{01} - \ket{10})$ and hence is deleting all information about the encoded state. The space is technically one-dimensional, but when considering the normalisation condition of quantum states, the projection is effectively to a zero-dimensional space in terms of the information that is retained. Although the motivation for performing these projections is to reduce the dimensionality of the encoding to aid the generalisation ability of the overall model, this is a trivially large reduction in dimensionality which has no use in practice. We should now investigate the irreducible representations of $S_n$ for $n>2$ to find a less trivial example. We would like to find a subspace where for $S_n$ the dimension of the irreducible representation is $d$ and the multiplicity is $1$. This means that the corresponding subspace of $U^{\otimes n}$, will have dimension $1$ and multiplicity $d$. This would mean that $U^{\otimes n}$ is still in the trivial representation, meaning the state is unchanged under the action $U^{\otimes n}$, however it gives a wider range of possible states within this subspace, and hence some information about the encoded state is preserved.
\begin{figure}\label{fig:schurdecomp}
\begin{center}
  \setcounter{MaxMatrixCols}{20}
  \scalebox{0.9}{\rotatebox{90}{$\begin{bmatrix}
  a^{4} & - 2 a^{3} \overline{b} & \sqrt{6} a^{2} \overline{b}^{2} & - 2 a \overline{b}^{3} & \overline{b}^{4} & 0 & 0 & 0 & 0 & 0 & 0 & 0 & 0 & 0 & 0 & 0 \\2 a^{3} b & a^{2} \left(4 \rvert a \rvert^{2} - 3\right) & \sqrt{6} a \left( 1 - 2 \rvert a \rvert^{2}\right) \overline{b} & \left(4 \rvert a \rvert^{2} - 1\right) \overline{b}^{2} & - 2 \overline{a} \overline{b}^{3} & 0 & 0 & 0 & 0 & 0 & 0 & 0 & 0 & 0 & 0 & 0\\\sqrt{6} a^{2} b^{2} & \sqrt{6} a b \left(2 \rvert a \rvert^{2} - 1\right) & 6 \rvert a \rvert^{4} - 6 \rvert a \rvert^{2} + 1 & \sqrt{6} \left( 1 - 2 \rvert a \rvert^{2}\right) \overline{a} \overline{b} & \sqrt{6} \overline{a}^{2} \overline{b}^{2} & 0 & 0 & 0 & 0 & 0 & 0 & 0 & 0 & 0 & 0 & 0\\2 a b^{3} & b^{2} \left(4 \rvert a \rvert^{2} - 1\right) & \sqrt{6} b \left(2 \rvert a \rvert^{2} - 1\right) \overline{a} & \left(4 \rvert a \rvert^{2} - 3\right) \overline{a}^{2} & - 2 \overline{a}^{3} \overline{b} & 0 & 0 & 0 & 0 & 0 & 0 & 0 & 0 & 0 & 0 & 0\\b^{4} & 2 b^{3} \overline{a} & \sqrt{6} b^{2} \overline{a}^{2} & 2 b \overline{a}^{3} & \overline{a}^{4} & 0 & 0 & 0 & 0 & 0 & 0 & 0 & 0 & 0 & 0 & 0\\0 & 0 & 0 & 0 & 0 & a^{2} & - \sqrt{2} a \overline{b} & \overline{b}^{2} & 0 & 0 & 0 & 0 & 0 & 0 & 0 & 0\\0 & 0 & 0 & 0 & 0 & \sqrt{2} a b & 2 \rvert a \rvert^{2} - 1 & - \sqrt{2} \overline{a} \overline{b} & 0 & 0 & 0 & 0 & 0 & 0 & 0 & 0\\0 & 0 & 0 & 0 & 0 & b^{2} & \sqrt{2} b \overline{a} & \overline{a}^{2} & 0 & 0 & 0 & 0 & 0 & 0 & 0 & 0\\0 & 0 & 0 & 0 & 0 & 0 & 0 & 0 & a^{2} & - \sqrt{2} a \overline{b} & \overline{b}^{2} & 0 & 0 & 0 & 0 & 0\\0 & 0 & 0 & 0 & 0 & 0 & 0 & 0 & \sqrt{2} a b & 2 \rvert a \rvert^{2} - 1 & - \sqrt{2} \overline{a} \overline{b} & 0 & 0 & 0 & 0 & 0\\0 & 0 & 0 & 0 & 0 & 0 & 0 & 0 & b^{2} & \sqrt{2} b \overline{a} & \overline{a}^{2} & 0 & 0 & 0 & 0 & 0\\0 & 0 & 0 & 0 & 0 & 0 & 0 & 0 & 0 & 0 & 0 & 1 & 0 & 0 & 0 & 0\\0 & 0 & 0 & 0 & 0 & 0 & 0 & 0 & 0 & 0 & 0 & 0 & a^{2} & - \sqrt{2} a \overline{b} & \overline{b}^{2} & 0\\0 & 0 & 0 & 0 & 0 & 0 & 0 & 0 & 0 & 0 & 0 & 0 & \sqrt{2} a b & 2 \rvert a \rvert^{2} - 1 & - \sqrt{2} \overline{a} \overline{b} & 0\\0 & 0 & 0 & 0 & 0 & 0 & 0 & 0 & 0 & 0 & 0 & 0 & b^{2} & \sqrt{2} b \overline{a} & \overline{a}^{2} & 0\\0 & 0 & 0 & 0 & 0 & 0 & 0 & 0 & 0 & 0 & 0 & 0 & 0 & 0 & 0 & 1
  \end{bmatrix}$}} 
  \caption{The action $(U \otimes U \otimes U \otimes U)$ written in the Schur basis \cite{kirby_thesis, Kirby_2018}. Blocks correspond to irreducible subspaces of $SU(2)^{\otimes 4}$. The dimension of a block is the dimension of the irreducible subspace and the amount of repeats is the multiplicity. The trivial representation of $SU(2)^{\otimes 4}$ corresponds to blocks of $[1]$, with multiplicity $2$. This corresponds to a $2$-dimensional subspace in $S_n$ with multiplicity $1$.}
\end{center}
\end{figure}

Using the work of \cite{kirby_thesis, Kirby_2018} we consider the group $S_4$. In this case the action of $U \otimes U \otimes U \otimes U$ when written in the Schur basis is shown in Figure~\ref{fig:schurdecomp}. Within this decomposition, we can identify that there are two basis elements in the Schur basis that are invariant under the action of $(U \otimes U \otimes U \otimes U)$. This is the trivial representation of $SU(2)^{\otimes 4}$ which has dimension equal to 1, but has a multiplicity equal to 2. This correspond to the basis states
\begin{align}
  \ket{d_1} = \frac{\sqrt{3}}{6}( & 2\ket{0011} + 2\ket{1100} - \ket{0101} \nonumber \\
  & - \ket{1010} - \ket{0110} - \ket{1001}),
\end{align}
 and
 \begin{equation}
   \ket{d_2} = \frac{1}{2}(\ket{0101} + \ket{1010} - \ket{0110} - \ket{1001}).
 \end{equation}
Due to the Schur-Weyl duality, we know that one of the irreducible subspaces of $S_n$ will correspond to the same basis states. However, in this case the multiplicity of the irreducible representation of $S_n$ will be 1 and the dimension will equal 2. One can find that by projecting onto the $r = 3$ representation of $S_4$ (as specified in Table~\ref{tab:char_table_s4}) using the projection circuit proposed in Section~\ref{sec:irrep_projections}, the quantum input state $\ket{\psi}$ will projected onto some state $\ket{\psi_{\text{rot}}} = \lambda_1 \ket{d_1} + \lambda_2 \ket{d_2} $. In this case there are now two basis states that can be used. However, due to the normalisation constraint $\rvert \lambda_1 \rvert^2 + \rvert \lambda_2 \rvert^2 = 1$ the effective dimension is reduced to be equal to 1 overall. This means that at least some information from the original input state is retained when the projection is performed and is therefore a non-trivial example, although effective projection to a $1$-dimensional space may be too great a dimensionality reduction in practice.

In general, if it is possible to calculate the multiplicity of the trivial representation of dimension 1 of $SU(2)^{\otimes n}$ to be $m_{SU}$, then by Schur-Weyl duality there exists an irreducible subspace of $S_n$ of dimension $m_{SU}$ and multiplicity of $1$. Projecting to this subspace will therefore result in reducing the dimensionality of the encoded data to $m_{SU} - 1$, after considering the normalisation condition. We provide this rotationally invariant encoding as an example to demonstrate how the irreducible subspace projection circuit could be utilised for point cloud data, complementing previous work on permutation invariant encodings of point cloud data \cite{heredge2023permutation}. It may be the case that this is not feasible or useful for $n>4$ and such investigations could be subject to further research. However, the main motivation was to provide an additional example of how the projection framework could be used to encode symmetry in a model. In general, the irreducible subspace projection framework we introduced can be used for any representation of any finite group $G$ along with any encoding procedure for $\ket{\psi}$ and hence there exists ample flexibility for application to a wide array of problems beyond those mentioned as examples in this work. 

\section{Point Cloud Numeric Details}\label{sec:symmetric-amplificaiton-numeric-details}

This section specifies the details of the numerical results reported in Section~\ref{sec:subspace-amplification}. To generate the data, $n = 3$ points were sampled from the surface of a shape, either a sphere or a torus as originally used in \cite{heredge2023permutation}, to create point clouds. The process is repeated with an equal amount of point clouds from each shape until the desired number of total samples is reached. These are split into training sets (80\%) and testing sets (20\%). We evaluated the performance of the model on the test set and averaged the accuracy in 10 experiments in which a newly generated dataset is used each time. Both the sphere and torus are centred at the origin, with the torus scaled to match the sphere's average point magnitude. All data is normalised between $-\frac{\pi}{2}$ and $\frac{\pi}{2}$ to be encoded as rotation angles.

To encode a point cloud, each point $\mathbf{p}_i = (x_i, y_i, z_i)$ is first encoded into a two qubit quantum state $\ket{\mathbf{p}_i}$ using two repeated layers of the encoding circuit shown in Figure~\ref{fig:encoding_per_point}. Then the corresponding statevector is found using Qiskit \textit{statevector\_simulator} \cite{qiskit}. From here we can write the entire point cloud as $\ket{P} = \ket{\mathbf{p}_1} \otimes \ket{\mathbf{p}_2} \otimes \ket{\mathbf{p}_3}$. In order to perform the simulations we then found the projection of $\ket{P}$ onto each irreducible subspace classically utilising the result specified in Theorem~\ref{thm:irreducible representationProj} and the character table for $S_3$. From here the $a_r$ parameters specified in Equation~\ref{a_values} are used to produce the correct linear combination of projections and arrive at the partially symmetric subspace amplified state $\ket{P}_\alpha$, with the amount of symmetry parameterised by $\alpha$. We then implement a quantum support vector machine classification by calculating the inner product between the data points to form a kernel matrix, which is then input to a classical support vector machine for the final classification \cite{havlicek_supervised_2019, heredge2023permutation}. We recorded the average test accuracy over 10 experiments for a given $\alpha$ and then repeated with different $\alpha$ values and plotted the results to see which value of $\alpha$ is optimal for the data.

\begin{figure}[h]%
\centering
\includegraphics[width=1\linewidth]{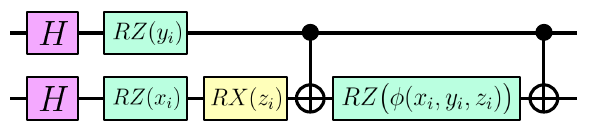}
\caption{One layer of the encoding circuit used to encode an individual point $\mathbf{p}_i = (x_i, y_i, z_i)$ in the symmetric subspace amplification model. The function $\phi(x_i, y_i, z_i) = \frac{2}{\pi^2}(\pi - x_i)(\pi - y_i)(\pi - z_i)$ means the entanglement of the state is dependent on some polynomial of all input variables. To generate the encoded state $\ket{\mathbf{p}_i}$ two repeated layers of this circuit are used \cite{havlicek_supervised_2019, Heredge21}.}\label{fig:encoding_per_point}
\end{figure}

\section{Dimensionality Reduction in Pooling Layers}\label{apn:dimensionality-reduction-pooling}

Regarding the average pooling layers implemented in this work in Section~\ref{sec:cnn}, it is crucial to understand that we have confined our operations to averaging quantum states across image translations. Notably, binary addition operates cyclically, where adding 1 to the highest value loops back to the lowest value, meaning our tiling would wrap around the edges of the image to perform the averaging. This may be an undesirable effect, unless the image has some periodic feature. It is also the case that the average pooling described does not provide any dimensionality reduction, as we can calculate the pooling window average for every pixel at no additional computational cost. 

To perform the dimensionality reduction, and to prevent the image being treated as periodic, one method could involve deleting any states in which $x_i > N - L \text{ or } y_j > N - L$ after the averaging part of the pooling layer is implemented by implementing
\begin{equation}
  \ket{x_i}\ket{y_j} \rightarrow \text{DELETE if } x_i \text{ or } y_j > N - L.
\end{equation}
To perform this without any a priori information of the data or architecture, one can define the $F_{\text{FLAG},x}$ and $F_{\text{FLAG},y}$ operators alongside two ancilla flag qubits, one for the $x$ and $y$ flags, respectively.
\begin{align}
 F_{\text{FLAG},x}& \ket{x_i}\ket{y_j} \ket{0}_x \ket{0}_y \nonumber \\
&= 
\begin{cases}
  \ket{x_i}\ket{y_j}\ket{0}_x \ket{0}_y \text{ if } x_j \leq N - L \\
  \ket{x_i}\ket{y_j}\ket{1}_x \ket{0}_y \text{ if } x_j > N - L 
\end{cases}
\end{align}
and
\begin{align}
 F_{\text{FLAG},y}& \ket{x_i}\ket{y_j}\ket{0}_x \ket{0}_y \nonumber \\
 & = \begin{cases}
  \ket{x_i}\ket{y_j}\ket{0}_x \ket{0}_y \text{ if } y_j \leq N - L \\
  \ket{x_i}\ket{y_j}\ket{0}_x \ket{1}_y \text{ if } y_j > N - L 
\end{cases}
\end{align}
One can now effectively remove unwanted states to reduce the dimensionality by measuring the ancilla warning qubits and disregarding results unless $\ket{0}_x \ket{0}_y$ is measured. This will delete any states in which the average pooling layer will have wrapped around the sides of the image. 

Clearly as this is a probabilistic procedure we expect that significantly more efficient implementations are possible depending on the exact specifics of the problem. For example, if it is known that some of the most significant qubits correspond to pixels that should not be used, then they can simply be discarded without needing to implement flag operators and extra ancilla qubits. Implementation of the discarding part of the pooling process will therefore be dependent on the exact requirements of the model, although in the worst-case scenario resorting to using $F_{\text{FLAG}}$ operators, as described above, to remove specific states would succeed in reducing the image dimensionality. We refer the reader to \cite{wei2021quantumconvolutionalneuralnetwork} for more details on implementing pooling within a full quantum convolutional neural network framework.

\section{Subtraction Circuit Scaling}\label{apn:decrementscaling}

The subtraction / decrement operator we consider is controlled by a single ancilla qubit meaning that overall we consider $n+1$ qubits consisting of a total of $n$ multi-controlled Toffolli gates, with the largest being controlled by $n$ qubits, and the amount of controlling qubits decreasing by $1$ each time such that the final gate is a CNOT gate. For $n+1$ total qubits the $n$ qubit multi-controlled Toffoli can be decomposed into at worst $\mathcal{O}(n^2)$ basic operations \cite{Barenco_1995}, which may be improved if ancillas are used \cite{baker2019decomposingquantumgeneralizedtoffoli}. The $m$ qubit controlled Toffoli with $m < n-1$ can be created by $\mathcal{O}(m)$ basic operations \cite{Barenco_1995}. As there are $n - 1$ gates of this type in total, their overall contribution will be bounded by $\mathcal{O}(n^2)$. Hence, overall we can upper bound the required number of basic operations at $\mathcal{O}(n^2)$ to create the basic controlled subtraction gate (Although it has also been suggested linear scaling may be possible for increment/decrement gates using extra ancillas \cite{gidney, khattar2024riseconditionallycleanancillae, gidney2018factoringn2cleanqubits}). Note that as we consider subtraction operators where the subtraction amount doubles each time, this means subsequent subtraction operators can be implemented by ignoring an additional least significant qubit as the subtraction amount doubles. This means that in practise subsequent subtraction operators will use even fewer gates, although we will not include this in our complexity calculation. As $n=\log(N)$ and as we require $\mathcal{O}(\log(D))$ operators in total then overall we require $\mathcal{O}(\log(D)(\log(N))^2)$ basic operations to implement the average pooling.

\section{Degeneracy Avoidance in Average Pooling}\label{apnd:degeneracy-prevention}

If window length $D < 2^L$ where $L$ is the number of ancilla qubits and controlled operators then we need to change the final subtraction operator value. For example if $D = 8$ then we would we would generate all $8$ operators $\mathbb{I} + \hat{T}_1+ \hat{T}_2+...+ \hat{T}_7$ by using only three controlled unitaries $\hat{T}_1, \hat{T}_2, \hat{T}_3$. However if we had $D=7$ then we could attempt to implement this using $\hat{T}_1, \hat{T}_2, \hat{T}_3$ but writing this out in full one will see this implements
\begin{equation}
    \mathbb{I}+\hat{T}_1+\hat{T}_2+ 2\hat{T}_3 + \hat{T}_4 + \hat{T}_5 + \hat{T}_6
\end{equation}
which correctly has only $D=7$ terms, but there is a degeneracy for $\hat{T}_3$ meaning that it has double weighting and we do not have an equal superposition. This may however be preferable if one wishes to create a general convolutional filter, where in this case the pixel corresponding to $\hat{T}_3$ in the convolutional filter is supposed to be given a weighting double that of all others. 

If one wishes to avoid this degeneracy though, then it can be achieved by avoiding an equal superposition of all ancilla qubits and instead initialising them to avoid repeating terms. In the above example both the ancilla states $\ket{100}$ and $\ket{110}$ result in an implementation of $\hat{T}_3$. Setting the amplitude of either of these states to zero would remove the degeneracy and correctly implement the average pooling layer again.

\section{Input Skip Connections Only ResNet}\label{apn:input-focussed-quantum-resnet}

In this section we shall discuss a less general but more qubit efficient implementation of quantum ResNet where the skipped connection is not between each layer, but only between the input and every other layer in the circuit (until immediately before the final layer).

\begin{theorem}[Input Skip Only Quantum ResNet]
  It is possible to probabilistically implement a quantum ResNet that only includes skipped connections between the input and every other layer by utilising the LCU method within a quantum variational model. In this context, the implementation of a Quantum Native ResNet for $L$ layers is defined as the implementation of the operator
\[ R_{\text{RES},L} \ket{\phi(\mathbf{x})} = \frac{1}{\Omega'}\sum_{f=1}^{L} \prod_{l = f}^{L} \rvert \gamma_f \rvert^2 W_l(\mathbf{\theta}_l) \ket{\phi(\mathbf{x})} , \]
 where $W_l(\mathbf{\theta}_l)$ corresponds to a unitary variational gate implemented in layer $l$ and $\Omega'$ is a normalisation constant. The coefficients $\gamma_f$ correspond to the desired weighting contribution for each layer and can be freely adjusted subject to the requirement that $\sum_f \rvert \gamma_f \rvert^2 = 1$.

\end{theorem}

\begin{proof}

The preparation operator can be implemented on $k$ ancilla qubits as a unitary operator 
\begin{equation}
  P_\text{PREP} = \sum_{f=1}^{L} \gamma_f \ket{b_f} \bra{b_1} + \sum_{j = 2}^L \sum_{f=1}^{L} u_{j,f} \ket{b_f} \bra{b_j},
\end{equation}
where the states $\{ \ket{b_f} \}_{f \in [1, 2^k]} \in (\mathbb{C}^2)^{\otimes k}$ are basis states of the $2^k$ dimensional Hilbert space for the $k$ qubit ancilla register denoted by $\mathcal{H}= (\mathbb{C}^2)^{\otimes k}$. This is usually taken to be the computational basis with $\ket{b_1} \equiv \ket{0}^{\otimes k}$. We are only concerned with the $\bra{b_1}$ term, as that operator is only applied to $\ket{b_1}$ such that
\begin{equation}
  P_\text{PREP}\ket{b_1} = \sum_{f=1}^{L} \gamma_f \ket{b_f} ,
\end{equation}

subject to the condition that $\sum_f \rvert \gamma_f \rvert^2 = 1$. Subsequently, the selection operator can be applied as
\begin{equation}
S_\text{SELECT} \ket{b_f} \ket{\phi(\mathbf{x})} =  \ket{b_f} \Big( \prod_{l = f}^{L} W_l(\mathbf{\theta}_l) \ket{\phi(\mathbf{x})} \Big),
\end{equation}
where $ \prod_{l = f}^{L} W_l(\mathbf{\theta}_l)$ is a product of unitary operators $ W_l(\mathbf{\theta}_l)$ and hence clearly implementable on a quantum circuit. Applying the LCU framework we can see that
\begin{align}
  & S_\text{SELECT} P_\text{PREP}\ket{b_1} \ket{\phi(\mathbf{x})} \nonumber \\
  & = \sum_{f=1}^{L} \gamma_f \ket{b_f} \Big( \prod_{l = f}^{L} W_l(\mathbf{\theta}_l) \ket{\phi(\mathbf{x})} \Big).
\end{align}
We can apply the conjugate preparation operator which in general can be written as
\begin{equation}
  (P_\text{PREP})^\dagger = \sum_{f=1}^{L} \gamma_f^* \ket{b_1} \bra{b_f} + \sum_{j = 2}^L \sum_{f=1}^{L} u_{jf}^* \ket{b_j} \bra{b_f},
\end{equation}
to find that
\begin{align}
  & P_\text{PREP}^\dagger S_\text{SELECT} P_\text{PREP}\ket{b_1} \ket{\phi(\mathbf{x})} \nonumber \\ 
  & = \sum_{f=1}^{L}  \rvert \gamma_f \rvert^2 \ket{b_1} \prod_{l = f}^{L} W_l(\mathbf{\theta}_l) \ket{\phi(\mathbf{x})} + \sum_{j = 2}^L \ket{b_j}(...).
\end{align}
The next step is to measure the ancilla qubits and discard results when the ancilla is not in the $\ket{b_1}$ state, hence we can ignore terms $ \ket{b_j}$ for $j \geq 2$. The probability $\pi_S$ of measuring the ancillas in the $\ket{b_1}$ state corresponds to
\begin{align}
  \pi_S = \rvert \sum_{f=1}^{L} \rvert \gamma_f \rvert^2 \prod_{l = f}^{L} 
 W_l(\mathbf{\theta}_l) \ket{\phi(\mathbf{x})} \rvert^2,
\end{align}
and will hence be dependent on the architecture and initial data encoding. Requiring the ancillas to be in the $\ket{b_1}$ state we see that
\begin{align}
  & \bra{b_1} P_\text{PREP}^\dagger S_\text{SELECT} P_\text{PREP}\ket{b_1} \ket{\phi(\mathbf{x})} \nonumber \\ 
  & = \frac{1}{\Omega'} \sum_{f=1}^{L}  \rvert \gamma_f \rvert^2 \prod_{l = f}^{L} W_l(\mathbf{\theta}_l) \ket{\phi(\mathbf{x})},
\end{align}
where $\Omega' = \sqrt{\pi_S}$, as required.
\end{proof}

This LCU method will successfully prepare the desired residual network. Hence we have shown that VQC models can be built with skipped connections under the quantum ResNet framework by utilising the LCU method to implement these non-unitary operations. An example circuit implementation for the four-layer case $L=4$ is shown in Figure~\ref{fig:quantum_resenet_input_skip_only}. Notably, this method is more qubit efficient requiring $\mathcal{O}(\log(L))$ qubits for $L$ layers, as opposed to $\mathcal{O}(L)$ qubits in the previous case.

 \begin{figure}[h]%
\centering
\includegraphics[width=1\linewidth]{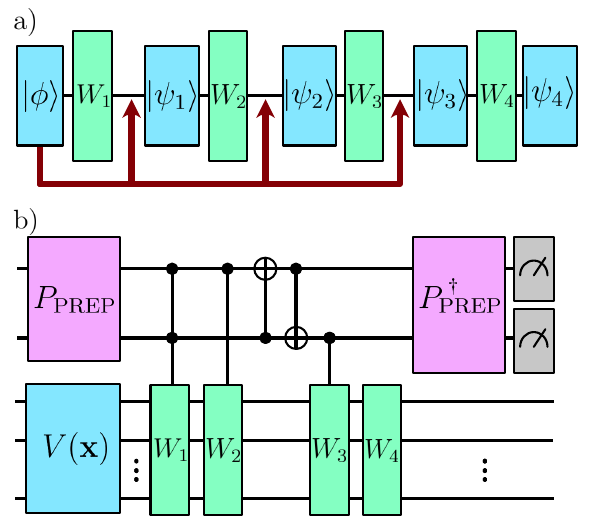}
\caption{a) ResNet conceptual illustration in which residual connections are only performed from the input layer to all subsequent layers (stopping before the final layer). b) Circuit implementation of this ResNet using the LCU method, showing the four gate case can be implemented with only two ancilla qubits. }\label{fig:quantum_resenet_input_skip_only}
\end{figure}

In terms of the probabilistic scaling we see
\begin{align}
  \pi_S = & \rvert \sum_{f=1}^{L}  \rvert \gamma_f \rvert^2 \prod_{l = f}^{L} W_l(\mathbf{\theta}_l) \ket{\phi(\mathbf{x})} \rvert^2 \nonumber \\
  = & \sum_{f, h } \rvert \gamma_h \gamma_f \rvert^2 \Big( \bra{\phi(x)} (\prod_{l = h}^{L}W_{L - l + h}(\mathbf{\theta}_h)^\dagger) \nonumber \\
  & (\prod_{l = f}^{L}W_l(\mathbf{\theta}_l))\ket{\phi(\mathbf{x})})\Big).
\end{align}
Hence, the probability of success depends on the encoding state $\ket{\phi(\mathbf{x})})$ and the variational layers, but can also be adjusted through the strength of the residual connection using the $\gamma$ parameters.

\subsection{Two Gate Input Worked Example}\label{sec:two-gate-resenet}

In this section we give a brief worked example for the two-layer case of our LCU quantum ResNet with equal weightings.

Consider the encoded quantum state as $\ket{\phi(\mathbf{x})}$ which will be evolved by the first layer of a variational circuit $W_1(\mathbf{\theta}_1)$ resulting in the state $ \ket{\psi_1(x, \mathbf{\theta}_1)} = W_1(\mathbf{\theta}_1)\ket{\phi(\mathbf{x})}$. 

If we wish to introduce a quantum ResNet type skipped connection between the input to the first layer then one realisation of this would be implementing the evolution
\begin{equation}
  W_2(\mathbf{\theta}_2)\big( \ket{\phi(\mathbf{x})} + \ket{\psi_1(x, \mathbf{\theta}_1)} \big).
\end{equation}
In general, this is not a unitary operation. However, the desired ResNet skipped connection can indeed be written as a linear combination of unitary operations
\begin{equation}
  \ket{\psi_2(x, \mathbf{\theta}_1, \mathbf{\theta}_2)} = \big( W_2(\mathbf{\theta}_2) + W_2(\mathbf{\theta}_2)W_1(\mathbf{\theta}_1) \big) \ket{\phi(\mathbf{x})}.
\end{equation}
This can be implemented using the LCU method by defining $P_{\text{PREP}}$ equal to the Hadamard gate such that

\begin{equation}
  P_\text{PREP} \ket{0} = \frac{1}{\sqrt{2}}(\ket{0} + \ket{1}),
\end{equation}
and
\begin{equation}
S_\text{SELECT}\ket{0}\ket{\phi(\mathbf{x})} = \ket{0}W_2(\mathbf{\theta}_2)\ket{\phi(\mathbf{x})} ,
\end{equation}
\begin{equation}
S_\text{SELECT}\ket{1}\ket{\phi(\mathbf{x})} = \ket{1}W_2(\mathbf{\theta}_2)W_1(\mathbf{\theta}_1)\ket{\phi(\mathbf{x})}.
\end{equation}
After applying the $P_\text{PREP}$ and $S_\text{SELECT}$ operators and measuring the ancilla qubit to be in the $\ket{0}$ state, then the appropriate $\ket{\psi_2(x, \mathbf{\theta}_1, \mathbf{\theta}_2)}$ can be prepared. We show this explicitly through each step for clarity starting with the preparation operator
\begin{equation}
   P_\text{PREP} \ket{0}\ket{\phi(\mathbf{x})} = \frac{1}{\sqrt{2}}(\ket{0} + \ket{1})\ket{\phi(\mathbf{x})},
\end{equation}
applying the selection operator
\begin{align}
   & S_\text{SELECT} P_\text{PREP} \ket{0}\ket{\phi(\mathbf{x})} \nonumber \\
   & = \frac{1}{\sqrt{2}}\Big( \ket{0} W_2(\mathbf{\theta}_2) \ket{\phi(\mathbf{x})} + \ket{1} W_2(\mathbf{\theta}_2)W_1(\mathbf{\theta}_1) \ket{\phi(\mathbf{x})} \Big),
\end{align}
applying the inverse preparation operator, which is a Hadamard gate we find
\begin{align}
   P_\text{PREP}^\dagger & S_\text{SELECT} P_\text{PREP} \ket{0}\ket{\phi(\mathbf{x})} \nonumber \\
    = \frac{1}{2}\Big(& \ket{0} \big( (W_2(\mathbf{\theta}_2) + W_2(\mathbf{\theta}_2)W_1(\mathbf{\theta}_1) \big) \ket{\phi(\mathbf{x})} \nonumber \\
   + & \ket{1} \big( W_2(\mathbf{\theta}_2) - W_2(\mathbf{\theta}_2)W_1(\mathbf{\theta}_1) \big) \ket{\phi(\mathbf{x})} \Big),
\end{align}
measuring the ancilla and requiring it to be in the $\ket{0}$ state gives
\begin{align}
  & \bra{0}P_\text{PREP}^\dagger S_\text{SELECT} P_\text{PREP} \ket{0}\ket{\phi(\mathbf{x})} \nonumber \\
  & = \frac{1}{\sqrt{\pi_S}}\big( W_2(\mathbf{\theta}_2) + W_2(\mathbf{\theta}_2)W_1(\mathbf{\theta}_1) \big) \ket{\phi(\mathbf{x})},
\end{align}
which is probabilistically prepared with a probability equal to
\begin{align}
  \pi_S & = \frac{1}{4}\bra{\phi(x)}(2\mathbb{I} + W_1(\mathbf{\theta}_1) + W_1^\dagger(\theta_1))\ket{\phi(\mathbf{x})} \nonumber \\
  & = \frac{1}{2}\Big(1 + \text{Re}\big(\bra{\phi(x)} W_1(\mathbf{\theta}_1) \ket{\phi(\mathbf{x})}\big)\Big).
\end{align}
Note that the important part of the probability scaling is the real component of $\text{Re}\big(\bra{\phi(x)} W_1(\mathbf{\theta}_1) \ket{\phi(\mathbf{x})} \big)$. In the case that $W_1(\mathbf{\theta}_1) = \mathbb{I}$ then success is guaranteed, and in the case $W_1(\mathbf{\theta}_1) = -\mathbb{I}$ then success is impossible. For general $W_1(\mathbf{\theta}_1)$ matrices and states $\ket{\phi(\mathbf{x})}$ then the success will vary. It may therefore in practice be worth placing restrictions on $W_1(\mathbf{\theta}_1)$ to ensure that $\text{Re}\big(\bra{\phi(x)} W_1(\mathbf{\theta}_1) \ket{\phi(\mathbf{x})} \big)$ does not approach too close to -1 resulting in extremely low success probabilities. However, as mentioned in the main text, if the strength of the residual connections are adjusted it can be possible to increase the lower bound. This example used equal strength in the residual connections for simplicity, which corresponds to the worst-case lower bound.

\end{appendices}

\clearpage

\bibliographystyle{unsrtnat}

\bibliography{bibliography}% Produces the bibliography via BibTeX.

\end{document}